\PassOptionsToPackage{unicode}{hyperref}
\PassOptionsToPackage{hyphens}{url}
\PassOptionsToPackage{dvipsnames,svgnames,x11names}{xcolor}
\documentclass[12pt]{article}

\usepackage{amsmath,amssymb}
\usepackage{iftex}
\ifPDFTeX
  \usepackage[T1]{fontenc}
  \usepackage[utf8]{inputenc}
  \usepackage{textcomp} 
\else 
  \usepackage{unicode-math}
  \defaultfontfeatures{Scale=MatchLowercase}
  \defaultfontfeatures[\rmfamily]{Ligatures=TeX,Scale=1}
\fi
\usepackage{lmodern}
\ifPDFTeX\else  
\fi
\IfFileExists{upquote.sty}{\usepackage{upquote}}{}
\IfFileExists{microtype.sty}{
  \usepackage[]{microtype}
  \UseMicrotypeSet[protrusion]{basicmath} 
}{}
\makeatletter
\@ifundefined{KOMAClassName}{
  \IfFileExists{parskip.sty}{%
    \usepackage{parskip}
  }{
    \setlength{\parindent}{0pt}
    \setlength{\parskip}{6pt plus 2pt minus 1pt}}
}{
  \KOMAoptions{parskip=half}}
\makeatother
\usepackage{xcolor}
\makeatletter
\ifx\paragraph\undefined\else
  \let\oldparagraph\paragraph
  \renewcommand{\paragraph}{
    \@ifstar
      \xxxParagraphStar
      \xxxParagraphNoStar
  }
  \newcommand{\xxxParagraphStar}[1]{\oldparagraph*{#1}\mbox{}}
  \newcommand{\xxxParagraphNoStar}[1]{\oldparagraph{#1}\mbox{}}
\fi
\ifx\subparagraph\undefined\else
  \let\oldsubparagraph\subparagraph
  \renewcommand{\subparagraph}{
    \@ifstar
      \xxxSubParagraphStar
      \xxxSubParagraphNoStar
  }
  \newcommand{\xxxSubParagraphStar}[1]{\oldsubparagraph*{#1}\mbox{}}
  \newcommand{\xxxSubParagraphNoStar}[1]{\oldsubparagraph{#1}\mbox{}}
\fi
\makeatother

\usepackage{longtable,booktabs,array}
\usepackage{calc} 
\usepackage{etoolbox}
\makeatletter
\patchcmd\longtable{\par}{\if@noskipsec\mbox{}\fi\par}{}{}
\makeatother
\IfFileExists{footnotehyper.sty}{\usepackage{footnotehyper}}{\usepackage{footnote}}
\makesavenoteenv{longtable}
\usepackage{graphicx}
\makeatletter
\def\maxwidth{\ifdim\Gin@nat@width>\linewidth\linewidth\else\Gin@nat@width\fi}
\def\maxheight{\ifdim\Gin@nat@height>\textheight\textheight\else\Gin@nat@height\fi}
\makeatother
\setkeys{Gin}{width=\maxwidth,height=\maxheight,keepaspectratio}
\makeatletter
\def\fps@figure{htbp}
\makeatother

\makeatletter
\@ifpackageloaded{caption}{}{\usepackage{caption}}
\AtBeginDocument{%
\ifdefined\contentsname
  \renewcommand*\contentsname{Table of contents}
\else
  \newcommand\contentsname{Table of contents}
\fi
\ifdefined\listfigurename
  \renewcommand*\listfigurename{List of Figures}
\else
  \newcommand\listfigurename{List of Figures}
\fi
\ifdefined\listtablename
  \renewcommand*\listtablename{List of Tables}
\else
  \newcommand\listtablename{List of Tables}
\fi
\ifdefined\figurename
  \renewcommand*\figurename{Figure}
\else
  \newcommand\figurename{Figure}
\fi
\ifdefined\tablename
  \renewcommand*\tablename{Table}
\else
  \newcommand\tablename{Table}
\fi
}
\@ifpackageloaded{float}{}{\usepackage{float}}
\floatstyle{ruled}
\@ifundefined{c@chapter}{\newfloat{codelisting}{h}{lop}}{\newfloat{codelisting}{h}{lop}[chapter]}
\floatname{codelisting}{Listing}

\makeatother
\makeatletter
\@ifpackageloaded{caption}{}{\usepackage{caption}}
\@ifpackageloaded{subcaption}{}{\usepackage{subcaption}}
\makeatother

\ifLuaTeX
  \usepackage{selnolig}  
\fi
\usepackage[]{natbib}
\usepackage{bookmark}
\usepackage{xr-hyper}
\IfFileExists{xurl.sty}{\usepackage{xurl}}{} 
\hypersetup{
  pdftitle={Title},
  pdfauthor={Author 1; Author 2},
  pdfkeywords={3 to 6 keywords, that do not appear in the title},
  colorlinks=true,
  linkcolor={blue},
  filecolor={black},
  citecolor={Blue},
  urlcolor={Blue},
  pdfcreator={LaTeX via pandoc}}

\newcommand{\anon}{1}

\usepackage{amsthm}
\usepackage{algorithm} 
\usepackage{algorithmicx} 
\usepackage{algpseudocode} 
\usepackage{dsfont}
\usepackage{tikz}
\usetikzlibrary{shapes, arrows, positioning, fit, calc}

\usepackage{listings}
\usepackage{enumerate}
\usepackage{upgreek}
\usepackage{comment}
\usepackage{enumitem}

\newcommand{\reals}{\mathbb{R}}
\newcommand{\nats}{\mathbb{N}}
\newcommand{\natsstar}{\,\mathbb{N}^*}
\newcommand{\Prob}{\mathbb{P}}
\newcommand{\E}{\mathbb{E}}
\newcommand{\Var}{\mathbb{V}\text{ar}}
\newcommand{\Cov}{\mathbb{C}\text{ov}}
\newcommand{\Norm}{\mathcal{N}}

\newcommand{\Bern}{\mathcal{B}}

\newcommand{\Unif}{\mathcal{U}}

\newcommand{\ind}{\mathds{1}}

\newcommand{\textopt}{{\text{opt}}}

\newcommand{\basw}{\underline{w}}
\newcommand{\barw}{\bar{w}}
\newcommand{\esssup}{\underset{\theta \in \reals^d}{\mathrm{ess\,sup}}}

\newcommand{\TV}{\mathrm{TV}}

\DeclareMathOperator*{\argmax}{\arg\!\max}

\renewcommand{\geq}{\geqslant}
\renewcommand{\leq}{\leqslant}

\newtheorem{theorem}{Theorem} 
\newtheorem{lemma}{Lemma}
\newtheorem{proposition}{Proposition}  
\newtheorem{corollary}{Corollary}
 
\newtheorem{remark}{Remark} 
\newtheorem{definition}{Definition} 
\newtheorem{hypothesis}{Assumption}

\begin{document}

\def\spacingset#1{\renewcommand{\baselinestretch}%
{#1}\normalsize} \spacingset{1}


\if1\anon
{
  \title{\bf Adaptive Pseudo-Marginal Algorithm
  \thanks{
      Mylène Bédard and Florian Maire acknowledge the support of \textit{Natural Sciences and Engineering Research Council of Canada (NSERC) through the Discovery Grants program.}
    }}

  \author{
    Sarra Abaoubida
    \thanks{Corresponding author: \texttt{sarra.abaoubida@umontreal.ca}.}
    \thanks{The author acknowledges the support of Fonds de recherche du Québec – Nature et technologies (FRQNT), and Institut des sciences mathématiques (ISM).}
    \\
    Département de mathématiques et statistique, Université de Montréal
    \and
    Mylène Bédard
    \\
    Département de mathématiques et statistique, Université de Montréal
    \and
    Florian Maire 
    \\
    Département de mathématiques et statistique, Université de Montréal
    }

  \maketitle
} \fi

\if0\anon
{
  \bigskip
  \bigskip
  \bigskip
  \begin{center}
    {\LARGE\bf Title}
\end{center}
  \medskip
} \fi

\bigskip
\begin{abstract}
The Pseudo-Marginal (PM) algorithm is a popular Markov chain Monte Carlo (MCMC) method used to sample from a target distribution when its unnormalized density is intractable, but can be estimated with a non-negative unbiased estimator. Its performance depends on a key parameter, \textit{N}, the number of samples or particles used to approximate the target density. Larger values of \textit{N} yield more accurate estimates but at increased running time. Previous studies have provided guidelines for selecting an optimal value of \textit{N} to balance this tradeoff. However, this approach involves tuning the algorithm using multiple steps and manual adjustments. To overcome these limitations, we propose an adaptive version of the PM algorithm, where \textit{N} is automatically adjusted during the iterative process toward its optimal value, thus eliminating the need for manual intervention. We establish the convergence of the resulting chain to the target distribution under suitable conditions. On two examples, including a real data problem on respiratory infection in preschool children, the proposed algorithm compares favorably to the existing manually tuned approach.
\end{abstract}

\noindent%
{\it Keywords:} Adaptive MCMC, Ergodicity, Intractable likelihood.
\vfill

\newpage
\spacingset{1.2} 

\section{Introduction}
\label{sec-intro}
In a Bayesian framework, let $y \in \mathcal{Y}^T$ denote the observed data, with likelihood function $p_T(y | \theta)$, and let $p(\theta)$ be the prior density of the parameter $\theta \in \Theta \subseteq \mathbb{R}^d$. The corresponding posterior density is given by $\pi(\theta) \propto p_T(y | \theta)\, p(\theta)$. The primary Bayesian objective is to compute posterior expectations of the form
$\pi(f) := \E[f(\theta) | y] = \int f(\theta)\,\pi(\mathrm{d}\theta)$,
for measurable functions $f$ satisfying $\pi(f^2)<\infty$. An approximation of this posterior expectation is easily accessible from a MCMC sampler, via the sample average $\widehat{\pi}_L(f) = L^{-1}\sum_{\ell=0}^{L-1} f(\theta_\ell)$ obtained from the Markov process $\{\theta_\ell\}_{\ell \geq 0}$. 

A widely used MCMC method is the Metropolis-Hastings (MH) algorithm, which generates a Markov chain with stationary distribution $\pi$. At each step, given the current state $\theta$, a candidate $\vartheta$ is proposed, and then accepted with a probability that depends on the likelihood ratio $p_T(y | \vartheta)/p_T(y | \theta)$. However, in many statistical models, direct computation of the likelihood \( p_T(y|\cdot) \) is challenging, rendering the standard MH impractical.

The PM algorithm, first proposed in \cite{beaumont_estimation_2003} and later formalized in \cite{andrieu_pseudo-marginal_2009}, offers a practical solution to this problem. It substitutes the intractable likelihood $p_T(y |\cdot)$ in the acceptance probability with an unbiased, non-negative estimator. A detailed description of the algorithm is provided in Section~\ref{sec-pm}. Standard approaches like importance sampling and particle filtering (for state-space models) generate such estimators by averaging over $N$ realizations. Substituting the true likelihood with an estimate in the MH introduces a trade-off: as \( N \) increases, the asymptotic accuracy of a PM chain improves (see \cite{andrieu_establishing_2016}), but at a higher computational cost (runtime). Therefore, a key challenge is balancing computational efficiency and statistical accuracy. The weak convergence properties of the PM algorithm, as the number of observations $T$ increases, were examined in \cite{schmon_large-sample_2021}, providing a method for selecting the optimal number of samples/particles \( N \) to use in the PM, based on the dimension \( d \) of the parameter $\theta$. However, this approach requires manual tuning and sequential execution of multiple steps. We summarize this non-adaptive method in Section~\ref{sec-review}.

In Section~\ref{sec-APM}, we propose an efficient alternative by letting $N$ evolve over the iterations of the PM algorithm, allowing it to approach its optimal value automatically. This scheme is inspired by the steps outlined in \cite{schmon_large-sample_2021} but integrates them into a single process, leading to the Adaptive Pseudo-Marginal (APM) algorithm. By doing so, we naturally lose the Markov property of the process because the adjustment of $N$ depends on the entire history of the chain. It is therefore important to establish conditions under which our adapted sampler is ergodic, that is, converges to the target distribution. As specified in \cite{roberts_coupling_2007}, two conditions ensure the convergence of an adaptive MCMC algorithm to the target distribution: the \textit{diminishing adaptation} and the \textit{containment condition}.

The proposed APM algorithm satisfies the \textit{diminishing adaptation} condition by construction, since the probability of adapting $N$ converges to $0$ as the algorithm progresses. Establishing the \textit{containment condition}, on the other hand, typically requires proving \emph{simultaneous geometric drift and minorization conditions for the family of Markov kernels indexed by the adapting parameter} (see \cite{roberts_coupling_2007}). However, this is challenging in the PM context, since, as shown in \cite{andrieu_convergence_2015}, it generally exhibits only \textit{polynomial drift}. To address this difficulty, we build upon the analysis of \cite{atchade_limit_2010}, which studies the ergodicity of adaptive MCMC using polynomially ergodic kernels and on the convex ordering of the likelihood estimations introduced in \cite{andrieu_establishing_2016}. Our main theoretical result, presented in Section~\ref{sec-APMproof}, is the establishment of ergodicity for the APM. First, we generalize the ergodicity conditions of \cite{atchade_limit_2010} to our setting, in which the target distribution explicitly depends on the number of particles \( N \). Second, we introduce a set of verifiable assumptions under which these generalized conditions hold, and we prove that, under these assumptions, the marginal distribution of the APM process converges in total variation distance to the posterior distribution.

In Sections~\ref{sec-syntheticex} and~\ref{sec-realex}, we compare the performance of the APM algorithm with two versions of the non-adaptive approach of \cite{schmon_large-sample_2021}. The first corresponds to the original implementation of the non-adaptive method, while the second is an automated implementation designed to provide a conservative comparison with the APM algorithm by reducing the manual tuning effort required by the non-adaptive approach.

In Section~\ref{sec-syntheticex}, this comparison is conducted on a synthetic example with simulated data. The simulations show that the running time of the proposed adaptive method is substantially lower than that of the competing approaches: the APM is approximately $22\%$ faster than the automated non-adaptive method and approximately $92\%$ faster than the non-adaptive implementation as described  in \cite{schmon_large-sample_2021}. In this section, we also prove that the model used in this example satisfies the theoretical conditions established in Section~\ref{sec-APMproof} for the ergodicity of the APM algorithm, thereby providing an explicit example illustrating how the ergodicity assumptions introduced in Section~\ref{sec-APMproof} can be verified in practice.

In Section~\ref{sec-realex}, we evaluate the performance of the APM on real data from \cite{zeger_generalized_1991} and again compare it with the two versions of the non-adaptive approach. The running time of the APM is similar to that of the automated non-adaptive method, with the APM being marginally faster than its counterpart. In contrast, the naive implementation requires substantially more running time, making the APM approximately $39\%$ faster.

In both examples, all methods produce identical posterior means and variances empirically, supporting the validity of the APM even in a real data setting. It is also important to emphasize that the main computational contribution of the APM is the removal of the manual tuning required for the parameter $N$, which is instead adapted automatically during the execution of the algorithm. The reduction in running time should therefore be viewed as an additional, though non negligible, practical benefit of the proposed approach.

\section{The Pseudo-Marginal Algorithm}
\label{sec-pm}

Throughout this section we introduce two equivalent representations of the PM algorithm, both of which will be used in the remainder of the paper. Before introducing the PM algorithm, we briefly recall the standard MH algorithm, which generates a Markov chain with stationary distribution $\pi$. Its transition kernel is
{
\begin{align*}
P(\theta,\mathrm{d}\vartheta)
&=
q(\vartheta|\theta)\{1\wedge r(\theta,\vartheta)\}\,\mathrm{d}\vartheta +
\left(1-\int_\Theta q(\vartheta|\theta)\{1\wedge r(\theta,\vartheta)\}\mathrm{d}\vartheta\right)\delta_\theta(\mathrm{d}\vartheta),
\end{align*}
}where $a\wedge b=\min(a,b)$, $\delta_\theta$ is the Dirac measure at $\theta$, $\vartheta \mapsto q(\vartheta|\theta)$ denotes the proposal density, and
\[
r(\theta,\vartheta)
=
\frac{
p_T(y|\vartheta)p(\vartheta)q(\theta|\vartheta)
}{
p_T(y|\theta)p(\theta)q(\vartheta|\theta)
}
\]
is the MH acceptance ratio.

As discussed in Section~\ref{sec-intro}, when the likelihood $p_T(y|\theta)$ is intractable, the MH algorithm becomes infeasible. The PM algorithm overcomes this difficulty by replacing the exact likelihood with a positive unbiased estimator
\(
\widehat p_{N,T}(y|\theta,U),
\)
where $U|\theta\sim m_{N,\theta}(\cdot)$ denotes the auxiliary random variables used to compute the likelihood estimator (see \cite{beaumont_estimation_2003,andrieu_pseudo-marginal_2009}).

We first present the PM algorithm through its formulation as an MH algorithm targeting the extended distribution
\begin{align}
\bar\pi_N(\theta,u)
=
\pi(\theta)
m_{N,\theta}(u)
\frac{\widehat p_{N,T}(y|\theta,u)}
     {p_T(y|\theta)},
\label{pmpost}
\end{align}
with proposal density
\(
(\vartheta,v)
\mapsto
q(\vartheta|\theta)
m_{N,\vartheta}(v).
\)
Its acceptance probability is
{
\begin{align}
\alpha_N(\theta,u;\vartheta,v)
=
1\wedge
\frac{
\widehat p_{N,T}(y|\vartheta,v)
p(\vartheta)
q(\theta|\vartheta)
}{
\widehat p_{N,T}(y|\theta,u)
p(\theta)
q(\vartheta|\theta)
},
\label{alphaN}
\end{align}
}which is computable because it only involves the likelihood estimator.

Since $\widehat{p}_{N,T}(y |\theta,U)$ is an unbiased estimator of $p_T(y |\theta)$, integrating $\bar{\pi}_N(\theta,u)$ with respect to $u$ yields $\pi(\theta)$. Consequently, although the PM algorithm replaces the exact likelihood with an unbiased estimator, it remains exact in the sense that its marginal stationary distribution is the desired posterior distribution $\pi$. The corresponding transition kernel of the PM is
{
\begin{align}
P_N(\theta, u; \mathrm{d}\vartheta, \mathrm{d}v)
&= q(\vartheta \lvert \theta)\,
   m_{N,\vartheta}(v)\,
   \alpha_N(\theta, u; \vartheta, v)\,
   \mathrm{d}\vartheta\, \mathrm{d}v + \varrho_{N}(\theta, u)\,
   \delta_{(\theta, u)}(\mathrm{d}\vartheta, \mathrm{d}v),
\label{equa4}
\end{align}}
\noindent
where
{
\begin{align}
\varrho_{N}(\theta, u)
= 1 - \int_{\Theta \times \mathcal{U}} & q(\vartheta \lvert \theta)\, m_{N,\vartheta}(v) \alpha_N(\theta, u; \vartheta, v)\, \mathrm{d}\vartheta\, \mathrm{d}v.
\label{varrho}
\end{align}}The auxiliary variable representation is the natural formulation of the PM algorithm and is the one used in practical implementations. However, for the review of the non-adaptive method in Section~\ref{sec-review} and for several theoretical developments in Section~\ref{sec-APMproof}, we follow \cite{andrieu_pseudo-marginal_2009}, \cite{pitt_properties_2012}, \cite{andrieu_convergence_2015}, \cite{doucet_efficient_2015}, \cite{sherlock_efficiency_2015}, \cite{andrieu_establishing_2016}, and \cite{schmon_large-sample_2021} by reparameterizing the PM algorithm in terms of the normalized likelihood estimator
\[
W_{N,\theta}
=\frac{\widehat{p}_{N,T}(y|\theta,U)}{p_T(y|\theta)}.
\]
Since,
\(
\widehat{p}_{N,T}(y|\theta,U)
=
p_T(y|\theta)\,W_{N,\theta},
\)
the normalized likelihood estimator $W_{N,\theta}$ represents the multiplicative error, or noise, introduced by replacing the intractable likelihood with its unbiased estimator. We now introduce the mathematical formulation of the PM algorithm in terms of the multiplicative noise $W_{N,\theta}$. 

Let $\{\mathcal{Q}_{N,\theta}\}_{\theta\in\Theta}$ denote the family of probability distributions of $W_{N,\theta}$ on $\left(\mathbb{R}_{+}^{*},\mathcal{B}(\mathbb{R}_{+}^{*})\right)$, satisfying
\[
\mathbb{E}[W_{N,\theta}]
=
\int w\,\mathcal{Q}_{N,\theta}(\mathrm{d}w)
=
1,
\qquad
\theta\in\Theta.
\]
Under this reparameterization, the PM algorithm can equivalently be viewed as a MH algorithm targeting
\[
\bar{\pi}_N(\theta,\mathrm{d}w)
=
\pi(\theta)\, w \,
\mathcal{Q}_{N,\theta}(\mathrm{d}w).
\]
For $(\theta,w)\in\Theta\times\mathcal W$, with
$\mathcal W=\mathbb R_+^*$, the corresponding transition kernel is
{
\begin{align}
{P_N}(\theta, w ; \mathrm{d} \vartheta,\mathrm{d} z) &= q(\vartheta|\theta) \left\{1\wedge r(\theta, \vartheta) \frac{z}{w}\right\} \smash{\mathcal{Q}_{N,\vartheta}}(\mathrm{d} z) \mathrm{d}\vartheta + {\varrho_N}(\theta, w)\delta_{\theta, w}(\mathrm{d} \vartheta ,\mathrm{d} z) ,
\label{equa5}
\end{align}}where the rejection probability is
{
\begin{align}
\varrho_{N}(\theta, w) = 1- \int_{\Theta \times \mathcal{W}} 
& q(\vartheta|\theta) \left\{1 \wedge r(\theta, \vartheta) \frac{z}{w}\right\}\mathcal{Q}_{N,\vartheta}(\mathrm{d}z) \,\mathrm{d}\vartheta.
\label{varrho_repar}
\end{align}}Although this representation is central to the theoretical analysis, the quantity $W_{N,\theta}$ is never explicitly computed in practice since it depends on the true likelihood. Instead, the PM algorithm operates directly with the unbiased likelihood estimators through the acceptance probability \eqref{alphaN}, without evaluating the intractable likelihood.

To facilitate the theoretical analysis, we now introduce an augmented unbiased likelihood estimator obtained by multiplying the original estimator by an additional independent uniform random variable. This augmented estimator will be used in the formulation of the APM algorithm presented in Section~\ref{sec-APM}.  From a practical perspective, however, this additional randomization is not required to implement the algorithm. Its purpose is solely theoretical. Specifically, the augmentation ensures that the normalized likelihood estimator admits a density with respect to the Lebesgue measure that can be uniformly minorized, which is the key ingredient in establishing the minorization condition in the ergodicity proof of Section~\ref{sec-APMproof}. Let
\[
\xi \sim \Unif(1-\delta,1+\delta), \qquad 0<\delta<1,
\]
independently of the auxiliary variable $U$. We define the augmented likelihood estimator by
\begin{align}
    \widehat{p}_{N,T}(y | \theta,U,\xi)
    =\widehat{p}_{N,T}(y | \theta,U)\,\xi.
    \label{augmentedlikeestm}
\end{align}
Since $\mathbb{E}[\xi]=1$ and $\xi$ is independent of $U$, the estimator remains unbiased,
\[
\mathbb{E}\!\left[
\widehat{p}_{N,T}(y | \theta,U,\xi)
\right]
=
p_T(y | \theta).
\]
Together with unbiasedness, this perturbation preserves positivity. Thus, the resulting algorithm remains within the PM framework. The corresponding normalized likelihood estimator is
\[
\widetilde{W}_{N,\theta}
=
\frac{\widehat{p}_{N,T}(y | \theta,U,\xi)}
     {p_T(y | \theta)}
=
W_{N,\theta}\xi.
\]

The distribution of $W_{N,\theta}$ is not assumed to admit a density with respect to the Lebesgue measure. By contrast, the independent uniform perturbation smooths the distribution while preserving positivity and unbiasedness. A straightforward conditioning argument yields that $\widetilde{W}_{N,\theta}$ admits a density with respect to Lebesgue measure, given by
\[
\widetilde{\mathcal{Q}}_{N,\theta}(z)
=
\int
\frac{1}{2\delta w}
\mathbf{1}_{\{z\in[w(1-\delta),\,w(1+\delta)]\}}
\,
\mathcal{Q}_{N,\theta}(\mathrm{d}w).
\]
Replacing $W_{N,\theta}$ by $\widetilde{W}_{N,\theta}$ in the reparameterized PM formulation gives rise to the following augmented PM kernel defined on $\Theta \times \mathcal{W}$,
{
\begin{align}
\tilde{P}_N(\theta,w;\mathrm{d}\vartheta,\mathrm{d}z)
&=
q(\vartheta|\theta)
\left\{
1\wedge
r(\theta,\vartheta)\frac{z}{w}
\right\}\widetilde{\mathcal{Q}}_{N,\vartheta}(z)
\,\mathrm{d}z\mathrm{d}\vartheta +\varrho_N(\theta,w)
\delta_{(\theta,w)}
(\mathrm{d}\vartheta,\mathrm{d}z),
\label{eq:perturbedPM}
\end{align}}where $\varrho_N(\theta,w)$ is the corresponding rejection probability. This augmented PM kernel targets the extended distribution
\[
\widetilde{\pi}_N(\theta,w)
=
\pi(\theta)\,
w\,
\widetilde{\mathcal{Q}}_{N,\theta}(w)\,
\mathrm{d}w.
\]

Having formally introduced two equivalent representations of the PM algorithm and an augmented representation of the PM, we now review the critical role of the Monte Carlo parameter (or number of particles) $N$ in determining its performance. Since the performance of MCMC algorithms is commonly assessed through their asymptotic variance, we begin by recalling this notion before reviewing the influence of $N$. 

Consider a function \(f \in L^2(\pi)\) and let \(\tilde f = f - \pi(f)\). For a Markov chain \(\{\theta_\ell\}_{\ell \geq 0}\) with transition kernel \(Q\) and stationary distribution \(\pi\), the asymptotic variance of the estimator \(\widehat{\pi}_L(f) = L^{-1}\sum_{\ell=0}^{L-1} f(\theta_\ell)\) exists, is finite under mild assumptions (see \cite{haggstrom_variance_2007}), and satisfies
\begin{align}
\overline{\Var}(f,Q) &= \lim_{L \to \infty} \frac{1}{L} \Var\Bigg(\sum_{\ell=0}^{L-1} \tilde f(\theta_\ell)\Bigg) = \pi(\tilde f^2) \operatorname{IF}(f,Q).
\label{asympvar2}
\end{align}
The inefficiency factor 
\begin{equation}
\operatorname{IF}(f,Q) = 1 + 2 \sum_{\ell=1}^\infty \frac{\Cov(f(\theta_0),f(\theta_\ell))}{\pi(\tilde f^2)} < \infty
\label{inefficiency}
\end{equation}
is usually larger than one and quantifies the penalty (or excess of variance) of the Monte Carlo estimator induced by the autocorrelation in the Markov chain. 

In the PM context, \cite{andrieu_establishing_2016} shows that, for likelihood estimators obtained via importance sampling, the asymptotic variance of the PM estimator (equivalently, its inefficiency factor) decreases as $N$ increases. However, increasing $N$ also raises the computational cost of evaluating the likelihood estimator, and therefore the cost of each PM iteration. Consequently, the performance of the PM algorithm is governed by a fundamental trade-off between efficiency and runtime.

To choose the number of particles \(N\) that balances efficiency and computational cost, \cite{schmon_large-sample_2021} proposed a non-adaptive tuning procedure, which we review in Section~\ref{sec-review}. Their approach relies on a weak convergence result showing that, as the dataset size tends to infinity, the PM chain converges to a universal limiting Markov chain parameterized solely by the variance of the log-likelihood noise
\(
\log{W_{N,\theta}}.
\)
This quantity is commonly referred to as the \emph{additive noise}. The computing time measure, a function of both runtime and efficiency, is then defined for this limiting chain and minimized with respect to the variance of the additive noise, yielding an optimal target value \(\sigma_{\textopt}\). Finally, the relationship between the variance of the additive noise and the number of particles is exploited to determine the corresponding optimal number of particles \(N_{\textopt}\). This leads to a practical non-adaptive procedure for implementing the PM with the value \(N_{\textopt}\) that minimizes the computing time, achieving a suitable balance between efficiency and runtime.

\section{Computing Time Optimization in Pseudo-Marginal Algorithms}
\label{sec-review}
The authors in \cite{schmon_large-sample_2021} established a weak convergence result for the PM algorithm in the large-sample regime (\(T \to \infty\)). Their analysis assumes a Bernstein-von-Mises theorem for posterior concentration around a limiting parameter \(\bar{\theta}\), together with a uniform central limit theorem for the additive noise. Under these assumptions, they proved that, after an appropriate rescaling of the parameter, the PM chain converges weakly to a limiting PM chain targeting a Gaussian distribution. The transition kernel of the limiting chain is given by
{\small
\begin{align}
   {P_\sigma}\left(\theta,\omega ; \mathrm{d} \vartheta,\mathrm{d}\zeta\right) &= q(\vartheta|\theta)\, \varphi(\zeta ; -\sigma^2/2, \sigma^2)\left\{1\wedge r(\theta, \vartheta) \exp\{\zeta-\omega\}\right\} \mathrm{d}\vartheta \mathrm{d}\zeta +  {\rho_\sigma}(\theta,\omega)\, \delta_{\theta, \omega}(\mathrm{d} \vartheta ,\mathrm{d}\zeta),
   \label{limitkernel}
\end{align}}where ${\rho_\sigma}(\theta,\omega)$ denotes the rejection probability, $\varphi(\cdot; \mu, \sigma^2)$ the density of a $\Norm(\mu,\sigma^2)$ and the proposed additive noise  $\zeta$ is distributed as 
\[\zeta \sim \mathcal{N}(-\sigma^2/2, \sigma^2)\]
independently of the current state and with variance $\sigma^2$ constant across $\theta$. We refer the reader to \cite[Theorem 1]{schmon_large-sample_2021} for a precise definition of this weak convergence result.

The limiting transition kernel $P_\sigma$ in~\eqref{limitkernel} depends only on the proposal distribution and is independent of the likelihood, the prior, and the distribution of the auxiliary variables; it can thus be qualified as universal. Using the Normal random walk proposal
\[
q(\vartheta|\theta) = \varphi\left(\vartheta; \theta, \frac{l^2}{d} I_d\right),
\]
where the scaling $l^2/d$ follows the framework of~\cite{sherlock_efficiency_2015}, the authors in \cite{schmon_large-sample_2021} determined, for each dimension $d$, the optimal values $(l_{\mathrm{opt}},\sigma_{\mathrm{opt}})$ minimizing the computing time $\mathrm{CT}(\bar f,P_\sigma)$, defined for functions $\bar f\in L^2(\bar\pi_N)$ as
\begin{equation}
    \text{CT}(\bar f,P_\sigma) :=  \frac{\operatorname{IF}(\bar f,P_\sigma)}{\sigma^2} \propto N \times \operatorname{IF}(\bar f,P_\sigma).
    \label{CT}
\end{equation}
\begin{remark}
    In the large-sample regime ($T \to \infty$), \cite{berard_lognormal_2014} and \cite{deligiannidis_correlated_2018} showed, under standard assumptions, that the variance of the additive noise, $\sigma^2$, is proportional to $N^{-1}$, where $N$ denotes the number of particles used to construct the unbiased likelihood estimator. Consequently, the computing time $\mathrm{CT}(\bar f,P_\sigma)$ reflects the trade-off between the runtime per iteration measured by $N \propto 1/\sigma^2$ and the inefficiency factor $\operatorname{IF}(\bar f,P_\sigma)$.
    \label{remark1}
\end{remark}

Since $\mathrm{CT}(\bar f, P_\sigma)$ has no closed-form expression and depends only on $(l, \sigma)$, the optimal values of these parameters were determined through simulation. For each dimension $d$ and considering the function $\bar f(\theta, \omega) = \theta_1$, where $\theta_1$ denotes the first coordinate of $\theta$, multiple chains were run over a finely spaced grid of candidate $(l, \sigma)$ values. The corresponding computing times were computed for each pair, and the pair $(l_{\text{opt}}, \sigma_{\text{opt}})$ minimizing the computing time was selected. The resulting optimal tuning parameters were reported for various values of $d$ \cite[Table~1]{schmon_large-sample_2021}.

Following these observations, \cite{schmon_large-sample_2021} proposed a procedure to implement the PM algorithm with the optimal number of particles, $N_{\text{opt}}$. After using Table 1 in \cite{schmon_large-sample_2021} to identify the estimated optimal scaling parameter $l_\text{opt}$ and additive noise standard deviation $\sigma_\text{opt}$ as a function of the parameter dimension $d$, this method can be summarized as follows:
\begin{enumerate}[
    label=\textbf{Step \arabic*},
    ref=Step \arabic*,
    leftmargin=2.0cm,
    labelindent=0cm
]
    \item
    \label{step1}
    Run a preliminary PM algorithm with some initial $N_1$ to obtain $\hat \theta_{N_1}$ and $\widehat \Sigma_{N_1}$, the posterior mean and covariance estimates of $\theta$. For this preliminary PM run, the proposal density is a Normal random walk with covariance ${l_{\text{opt}}^2 \Sigma_p}/{d}$, where $\Sigma_p$ is some positive-definite matrix.

    \item
    \label{step2}
    Let $\sigma_N^2(\hat \theta_{N_1})$ denote the variance of the additive noise $\log{W_{N,\hat \theta_{N_1}}}$, defined as:
    \begin{align}
        \sigma_N^2(\hat \theta_{N_1}) &= \Var(\log{W_{N,\hat \theta_{N_1}}}|\hat \theta_{N_1})
        = \Var\big(\log{\widehat p_N(y |\hat \theta_{N_1}, \widehat U)}|\hat \theta_{N_1}\big),
        \label{estimationsd}
    \end{align}

    \noindent
    where the auxiliary variables $\widehat U|\hat\theta_{N_1} \sim m_{N, \hat \theta_{N_1}}$. 
    The variance $\sigma_N^2(\hat{\theta}_{N_1})$ is estimated by Monte Carlo for several values of $N$. The value $N_{\mathrm{opt}}$ is then chosen as the one for which the estimated variance $\hat{\sigma}_N^2(\hat{\theta}_{N_1})$ is closest to the target value $\sigma_{\mathrm{opt}}^2$.
   
    \item
    \label{step3}
    Execute the PM algorithm using $N_{\text{opt}}$ and a Normal random walk proposal density with covariance ${l_{\text{opt}}^2\widehat \Sigma_{N_1}}/{d}$.
\end{enumerate}
\label{NAmethod}

\begin{remark}
        Choosing an excessively large value for \( N_1 \) may lead to unnecessary computational burden, while a value too small may yield inaccurate estimates of \( \hat\theta_{N_1} \) and \( \hat\Sigma_{N_1} \), thereby degrading the performance of subsequent steps.
\end{remark}
\begin{remark}
    \ref{step2} is typically carried out via manual tuning of the number of particles \( N \), by iteratively adjusting \( N \) and visually inspecting the resulting variability of the additive noise. However, such trial-and-error procedures are inherently subjective, difficult to automate, and complicate reproducibility. To address this, we first implement a naive version of \ref{step2} by testing a wide range of \(N\) values and selecting the one that yields the target \(\sigma_\textopt\) based on Monte Carlo simulations. We then propose an automated procedure using a \emph{dichotomic search} algorithm, which iteratively narrows the search interval until a predefined precision \(a_1\) is reached, achieving the same target \(\sigma_\textopt\) with fewer simulations (see Section~\ref{dichotomic} for an illustration). This automated approach allows an objective, yet flattering, comparison with the APM algorithm introduced in Section~\ref{sec-APM}. 
\end{remark}

The above implementation separates the tuning and sampling stages by requiring a preliminary pilot run to estimate the posterior covariance and determine an appropriate number of particles. We instead propose an adaptive implementation in which these tasks are combined into a single process, allowing the number of particles $N$ to progressively approach its optimal value as the PM iterations progress. The performances of the adaptive and non-adaptive approaches are compared on synthetic and real data in Sections~\ref{sec-syntheticex} and~\ref{sec-realex}.

\section{The Adaptive Pseudo-Marginal Algorithm}
\label{sec-APM}
In many adaptive MCMC methods, tuning parameters (e.g., proposal variance) are updated using \textit{epoch}-based strategies to ensure computational stability and to avoid the additional runtime incurred by adapting at every iteration (see \cite{roberts_examples_2009}). An epoch consists of $K$ consecutive MCMC iterations, \( \{\theta_{\ell-K+1}, \dots, \theta_\ell\} \), during which no adaptation occurs. After each epoch, empirical estimates (e.g., mean, covariance, or acceptance rate) are computed from the samples in that epoch. These estimates serve as inputs to an \textit{adaptation criterion}, such as whether the observed acceptance rate is sufficiently close to its target value (e.g., 0.234 for the Random Walk Metropolis (RWM), \cite{gelman_weak_1997}). If the criterion suggests suboptimal performance, the algorithm adjusts its parameters before proceeding to the next epoch. Restricting adaptation to the end of each epoch reduces the variability of the empirical estimates.

In this section, we introduce our adaptive version of the PM algorithm, called the APM algorithm, in which the number of particles \(N\) used to estimate the intractable likelihood is updated according to an epoch-based adaptation strategy. As discussed in Section~\ref{sec-review}, \cite{schmon_large-sample_2021} provides, in the large-sample regime, dimension-dependent optimal values for the additive noise standard deviation, denoted by \(\sigma_{\mathrm{opt}}\). As noted in Remark~\ref{remark1}, the standard deviation decreases as \(N\) increases. Consequently, at the end of each epoch, the APM estimates the current additive noise standard deviation, compares it with the target value \(\sigma_{\mathrm{opt}}\), and updates \(N\) by adding or subtracting a step size \(a \in \mathbb{N}^*\). This adaptation progressively steers the standard deviation toward its optimal value.

A key ingredient of this adaptation is a reliable estimator of the additive noise variance. A natural approach would be to mimic \ref{step2} of the non-adaptive implementation by generating, after each epoch, a fresh collection of auxiliary variables and estimating the additive noise variance from the resulting log-likelihood estimates. However, this would require additional Monte Carlo likelihood evaluations after every epoch, substantially increasing the computational cost of the algorithm while defeating the purpose of leveraging adaptive MCMC techniques. Instead, the APM reuses the auxiliary variables already generated during the current epoch, thereby avoiding any additional likelihood evaluations. The main challenge is therefore to construct a reliable estimator of the additive noise variance using only the quantities generated by the PM chain.

The remainder of this section is organized as follows. We first explain how the variance estimator is constructed from the history of the PM chain. We then use this estimator to define the adaptation rule for $N$ and present a rigorous mathematical formulation of the APM algorithm. Finally, we state the complete APM algorithm in Algorithm~\ref{algo:apm} and conclude with a discussion of the role and practical importance of its tuning parameters.

The APM is formulated in terms of the augmented PM algorithm introduced in Section~\ref{sec-pm}. This augmented representation is required for the theoretical analysis developed in Section~\ref{sec-APMproof}. 

We now describe the construction of an estimator of the additive noise variance based on the log-likelihood estimates generated during the current epoch. As discussed in Section~\ref{sec-review}, \cite{schmon_large-sample_2021} shows that, under a Bernstein-von-Mises posterior concentration around a limiting parameter \(\bar{\theta}\), together with a uniform central limit theorem for the additive noise, the computing time of the limiting PM kernel is minimized when the standard deviation of the additive noise equals the dimension-dependent target value \(\sigma_{\textopt}\). The limiting variance of the additive noise is
\begin{align}
\sigma^2
&=
\lim_{T\to\infty}
\Var\!\left(\log W_{N,\bar{\theta}}|\theta=\bar{\theta}\right)=\lim_{T\to\infty}
\Var\!\left(\log\widehat{p}_{N,T}(y|\bar{\theta},\bar{V})|\theta=\bar{\theta}\right),
\label{variance}
\end{align}
where \(\bar{V}|\bar{\theta}\sim m_{N,\bar{\theta}}\). Since the intractable likelihood \(p_T(y|\bar{\theta})\) is fixed, the additive noise differs from the log-likelihood estimate only by the constant \(-\log p_T(y|\bar{\theta})\). Consequently, the limiting variance in \eqref{variance} can be estimated from the log-likelihood estimates generated by the PM algorithm, without requiring access to the intractable likelihood.

We deliberately recycle the proposed auxiliary variables $V_\ell|\vartheta_\ell\sim m_{N_{\ell},\vartheta_\ell}$ rather than the accepted ones $U_\ell|\theta_\ell\sim m_{N_{\ell},\theta_\ell}$ because they preserve the desired limiting variance (Theorem~1 of \cite{schmon_large-sample_2021}) while avoiding the sticking behaviour exhibited by the accepted auxiliary variables.

To make this recycling strategy possible, two additional ingredients are required. First, the recycled auxiliary variables are distributed according to $m_{N,\vartheta_\ell}$ rather than the desired distribution $m_{N,\bar{\theta}}$. Second, the limiting parameter $\bar{\theta}$ is unknown. We therefore approximate $\bar{\theta}$ at the end of each epoch using the empirical mean of all parameters up to iteration $\ell$, namely
\begin{equation}
\hat{\theta}_\ell=
\begin{cases}
\hat{\theta}_{\ell-1}, & \ell\notin K\natsstar,\\[0.3cm]
\displaystyle\frac{1}{\ell}\sum_{i=1}^{\ell}\theta_i,
& \ell\in K\natsstar,
\end{cases}
\label{samplemean}
\end{equation}
where $K\natsstar=\{Kj:j\in\natsstar\}$ denotes the set of adaptation iterations.

Using this estimate, we assume the existence of a transformation $h$ satisfying
\[
\widehat V_\ell
=
h(V_\ell,\vartheta_\ell,\hat\theta_{\ell})
|
\hat\theta_{\ell}
\sim
m_{N_{\ell},\hat\theta_{\ell}},
\]
whenever \(V_\ell|\vartheta_\ell\sim m_{N_{\ell},\vartheta_\ell}.\)
Under this assumption, the recycled auxiliary variables are transformed so that they have the same distribution as if they had been generated directly under the estimated limiting parameter. 

\begin{remark}
In many PM implementations, the auxiliary variables used for likelihood estimation are generated from Gaussian proposal distributions. In such settings, the transformation \(h\) reduces to a simple location shift. For example, under classical importance sampling,
\[V_n | \vartheta
\stackrel{\mathrm{iid}}{\sim}
\mathcal{N}(\vartheta,1),
\qquad n=1,\ldots,N,\]
defining
\(
\widehat{V}_n
=
h(V_n,\vartheta,\hat{\theta})
=
V_n-\vartheta+\hat{\theta}
\)
yields
\(
\widehat{V}_n | \hat{\theta}
\sim
\mathcal{N}(\hat{\theta},1).
\)
The transformation used in Sections~\ref{sec-syntheticex} and~\ref{sec-realex} is different but remains linear and serves the same purpose.
\end{remark}

At the end of each epoch, the variance of the additive noise is estimated by the empirical variance of the transformed log-likelihood estimates computed over the \(K\) iterations of the current epoch. Specifically, for epoch \(j\), corresponding to iteration \(\ell=jK\),
{
\begin{equation}
\begin{split}
\hat\sigma_{\lfloor\frac{\ell}{K}\rfloor}^{\,2}
&=
\frac{1}{K-1}\sum_{i=\ell-K+1}^{\ell}\Bigg(\log\left\{\widehat p_{N_{i},T}(y|\hat\theta_{i},\widehat V_i)\right\}-\frac1K\sum_{i=\ell-K+1}^{\ell}\log\left\{\widehat p_{N_{i},T}(y|\hat\theta_{i},\widehat V_i)\right\}\Bigg)^2,
\end{split}
\label{hatsigma}
\end{equation}
}
which is subsequently used by the adaptation rule described next.

The number of particles is updated only at the end of each epoch. Specifically, the estimate \(\hat{\sigma}_{\lfloor\ell/K\rfloor}\) of the additive noise standard deviation is compared with the target value \(\sigma_{\textopt} \pm \sigma_\text{e}\), where $\sigma_\text{e}$ is a tolerance parameter. If \(\hat{\sigma}_{\lfloor\ell/K\rfloor}>\sigma_{\textopt}+\sigma_{\mathrm e}\), the number of particles is increased by a fixed step size \(a\); if \(\hat{\sigma}_{\lfloor\ell/K\rfloor}<\sigma_{\textopt}-\sigma_{\mathrm e}\), it is decreased by the same amount. Otherwise, no adaptation is performed. To satisfy the \textit{diminishing adaptation} condition, this update is carried out only with a probability that converges to zero.

We now formulate the APM algorithm as an adaptive MCMC. Let \(N_\ell\) be an \(\mathbb N^*\)-valued random variable denoting the number of particles used by the PM transition kernel at iteration \(\ell\). A state of the algorithm at iteration \(\ell\) is the \(\Theta\times\mathcal W\)-valued random variable \((\theta_\ell,\widetilde W_\ell)\). The filtration
\[
\mathcal{G}_\ell=\sigma(\theta_0,\widetilde W_0,N_0,\ldots,\theta_\ell,\widetilde W_\ell,N_\ell)
\]
contains the complete history of the algorithm up to iteration \(\ell\).

Let \(A\in\mathcal{B}(\Theta\times\mathcal{W})\). The transition dynamics satisfy
\begin{align}
\mathbb{P}\!\left[(\theta_{\ell+1},\widetilde W_{\ell+1})\in A\,|\,\mathcal{G}_\ell\right]
=\widetilde P_{N_\ell}(\theta_\ell,\widetilde W_\ell;A)
=\widetilde P_{\psi(N_{\ell-1},\hat{\sigma}_{\lfloor\frac{\ell}{K}\rfloor})}
(\theta_\ell,\widetilde W_\ell;A),
\label{adaptivetransition}
\end{align}
where \(\widetilde P_{N_\ell}\) is the PM transition kernel defined in Equation~\eqref{eq:perturbedPM}, and the adaptation function is defined by
{
\begin{align*}
\psi(N_{\ell-1},\hat{\sigma}_{\lfloor\frac{\ell}{K}\rfloor})
&=
N_{\ell-1}
+a\,B_\ell
\left\{
\ind(\hat{\sigma}_{\lfloor\frac{\ell}{K}\rfloor}>\sigma_{\textopt}+\sigma_{\mathrm e})-
\ind(\hat{\sigma}_{\lfloor\frac{\ell}{K}\rfloor}<\sigma_{\textopt}-\sigma_{\mathrm e})
\right\}
\ind(\ell\in K\nats^*).
\end{align*}}Here, \(K\in\nats^*\) denotes the epoch size, so that adaptations are only performed at iterations belonging to the set \(K\nats^*=\{Kj:j\in\nats^*\}\), and \(a\in\nats^*\) is the adaptation step size. The Bernoulli random variable \(B_\ell\sim\Bern(p_{\lfloor \ell/K\rfloor})\) determines whether an adaptation is performed at the end of the current epoch, where the adaptation probabilities satisfy \(p_{\lfloor \ell/K\rfloor}\to0\) as \(\ell\to\infty\).

Equivalently, the APM kernel at iteration $\ell \geq 1$ is given by the following non-homogeneous Markov kernel
\begin{align}
  \bar{P}(\ell; \theta,w,N; \mathrm{d}\vartheta, \mathrm{d}z, M) = \tilde{P}_{N}(\theta,w ; \mathrm{d}\vartheta, \mathrm{d}z) \left\{\left(1-p_{\lfloor \frac{\ell}{K} \rfloor}\right) \ind_{\{N=M\}}+ p_{\lfloor \frac{\ell}{K} \rfloor} \ind_{\{\psi(N, \hat{\sigma}_\ell)=M\}}\right\}.
  \label{APMkernel}
\end{align}

The complete APM algorithm is presented in Algorithm~\ref{algo:apm}, incorporating all components developed in this section. 

\begin{algorithm}
\caption{Adaptive Pseudo-Marginal Algorithm}\label{algo:apm}
\begin{algorithmic}[1]
\State \textbf{Input}: Initial $\theta_0$, $u_0$,  uniform auxiliary variable parameter $\delta$, $N_0$, proposal covariance $\Sigma_p$, epoch size $K$, target standard deviation $\sigma_\textopt$, optimal scaling $l_{\text{opt}}$, tolerance $\sigma_\text{e}$, step size $a$, adaptation probability $p_j$, transformation $h$, $\hat \theta_0=\theta_0$
\State \textbf{Output}: Marginal sample chain $\{\theta_\ell, 1 \leq \ell \leq L\}$

\For{each iteration $\ell = 1$ to $L$}
    \State Propose $\vartheta_\ell|\theta_{\ell-1} \sim \Norm(\theta_{\ell-1}, {l_{\text{opt}}^2 \Sigma_p}/{d})$
    \State Generate the auxiliary variables $v_\ell | \vartheta_\ell \sim m_{N_{
    \ell-1},\vartheta_\ell}$
    and $\zeta_\ell \sim \Unif(1-\delta, 1+\delta)$
    \State Compute $\widehat{p}_{N_{\ell-1},T}(y|\vartheta_{\ell},v_\ell, \zeta_\ell)$  using Equation~\eqref{augmentedlikeestm}
    \State Accept the proposal $(\vartheta_\ell, v_\ell, \zeta_\ell)$ with  probability 
    $$\alpha_{N_{\ell-1}}(\theta_{\ell-1}, u_{\ell-1}, \xi_{\ell-1}; \vartheta_\ell, v_\ell, \zeta_\ell)= 
    1\wedge
    \frac{
    \widehat p_{N_{\ell-1},T}(y|\vartheta_\ell,v_\ell, \zeta_\ell)
    p(\vartheta_\ell)
    }{
    \widehat p_{N_{\ell-1},T}(y|\theta_{\ell-1},u_{\ell-1}, \xi_{\ell-1})
    p(\theta_{\ell-1})
    }$$ 
    \hspace{0.4cm} If accepted, set $(\theta_\ell,u_\ell,\xi_\ell)=(\vartheta_\ell,v_\ell,\zeta_\ell)$, otherwise set $(\theta_{\ell},u_{\ell}, \xi_{\ell}) =(\theta_{\ell-1},u_{\ell-1}, \xi_{\ell-1})$
    \State Compute $\hat v_\ell = h(v_\ell, \vartheta_\ell, \hat\theta_{\ell-1})$ and  $\log {\widehat{p}_{N_{\ell-1},T}(y|\hat{\theta}_{\ell-1},\hat{v}_\ell, \zeta_\ell)}$
    \State Set $N_\ell=N_{\ell-1}$ and $\hat\theta_{\ell} = \hat\theta_{\ell-1}$
    \If{$\ell \in K\natsstar$}
        \State Compute the additive noise standard deviation estimator $\widehat{\sigma}_\ell$ using \eqref{hatsigma}
        \If{$\widehat{\sigma}_\ell > \sigma_\textopt + \sigma_\text{e}$} \label{adaptstep}
            \State $N_\ell=N_{\ell}+a$ with probability $p_j$ ($\ell=Kj$)
        \ElsIf{$\widehat{\sigma}_\ell < \sigma_\textopt - \sigma_\text{e}$ and $N_{\ell-1} > a$}
            \State $N_\ell=N_{\ell}-a$  with probability $p_j$
        \EndIf
        \State Update sample mean $\hat{\theta}_\ell$ of parameters using Equation \eqref{samplemean}
    \EndIf
\EndFor
\end{algorithmic}
\end{algorithm}

In addition to the usual tuning parameters for the random-walk proposal, the APM algorithm involves four additional tuning parameters: the step size \(a\), the epoch size \(K\), the adaptation probability \(p_j\), and the tolerance \(\sigma_{\mathrm e}\). Among these, our numerical experiments indicate that the algorithm is primarily sensitive to the choice of the epoch size \(K\), which controls the accuracy of the estimated additive noise standard deviation used for adaptation. Larger values of \(K\) produce more stable estimates but result in slower adaptation.

The choice \(p_j=j^{-1/2}\) satisfies the diminishing adaptation condition while allowing adaptation to continue throughout the simulation, and is commonly adopted in the adaptive MCMC literature (see \cite{roberts_examples_2009}). The remaining parameters were found to have little effect on the algorithm's performance in the sensitivity analyses and were therefore fixed at the values reported in Table~\ref{tab:APMparameters}.

\begin{table}[ht]
\centering
\caption{Default values of the APM tuning parameters and augmentation parameter used throughout the paper.}
\label{tab:APMparameters}
\begin{tabular}{lll}
\hline
Parameter & Default value & Comment \\
\hline
Step size \(a\) 
    & \(1\) 
    & Low sensitivity \\

Epoch size \(K\) 
    & \(100\) 
    & Main tuning parameter \\

Adaptation probability \(p_j\) 
    & \(j^{-1/2}\) 
    & Ensures diminishing adaptation \\

Tolerance \(\sigma_{\mathrm e}\) 
    & \(0.015\) 
    & Low sensitivity \\

Augmentation parameter \(\delta\) 
    & \(0\) 
    & Used only for theoretical analysis \\
\hline
\end{tabular}
\end{table}

The augmentation parameter \(\delta\) is introduced solely for the theoretical analysis developed in Section~\ref{sec-APMproof}. In all numerical experiments (see Sections~\ref{sec-syntheticex} and \ref{sec-realex}), we implement the limiting case \(\delta=0\), corresponding to the original PM algorithm with transition kernel~\eqref{equa4}.

\section{Ergodicity of the Adaptive Pseudo-Marginal Algorithm}
\label{sec-APMproof}
In this section, we establish sufficient conditions ensuring the ergodicity of the APM algorithm. Our main result is Theorem~\ref{thm1}, which establishes convergence of the parameter component of the APM chain to the target distribution~$\pi$. To improve readability, the complete proofs are deferred to Appendix~\ref{appendixB}.

As discussed in Section~\ref{sec-intro}, the ergodicity of an adaptive MCMC algorithm is classically established by verifying the \textit{diminishing adaptation} and \textit{containment} conditions. We first establish the diminishing adaptation condition. This property is satisfied by construction, since the probability of performing an adaptation at the $j$-th epoch, denoted by $p_j$, converges to zero as $j\to\infty$. The corresponding result is stated in Proposition~\ref{proposition1}.

The containment condition is generally more difficult to verify. The authors in \cite{roberts_coupling_2007} showed that this condition follows from the simultaneous strong aperiodic geometric ergodicity condition (see Definition~\ref{def2} in Appendix~\ref{appendixA}), which combines a minorization condition with a geometric drift condition that is uniform over the set of adaptation parameters. However, this framework is not well suited to the PM setting, where the underlying kernels typically exhibit polynomial rather than geometric ergodicity (see \cite{andrieu_convergence_2015}). We therefore adopt the framework of \cite{atchade_limit_2010}, which provides sufficient conditions for the ergodicity of adaptive MCMC algorithms using polynomially ergodic kernels. Our strategy is to verify these conditions for the family of PM kernels.

The proof of adaptive ergodicity in \cite{atchade_limit_2010} proceeds in two steps. First, the simultaneous polynomial drift and minorization conditions imply the properties collected in Assumption~A1 of \cite{atchade_limit_2010} (see also Assumption~\ref{SM-def2} in Appendix~\ref{appendixA}). Second, Assumption~A1, together with the diminishing adaptation condition, yields ergodicity of the adaptive algorithm. Both steps rely on the assumption that all non-adaptive kernels share a common invariant distribution, which is not satisfied in the PM setting since the invariant distribution $\bar{\pi}_N$ depends on the number of particles $N$, which evolves throughout the algorithm.

We therefore begin by stating Lemma~\ref{lemma1}, which establishes that the simultaneous polynomial drift and minorization conditions, together with $\phi$-irreducibility and aperiodicity, imply a generalized version of Assumption~A1 in which the invariant distribution depends on the adaptation parameter. After verifying the assumptions of Lemma~\ref{lemma1} for the family of PM kernels, Theorem~\ref{thm1} follows from the proof of \cite[Theorem 2.1]{atchade_limit_2010}, with only one minor modification: the resulting convergence is restricted to functions of the form $f(\theta,w)=g(\theta)$. This restriction amounts to considering the marginal convergence of the parameter component and is natural in the PM setting, since the auxiliary variables introduced by the PM construction are ultimately discarded. Thus, the relevant convergence is that of the marginal distribution of the parameter toward the target distribution $\pi$, which does not depend on the adaptation parameter $N$.

It thus remains to establish the simultaneous polynomial drift and minorization conditions, together with $\phi$-irreducibility and aperiodicity for the family of non-adaptive PM kernels. These are the main technical results of this section and are established in Propositions~\ref{prop2},~\ref{prop3}, and~\ref{prop4}, respectively, under verifiable assumptions on the target distribution, the proposal, and the distribution of the likelihood noise.

We begin by introducing some notation that will be used throughout this section. Specifically, we recall the supremum norm on functions, the total variation norm on signed measures, the associated \(V\)-norm and the shifted transition kernels of a non-homogeneous Markov chain. For a function \(f:\mathcal{X}\to\reals\), the supremum norm is defined by
\[
\|f\|_\infty
=
\sup_{x\in\mathcal{X}} |f(x)|.
\]
For a function \(V:\mathcal{X}\to[1,\infty)\), the associated \(V\)-norm of \(f\) is defined by
\[
\|f\|_V
=
\sup_{x\in\mathcal{X}}
\frac{|f(x)|}{V(x)}.
\]
For a signed measure \(\mu\) on the measurable space \((\mathcal{X},\mathcal{B}(\mathcal{X}))\), we consider the total variation norm
\[
\|\mu\|_{\mathrm{TV}}
=
\sup_{\{f, \|f\|_\infty \leq 1\}}|\mu(f)|
=
2\sup_{A\in\mathcal{B}(\mathcal{X})}
|\mu(A)|,
\]
and the associated \(V\)-norm
\[
\|\mu\|_V
=
\sup_{\substack{f:\,\|f\|_V\leq 1}}
|\mu(f)|.
\]
Let $\{P_\gamma\}_{\gamma \in \mathcal{Y}}$ be a family of Markov transition kernels on a measurable space $(\mathcal{X}, \mathcal{B}(\mathcal{X}))$. 
We consider a non-homogeneous Markov chain
$\{(X_\ell,\Gamma_\ell),\ell\geq0\}$ on
$\mathcal{X}\times\mathcal{Y}$ with transition kernels
$\{\bar{P}(\ell;\cdot,\cdot),\ell\geq0\}$ satisfying, for any
$A\in\mathcal{B}(\mathcal{X})$,
\begin{align}
    \bar{P}(\ell;(x,\gamma);A\times\mathcal{Y})=P_\gamma(x,A).
    \label{condition1}
\end{align}

For $n\geq0$, define the shifted sequence of transition kernels by
\[\bar{P}_n(\ell;(x,\gamma),\cdot)=\bar{P}(n+\ell; (x,\gamma),\cdot),\qquad \ell\geq0.\]
We denote by $\mathbb{P}^{(n)}_{x, \gamma}$ and $\mathbb{E}^{(n)}_{x, \gamma}$ the
probability and expectation, respectively, associated with the
non-homogeneous Markov chain with transition kernels
$\{\bar{P}_n(\ell;\cdot,\cdot),\ell\geq0\}$ and initial distribution
$\delta_{x, \gamma}$. For the unshifted chain, we write
$\mathbb{P}_{x, \gamma}=\mathbb{P}^{(0)}_{x, \gamma}$ and
$\mathbb{E}_{x, \gamma}=\mathbb{E}^{(0)}_{x, \gamma}$.

We first state the formal diminishing adaptation result.
\begin{proposition}[Diminishing Adaptation] Let $D_\ell = \sup_{(\theta,w) \in \Theta \times \mathcal{W}} \| P_{N_{\ell+1}}(\theta, w; \cdot) - P_{N_\ell}(\theta, w; \cdot) \|_{\TV}$.
Suppose that $p_{\lfloor{\ell/K}\rfloor} \to 0$ as $\ell \to \infty$. Then the diminishing adaptation is satisfied; that is, for any \( \epsilon > 0 \), \(
\Prob(D_\ell \geq \epsilon) \to 0\) when $\ell \to \infty$.
\label{proposition1}
\end{proposition}

We now turn to the second ingredient of the proof. We first recall the {Simultaneous Minorization and Polynomial Drift Conditions} (Definition~\ref{def6}), which correspond to Assumption~A2 of \cite{atchade_limit_2010}. We then state Lemma~\ref{lemma1}, which establishes that these conditions imply a generalized version of Assumption~A1 in which the invariant distribution is allowed to depend on the adaptation parameter. It remains to verify the conditions of Definition~\ref{def6}. This verification is deferred to Propositions~\ref{prop2},~\ref{prop3}, and~\ref{prop4}.
\begin{definition}[Simultaneous Minorization and Polynomial Drift Conditions]
The family of Markov transition kernels $\{P_\gamma\}_{\gamma \in \mathcal{Y}}$ is said to satisfy the simultaneous minorization and polynomial drift conditions if:
\begin{enumerate}
    \item[(i)] For every $\gamma \in \mathcal{Y}$, the kernel $P_\gamma$ is $\phi$-irreducible and aperiodic.

    \item[(ii)] \textit{Simultaneous in $\gamma$ Polynomial Drift:} There exist a measurable set $C \subseteq \mathcal{X}$, a function $V:\mathcal{X}\to[1,\infty)$, a constant $\alpha\in(0,1)$, and constants $b,c>0$ such that, for all $\gamma\in\mathcal{Y}$ and $x\in\mathcal{X}$,
    \[
    P_\gamma V(x)\leq V(x)-cV^{1-\alpha}(x)+b\ind_C(x).
    \]
    \label{SimultaneousDrift}
    \item[(iii)] \textit{Simultaneous in $\gamma$ Minorization:} For every level set $B=\{x\in\mathcal{X}:V(x)\leq b\}$ of $V$ (for some $b>1$), there exist $\varepsilon_B>0$ and a probability measure $\nu_B$ such that, for all $\gamma\in\mathcal{Y}$,
    \[
    P_\gamma(x,\cdot)\geq \varepsilon_B\ind_B(x)\nu_B(\cdot),\qquad \forall x\in\mathcal{X}.
    \]
    \label{SimultaneousMino}
\end{enumerate}
\label{def6}
\end{definition}

\begin{lemma}
Let $\{P_\gamma\}_{\gamma \in \mathcal{Y}}$ be a family of Markov kernels on $(\mathcal{X}, \mathcal{B}(\mathcal{X}))$, each with invariant distributions $\pi_\gamma$. If the family satisfies the \textit{Simultaneous Minorization and Polynomial Drift Conditions} of Definition~\ref{def6}, then:
\begin{enumerate}
    \item There exists a non-decreasing function $r \colon \mathbb{N} \to (0,\infty)$ such that $\sum_{n=1}^{\infty}1/r(n)<\infty$ and
    \(
    \sup_{n\geq 0} \sup_{(x,\gamma)\in C\times\mathcal{Y}}
    \mathbb{E}_{x,\gamma}^{(n)}[r(\tau_C)]
    <\infty,
    \) where \(\tau_C=\inf\{\ell\geq1:X_\ell\in C\}\) denotes the return time to $C$.

    \item The following uniform convergence property holds
    \(
    \lim_{\ell \to\infty}
    \sup_{x\in\mathcal{X}}
    \sup_{\gamma\in\mathcal{Y}}
    \|P_\gamma^\ell(x,\cdot)-\pi_\gamma\|_V
    =0.
    \)

    \item The inequality $\sup_\gamma P_\gamma V \leq V$ holds on $C^c$, and
    \(
    \sup_{(x,\gamma)\in C\times\mathcal{Y}}
    \left\{P_\gamma V(x)+V(x)\right\}
    <\infty.
    \)
\end{enumerate}
\label{lemma1}
\end{lemma}

The proof of Theorem~\ref{thm1} now reduces to verifying the conditions of Definition~\ref{def6}. Once these conditions have been established, Theorem~\ref{thm1} follows by combining Lemma~\ref{lemma1} with the diminishing adaptation result of Proposition~\ref{proposition1} and restricting the convergence to the parameter component. The remainder of this section is therefore devoted to verifying the conditions of Definition~\ref{def6}.

We begin by verifying the simultaneous polynomial drift condition~\eqref{SimultaneousDrift}. For a fixed number of particles, the polynomial drift condition follows from the ergodicity analysis of PM algorithms developed in \cite{andrieu_convergence_2015}. The first three assumptions below are standard in that setting. They impose regularity conditions on the target distribution, the proposal kernel, and the likelihood estimator, and together ensure that the PM kernel associated with each fixed value of $N$ satisfies a polynomial drift condition.

\begin{hypothesis}[Regularity and Tail Behavior of the Target Density]
The target density $\pi$ is continuously differentiable and supported on $\reals^d$. We assume it possesses both super-exponentially decaying tails and regular contours. More precisely,
\begin{align*}
\frac{\theta}{\|\theta\|} \nabla\log\pi(\theta)
\underset{\|\theta\| \rightarrow \infty}{\longrightarrow} -\infty \quad \text{and} \quad \limsup_{\|\theta\| \rightarrow \infty} 
\frac{\theta}{\|\theta\|} \frac{\nabla\pi(\theta)}{\|\nabla\pi(\theta)\|} < 0,
\end{align*}
where $\|\cdot\|$ is the Euclidean norm.
\label{hyp1}
\end{hypothesis}

\begin{hypothesis}[Local Lower Bound for the Proposal Density]
The proposal distribution is a random walk with symmetric density $q$. We assume that for every
$\delta > 0$, there exists $\varepsilon_{\delta,q} > 0$ such that for all
$\theta, \vartheta \in \reals^d$,
\begin{align*}
    \|\vartheta-\theta\| \leq  \delta \Rightarrow q(\vartheta|\theta) \geq \varepsilon_{\delta,q}.
\end{align*}
\label{hyp5}
\end{hypothesis}
 
\begin{hypothesis}[Uniform Bound on the Moments of the Weights]
There exist constants $\alpha_{0} > 0$ and $\beta_{0} >1$ such that for some $N_0 \geq 1$,
$$
M_{W_{N_0}} = \underset{\theta \in \Theta}{\mathrm{ess\,sup}} \int (w^{-\alpha_{0}} \vee w^{\beta_{0}}) \mathcal{Q}_{N_0,\theta}(\mathrm{d}w) < \infty,
$$
where $a\vee b=\max(a,b)$.
\label{hyp3}
\end{hypothesis}
To extend the polynomial drift condition uniformly to all \(N \geq N_0\), one additional assumption is required. The key idea is to compare the weight distributions corresponding to different values of \(N\) through the convex order. This ordering allows the polynomial drift condition established for the PM kernel associated with \(N_0\) to be transferred uniformly to every kernel with \(N \geq N_0\), yielding drift constants that are independent of \(N\).

\begin{definition}[Convex Order]  
For two random variables $X$ and $Y$, we say that $X$ is smaller than $Y$ in the convex order (denoted $X \preceq_{cx} Y$) if
$$
\E[h(X)] \leq \E[h(Y)]
$$
for all convex functions $h : \reals \rightarrow \reals$ for which the expectations exist.
\label{def1}
\end{definition}
\begin{hypothesis}[Convex Order of Weights]
There exists a fixed integer $N_0 \geq 1$ such that for all $N \geq N_0$, the weights $W_{N,\theta}$ are stochastically smaller in the convex order than $W_{N_0,\theta}$, i.e.,
$$W_{N,\theta}\preceq_{cx} W_{N_0,\theta}, \quad \forall \theta \in \Theta.$$
\label{hyp2}
\end{hypothesis}
\begin{remark}
    For likelihood estimators constructed using Importance Sampling (see Definition~3.9 in \cite{robert_monte_2004}), Assumption~\ref{hyp2} is satisfied. A proof is provided in Appendix~\ref{appendixD0} for the synthetic example of Section~\ref{sec-syntheticex}.
\end{remark}
Under these assumptions, the polynomial drift inequality can be established uniformly for all \(N\geq N_0\).
\begin{proposition}[Simultaneous Polynomial Drift]
Under Assumptions~\ref{hyp1}, \ref{hyp5}, \ref{hyp3}, and~\ref{hyp2}, there exist constants \(\overline{w}, M, b \in [1,\infty)\), \(\underline{w} \in (0,1]\), and \(\delta_V>0\) such that, for all \(N\geq N_0\) and all \((\theta,w)\in\Theta\times\mathcal{W}\),
\[
\tilde{P}_N V(\theta,w)
\leq
V(\theta,w)
-\delta_V V^{(\beta-1)/\beta}(\theta,w)
+b\,\ind_K(\theta,w),
\]
where
\[
V(\theta,w)
=
c_\pi^\eta \pi^{-\eta}(\theta)\left(w^{-\alpha}\vee w^\beta\right),
\]
with \(c_\pi:=\sup_{\vartheta\in\Theta}\pi(\vartheta)\), \(\eta\in(0,\alpha_0\wedge1\wedge(\beta_0-1))\), \(\alpha\in(\eta,\alpha_0]\), and \(\beta\in(1,\beta_0-\eta]\), where \(\alpha_0\) and \(\beta_0\) are defined in Assumption~\ref{hyp3}. The set \(K\subset\Theta\times\mathcal{W}\) is given by
\[
K=\left\{(\theta,w)\in\Theta\times\mathcal{W}:\|\theta\|\leq M,\; w\in[\underline{w},\overline{w}]\right\}.
\]
\label{prop2}
\end{proposition}

We now verify the simultaneous minorization condition~\eqref{SimultaneousMino}. Since this condition must hold for every level set of the drift function \(V\), we first show that the level sets of \(V\) are compact. We then establish a minorization condition for all compact subsets of \(\Theta\times\mathcal W\).
\begin{lemma}
For any $b > 1$, the level set
\(
B=\{(\theta,w)\in\Theta\times\mathcal W \; | \; V(\theta,w)\leq b\}
\)
is compact.
\label{lemma2}
\end{lemma}

We establish, now, the simultaneous minorization condition for compacts. Assumptions~\ref{hyp1}, \ref{hyp5}, and~\ref{hyp3}, which were introduced to establish the polynomial drift condition, also allow us to establish a minorization condition for the PM kernel associated with a fixed value of \(N\). This argument follows the general strategy used in the ergodicity analysis of PM algorithms (see \cite[Section 6]{andrieu_convergence_2015}). However, this fixed-\(N\) argument does not directly provide a minorization condition that is uniform over the whole family of adaptation parameters. We therefore need to introduce additional conditions which, together with the previous assumptions, ensure that the resulting minorization condition holds uniformly for all $N\geq N_0$.

Assumption~\ref{hyp7} is an expected property of a well-behaved likelihood estimator: its error ratio (weight) converges to $1$ as $N$ increases. This property is satisfied, for example, by importance sampling estimators; see \cite[page 137]{pitt_properties_2012}. However, this convergence alone is not sufficient to establish a minorization condition that is uniform in $N$, since the distribution of the weight concentrates at $1$ as $N\to\infty$. In the limit, the weight therefore behaves like a point mass at $1$, which does not provide an interval on which a uniform-in-$N$ lower bound can be obtained. This is precisely where the auxiliary perturbation $\xi\sim\Unif(1-\delta,1+\delta)$, introduced in Section~\ref{sec-pm}, becomes essential. Since the product of a point mass at $1$ and $\xi$ has distribution  $\Unif(1-\delta,1+\delta)$, the perturbation spreads the limiting point mass over an interval. For a fixed $\delta>0$, Assumptions~\ref{hyp7} and~\ref{hyp8} then allow us to establish a uniform lower bound for the density of the perturbed weight on a subinterval $[1-\delta/2,1+\delta/2]$, for all \(N\geq N_1\), where \(N_1\) is large enough. The remaining finite collection of particle numbers \(N_0\leq N<N_1\) is then handled using the fixed-\(N\) minorization together with the convex order assumption~\ref{hyp2}.

\begin{hypothesis}[Convergence of Estimator]
We assume that for every $\theta$ in a compact set $C\subset\mathbb{R}^d$, the weights $W_{N,\theta}$ converge in probability to $1$ as $N\to\infty$.
\label{hyp7}
\end{hypothesis}

\begin{hypothesis}[Continuity of the Mean Absolute Deviation]
For each $N\geq N_0$, the mapping
\(\theta\mapsto\mathbb{E}\!\left[|W_{N,\theta}-1|\right]\)
is continuous on every compact subset of $\mathbb{R}^d$.
\label{hyp8}
\end{hypothesis}

\begin{remark}
    A sufficient condition for Assumption~\ref{hyp8} is the following. On every compact set $C\subset\mathbb{R}^d$, suppose that the random variables $W_{N,\theta}$ can be defined on a common probability space such that $\theta\mapsto W_{N,\theta}$ is almost surely continuous and $|W_{N,\theta}-1|$ is bounded by an integrable function, uniformly in $\theta\in C$. Then, for any sequence $\theta_m\to\theta$ in $C$, almost sure continuity gives
    \[
    |W_{N,\theta_m}-1|
    \longrightarrow
    |W_{N,\theta}-1|,
    \]
    and the Dominated Convergence Theorem yields
    \[
    \E\left[|W_{N,\theta_m}-1|\right]
    \longrightarrow
    \E\left[|W_{N,\theta}-1|\right].
    \]
\end{remark}

The following proposition formalizes the uniform-in-$N$ minorization condition.

\begin{proposition}[All compacts are uniformly-in-$N$ 1-small for $\tilde{P}_N$]
Let $K$ be a compact set of $\Theta\times\mathcal W=\mathbb{R}^d\times\mathbb{R}_{+}^{*}$. Under Assumptions~\ref{hyp1},~\ref{hyp5}, \ref{hyp3}, \ref{hyp2}, \ref{hyp7}, and~\ref{hyp8}, $K$ is \(1\)-small for $\tilde{P}_N$ with a minorization measure independent of \(N\), that is,
\[
\exists\,\nu \text{ a non-trivial measure on }\mathcal{B}(\Theta\times\mathcal W)
\text{ such that }\forall N\geq N_0,
\]
\[
\tilde{P}_N(\theta,w;\cdot)\geq \ind_K(\theta,w)\,\nu(\cdot).
\]
\label{prop3}
\end{proposition}

It remains to verify that each non-adaptive PM kernel is $\mu_{\mathrm{Leb}}$-irreducible and aperiodic. Building on the irreducibility and aperiodicity results of \cite{andrieu_convergence_2015}, this property follows under the following additional assumption.

\begin{hypothesis}[Positive Rejection Probability]
The rejection probability of the PM algorithm in~\eqref{varrho} remains strictly positive for all $N\geq N_0$ and for all states $(\theta,u)\in\Theta\times\mathcal U$.
\label{hyp4}
\end{hypothesis}

\begin{proposition}
Under Assumptions~\ref{hyp1}, \ref{hyp5}, \ref{hyp2}, and~\ref{hyp4}, the PM kernel $\tilde{P}_N$ is $\mu_{\mathrm{Leb}}$-irreducible and aperiodic for every $N\geq N_0$.
\label{prop4}
\end{proposition}

Propositions~\ref{prop2}, \ref{prop3}, and~\ref{prop4} establish, respectively, the simultaneous polynomial drift condition, the simultaneous minorization condition, and the $\mu_{\mathrm{Leb}}$-irreducibility and aperiodicity of each non-adaptive PM kernel. These results verify the assumptions of Lemma~\ref{lemma1}. Although the invariant distribution of the PM kernel depends on the number of particles \(N\), its marginal distribution with respect to the parameter component is the target distribution \(\pi\) for every \(N\geq N_0\). Therefore, applying Lemma~\ref{lemma1} together with the diminishing adaptation result in Proposition~\ref{proposition1} yields convergence in total variation of the marginal APM chain towards \(\pi\), which is the relevant notion of convergence for PM algorithms.
\begin{theorem}
Under Assumptions~\ref{hyp1}, \ref{hyp5}, \ref{hyp3}, \ref{hyp2}, \ref{hyp7}, \ref{hyp8}, and~\ref{hyp4}, the APM algorithm is ergodic in the following sense
\[
\left\|
\mathcal{L}(\theta_\ell)-\pi
\right\|_{\mathrm{TV}}
\underset{\ell\to\infty}{\longrightarrow}0,
\]
where $\mathcal{L}(\theta_\ell)$ denotes the distribution of the parameter component of the APM chain at iteration $\ell$.
\label{thm1}
\end{theorem}
The proof of Theorem~\ref{thm1} follows the same argument as the proof of Theorem 2.1 in \cite{atchade_limit_2010}. The proof is provided in Appendix~\ref{sub6appendixB} for completeness.

\section{Synthetic Data Example}
\label{sec-syntheticex}
In this section, we consider a synthetic data example to verify the assumptions required by Theorem~\ref{thm1}, apply the APM algorithm, and compare its performance with the two versions of the non-adaptive method introduced in Section~\ref{sec-review}. We first introduce the model and describe how the synthetic data are generated. We then define the quantities needed both for the theoretical verification and for the implementation of the algorithms, namely the likelihood estimator, the posterior distribution, and the proposal distribution. Based on these definitions, we state a proposition establishing that the assumptions required by Theorem~\ref{thm1} hold for this example. We subsequently define the parameters required for the implementation of the different algorithms, specify the quantities used for comparison, and present the corresponding results.

We consider the Bayesian latent variable model
\[
Y_t | U_t \sim \mathcal{N}(U_t,1),
\qquad
U_t | \theta
\sim
\mathcal{N}\left(
\theta,\frac{1}{\theta^2+1}
\right),
\qquad t\in\{1,\ldots,T\},
\]
where the observations are conditionally independent given the latent variables, and the latent variables are conditionally independent given $\theta\in\mathbb{R}$. The parameter is assigned a diffuse Gaussian prior
\[
\theta\sim\mathcal{N}(0,\sigma_0^2),
\qquad \sigma_0=10^5.
\]
The data were generated under the parameter value $\bar{\theta}=0$, with $T=200$ observations. More precisely, for each $t$, $U_t$ was first generated from
\(
\mathcal{N}\left(
\bar{\theta},
{1}/\{{\bar\theta}^2+1\}
\right),
\)
followed by the generation of $Y_t|U_t$ from $\mathcal{N}(U_t,1)$.

Although the observed likelihood is available in closed form, we use an unbiased positive estimator to formulate the example within the PM and APM framework. The estimator used by the PM and APM algorithms is defined as
\begin{align}
    \widehat{p}_{T,N}(y |\theta, U)
    &= \prod_{t=1}^{T} \frac{1}{N} \sum_{n=1}^{N}
    \varphi(y_t ; U_{t,n}, 1),
    \label{se:estimator}
\end{align}
where \( U_{t,n} |\theta \sim \mathcal{N}(\theta, 1/\{\theta^2 + 1\}) \) are independent and identically distributed for \( t \in \{1,\ldots,T\} \) and \( n \in \{1,\ldots,N\} \).

The corresponding posterior density is available up to a normalizing constant and is given by
\begin{align}
\pi(\theta)
&\propto
\left(
\frac{\theta^2+1}{\theta^2+2}
\right)^{T/2}
\exp\left\{
-\frac{1}{2}
\left(
\frac{\theta^2+1}{\theta^2+2}
\sum_{t=1}^{T}(\theta-y_t)^2
+
\frac{\theta^2}{\sigma_0^2}
\right)
\right\}.
\label{SEpi}
\end{align}
This posterior is used to implement a MH algorithm, which provides a benchmark for the posterior estimates obtained with the PM and APM algorithms.

For the MH, PM, and APM algorithms, we use the Gaussian random-walk proposal
\begin{align}
q(\vartheta|\theta)
&=
\varphi\left(
\vartheta;\theta,
l_\text{opt}^2\frac{2}{T}
\right)
=
\varphi\left(
\vartheta;\theta,\frac{8}{T}
\right),
\label{intru}
\end{align}
where the scaling $l_\text{opt}^2=4$ follows the recommendation of \cite{schmon_large-sample_2021} for a one-dimensional parameter. The term $2/T$ corresponds to the inverse of the Fisher information evaluated at $\bar{\theta}=0$, the parameter value used to generate the synthetic data; see Theorem 10.1 in \cite{Vaart_1998} for a theoretical justification of using the inverse Fisher information as the proposal variance.

With the quantities defined above, assumptions required by Theorem~\ref{thm1} are verified for this example. In particular, we have the following proposition.

\begin{proposition}
\label{se:verification}
Theorem~\ref{thm1} holds for the synthetic data example.
\end{proposition}

\begin{proof}[Proof of Proposition~\ref{se:verification}]
See Appendix~\ref{appendixD0} for the verification of Assumptions~\ref{hyp1}, \ref{hyp5}, \ref{hyp3}, \ref{hyp2}, \ref{hyp7}, \ref{hyp8}, and~\ref{hyp4}.
\end{proof}

We next describe the implementation of the four approaches considered: the MH algorithm, the two versions of the non-adaptive method, and the APM. Each method was run independently 10 times. Table~\ref{tab:method_comparison} in Appendix~\ref{appendixD1} summarizes the common parameter settings and initial values used across the four approaches to ensure a fair comparison. Further details on the computational setup are provided in the same appendix.

Since the posterior density in Equation~\eqref{SEpi} is available up to a normalizing constant, the MH algorithm can be implemented directly. The results from the 10 independent MH runs are reported in Table~\ref{se:tab3}.

For the non-adaptive methods, each run consisted of three steps. First, a preliminary PM run was performed as described in \ref{step1}, with the corresponding results reported in Appendix~\ref{appendixD2}. Second, in \ref{step2}, we set $\sigma_{\mathrm{opt}}=1.16$, following the guidelines in Table~1 of \cite{schmon_large-sample_2021}. Two implementations of this step were considered. For the naive non-adaptive method, the empirical standard deviation of the log-likelihood estimator $\widehat{\sigma}_{N}(\hat{\theta}_1)$ was estimated for every integer $N$ between 100 and 1000, yielding values of $N_{\mathrm{opt}}$ ranging from 197 to 208 across the 10 runs. For the non-adaptive method with dichotomic search, the same quantity was estimated using the search procedure described in Appendix~\ref{appendixD3}, initialized with the interval $[100,1000]$ and stopped when the length of the final interval reached $a_1=1$. This procedure yielded values of $N_{\mathrm{opt}}$ ranging from 199 to 206 across the 10 runs. Further details on both search procedures are provided in Appendix~\ref{appendixD3}.

Finally, \ref{step3} was performed using the corresponding $N_{\mathrm{opt}}$ obtained in \ref{step2}, with the results reported in Table~\ref{se:tab3}.

The APM algorithm was implemented using the default tuning parameters reported in Table~\ref{tab:APMparameters}. As outlined in Section~\ref{sec-APM}, a transformation of the proposed auxiliary variables is required to evaluate the likelihood estimator at the fixed parameter value $\hat{\theta}_\ell$, which estimates $\bar{\theta}$, the true parameter used to generate the data. For each epoch \( j \), with \( \ell = Kj \), the auxiliary variables \( V_{i,t,n}|\vartheta_i \sim \mathcal{N}(\vartheta_i, 1/\{\vartheta_i^2 + 1\}) \) for \( i \in \{\ell - K + 1, \cdots, \ell\} \), \( t \in \{1, \cdots, T\} \), and \( n \in \{1, \cdots, N_{i-1}\} \), were transformed linearly as
\begin{align}
    \widehat{V}_{i,t,n} = h(V_{i,t,n},\vartheta_i,\hat{\theta}_{i-1}) = \sqrt{\frac{\vartheta_i^2 + 1}{\hat{\theta}_{i-1}^2 + 1}}(V_{i,t,n} - \vartheta_i) + \hat{\theta}_{i-1},
    \label{SE:transformation}
\end{align}
so that \( \widehat{V}_{i,t,n}|\hat{\theta}_{i-1} \sim \mathcal{N}(\hat{\theta}_{i-1}, 1/\{\hat{\theta}_{i-1}^2 + 1\}) \). This transformation allowed the evaluation of the log-likelihood estimator at $\hat\theta_{i-1}$ and $\widehat{V}_i=\{\widehat{V}_{i,t,n}\}_{1 \leq t\leq T, 1 \leq n\leq N_{i-1}}$.

\noindent
Using Equation~\eqref{se:estimator} we get
{
\begin{align*}
   \log\{ \widehat{p}_{T, N_{i-1}}(y |\hat{\theta}_{i-1}, \widehat{V}_i) \} 
   &= \sum_{t=1}^{T} \log \{ \bar\varphi(y_t ; \widehat{V}_{i,t}, 1)\},
\end{align*}}

\noindent
where 
$\bar\varphi(y_t ; \widehat{V}_{i,t}, 1)=\sum_{n=1}^{N_{i-1}} \varphi(y_t ; \widehat{V}_{i,t,n}, 1)/N_{i-1}$.
The standard deviation of this log-likelihood is then computed after each epoch using Equation~\eqref{hatsigma}.

All PM-based methods are compared in terms of posterior mean and variance, acceptance rate $\widehat{P}$, inefficiency factor $\widehat{IF}$, number of particles, and execution time. The inefficiency factor was estimated using the \textit{overlapping batch means method} of \cite{flegal_batch_2010}. The reported execution time for each non-adaptive method corresponds to the total time required to complete all three steps. Thus, the comparison accounts for the running time of the preliminary run, selecting $N_{\mathrm{opt}}$ as well as that of the final PM run. For the APM, the reported execution time corresponds to the complete APM run.

\begin{table}
\caption{Comparison of summary statistics and execution times for MH, the two versions of the non-adaptive method, and the APM using $10^6$ iterations with a $20\%$ burn-in. Values are reported as mean $\pm$ standard deviation, except for $N$ using median [min,max] over 10 independent runs.}
\centering
\resizebox{\textwidth}{!}{
\begin{tabular}{lcccc}
\hline
 & MH & Naive Non-adaptive & Non-adaptive & APM \\
\hline
Optimal $N$ & -- & 202 [197,208] & 203 [199,206] & 199 [195,205] \\
Posterior Mean & $-0.0261 \pm 0.0002$ & $-0.0260 \pm 0.0004$ & $-0.0260 \pm 0.0004$ & $-0.0259 \pm 0.0003$ \\

Posterior Variance & $0.00911 \pm 0.00003$ & $0.00910 \pm 0.00006$ & $0.00909 \pm 0.00003$ & $0.00910 \pm 0.00005$ \\

Acceptance Rate $\widehat P$ (\%) & $48.68 \pm 0.04$ & $26.18 \pm 0.31$ & $26.28 \pm 0.22$ & $25.00 \pm 0.22$ \\

Inefficiency Factor $\widehat{IF}$ & $4.49 \pm 0.15$ & $12.00 \pm 0.59$ & $11.77 \pm 0.30$ & $12.31 \pm 0.54$ \\

Execution Time & 15s $\pm$ 0.4s  & 13h31m05s $\pm$ 4m15s & 1h19m31s $\pm$ 11m21s & \textbf{1h02m22s $\pm$ 1m03s} \\
\hline
\end{tabular}
}
\label{se:tab3}
\end{table}

The results in Table~\ref{se:tab3} show that the three PM-based methods produce approximately identical posterior estimates, all of which agree with the MH benchmark. 

The computational cost of tuning $N$ differs substantially between the two versions of the non-adaptive PM method. The naive version requires an average execution time of 13h31m05s, as it exhaustively evaluates all values of $N$ between 100 and 1000 in \ref{step2}. In contrast, the complete non-adaptive method uses a dichotomic search and requires an average of 1h19m31s. Thus, the dichotomic procedure substantially reduces the computational cost of tuning $N$.

For the comparison with APM, we use the non-adaptive method with the dichotomic search as the benchmark. This provides a conservative comparison by giving the non-adaptive method the benefit of the more efficient dichotomic tuning procedure. The main practical advantage of APM is that it eliminates the need for this preliminary optimization of $N$. Interestingly, APM also provides an additional computational benefit: in this example, its average execution time is 1h02m22s, compared with 1h19m31s for the complete non-adaptive method. Thus, APM is approximately $1.275$ times faster than the complete non-adaptive method.

The APM also provides a higher estimated effective sampling rate when running time is taken into account. Using the $800\,000$ post-burn-in samples, the number of effective samples per minute is
\begin{align}
\frac{\text{Sample size}}
{\widehat{IF}\times\text{Time(min)}}
&=
\frac{800\,000}{12.31\times62.37}
\approx
1\,042
\text{ effective samples/minute}.
\label{ESperMin}
\end{align}
For the non-adaptive method with dichotomic search, the estimated  effective samples per minute is 855, giving APM approximately $22\%$ more effective samples per minute.

Finally, we see in Figure~\ref{se:fig2} that the number of particles progressively stabilizes around values close to those identified by the non-adaptive method, with final values ranging from 195 to 205 across the 10 runs. Additional trace and autocorrelation plots for $\theta$ are provided in Appendix~\ref{appendixD4}.

\begin{figure}
    \centering
    \includegraphics[width=0.8\linewidth]{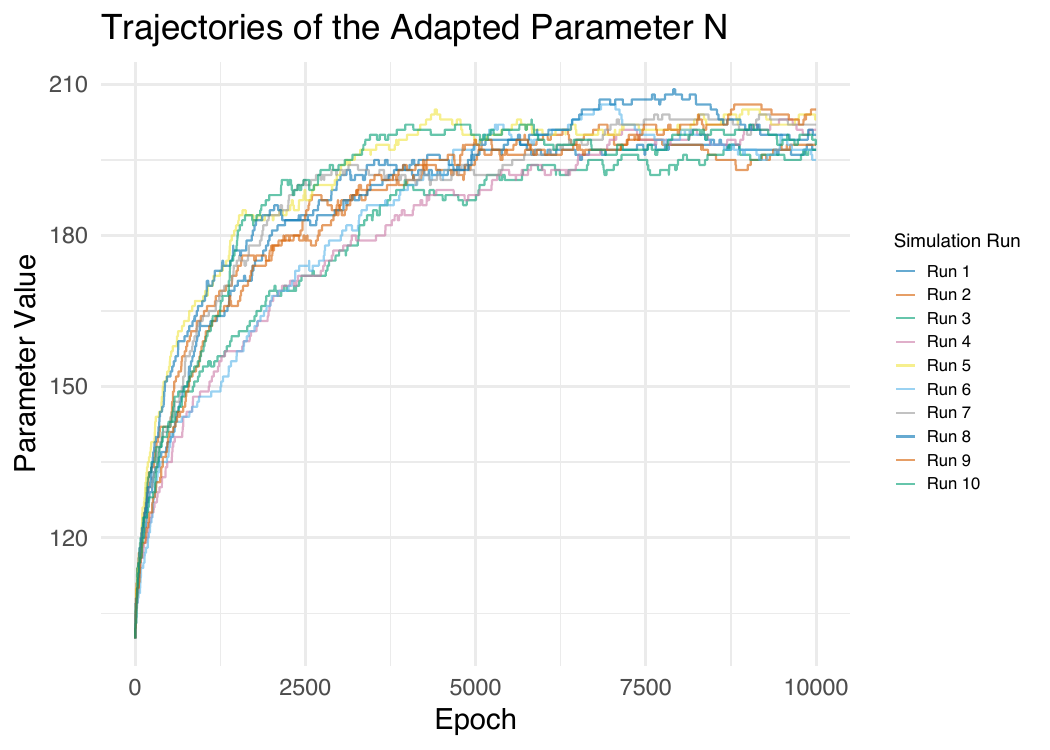}
    \caption{Trace of \( N \) using \( 10^4 \) \textit{epochs} of $100$ iterations across 10 APM runs.}
    \label{se:fig2}
\end{figure}

\section{Real Data Study}
\label{sec-realex}
In this section, we evaluate the performance of the APM algorithm relative to the non-adaptive method using a real dataset from a longitudinal cohort study of preschool-aged children in Indonesia. The dataset was previously analyzed by \cite{zeger_generalized_1991} via Bayesian mixed-effects models, and later by \cite{schmon_large-sample_2021} to illustrate weak convergence properties of PM chains. We use the same modeling framework as in~\cite{schmon_large-sample_2021}.

The data consist of \(1200 \) repeated binary responses observed across \( T = 275 \) subjects. Each response indicates the presence or absence of a respiratory infection. Covariates include age, sex, height, a vitamin deficiency indicator, a below-average height indicator, two seasonal terms, and an intercept, yielding a total of eight covariates.

To account for intra-subject correlation, a subject-specific random intercept \( U_t|\tau \sim \mathcal{N}(0, \tau) \) is introduced for each \( t \in \{1, \dots, T\} \), with \( \tau > 0 \). Conditional on \( \theta = (\beta, \tau) \in \mathbb{R}^9 \), the variables \( U_t \) are mutually independent.

Let \( Y_t = (Y_{t,1}, \dots, Y_{t,J_t}) \in \{0,1\}^{J_t} \) denote the binary responses observed for subject \( t \), where \( J_t \) represents the number of repeated measurements for that subject. Conditionally on the random effect \( U_t \in \mathbb{R} \) and parameters \( \theta \), \( Y_t \) are modeled as independent variables via a logistic regression model.
\[
g(y_t | u_t, \theta) = \prod_{j=1}^{J_t} \frac{\exp(y_{t,j} \eta_{t,j})}{1 + \exp(\eta_{t,j})}, \quad \eta_{t,j} = c_{t,j}^\top \beta + u_t,
\]
where \( c_{t,j} \in \mathbb{R}^{8} \) denotes the covariate vector. The random effects density is given by
\[
f(u_t | \theta) = \varphi(u_t; 0, \tau).
\]
The prior on \( \theta \) factorizes as
\begin{align*}
    p(\theta) = p(\beta) p(\tau) = \varphi(\beta; 0_{8}, 10^4 I_{8}) \cdot p(\tau; 1, 1.5),
\end{align*}
where \( p(\tau; a_1, a_2) \) denotes the density of an inverse-gamma distribution with shape \( a_1 \) and scale \( a_2 \), i.e.,
\[
p(\tau; a_1, a_2) = \frac{1}{\Gamma(a_1)} \cdot \frac{a_2^{a_1}}{\tau^{a_1 + 1}} \exp\left( -\frac{a_2}{\tau} \right).
\]
This setup yields the following likelihood function,
{
\begin{align*}
    p_T(y | \theta) = \prod_{t=1}^T &\int_{\mathbb{R}} \left[ \prod_{j=1}^{J_t} \frac{\exp(y_{t,j}(c_{t,j}^\top \beta + u))}{1 + \exp(c_{t,j}^\top \beta + u)} \right] \varphi(u; 0, \tau) \, \mathrm{d}u.
\end{align*}
}The resulting posterior distribution is therefore given by:
{\small
\begin{align*}
    \pi(\theta) 
    &\propto p_T(y | \theta) \, p(\theta) = \bigg\{ \prod_{t=1}^T \int_{\mathbb{R}} \prod_{j=1}^{J_t} \frac{\exp \left\{ y_{t,j}(c_{t,j}^\top \beta + x) \right\}}{1+\exp \left\{ c_{t,j}^\top \beta + x \right\}} 
    \varphi(x; 0, \tau) \, \mathrm{d}x \bigg\} \, \varphi(\beta; 0_{8}, 10^4 I_{8}) \, p(\tau; 1, 1.5).
\end{align*}}

The posterior \( \pi\) cannot be evaluated pointwise due to intractable integrals in the likelihood. Hence, direct implementation of the MH algorithm is not feasible.

\noindent
Within the PM and APM frameworks, the likelihood is estimated using Importance Sampling. The estimator takes the form
\[
\widehat{p}_{T,N}(y | \theta, U) = \prod_{t=1}^T  \frac{1}{N} \sum_{n=1}^N \varpi(y_t, U_{t,n}, \theta),
\]
where the importance weights are defined by
\[
\varpi(y_t, U_{t,n}, \theta) = \frac{g(y_t | U_{t,n}, \theta) f(U_{t,n} | \theta)}{s(U_{t,n} | y_t, \theta)},
\]
with $s(U_{t,n} | y_t, \theta) = \varphi(U_{t,n}; \hat{u}_t, \tau)$, \( U_{t,n} \sim \mathcal{N}(\hat{u}_t, \tau) \) and \( \hat{u}_t = \argmax_{u}  g(y_t | u, \theta) f(u | \theta) \).

Substituting the expressions for \( g \), \( f \), and \( s \), the estimator becomes
{
\begin{align}
\widehat{p}_{T,N}(y | \theta, U) 
= \prod_{t=1}^T \frac{1}{N} &\sum_{n=1}^N 
 \frac{ \varphi(U_{t,n}; 0, \tau)}{\varphi(U_{t,n}; \hat{u}_t, \tau)} \left[ \prod_{j=1}^{J_t} \frac{\exp\{y_{t,j}(c_{t,j}^\top \beta + U_{t,n})\}}{1 + \exp\{c_{t,j}^\top \beta + U_{t,n}\}} \right].
\label{re:estim}
\end{align}}

In this application, it is recalled that the parameter dimension is \( d = 9 \). To implement both the non-adaptive and the APM algorithm, the following Gaussian proposal distribution is used
{\small
\begin{align}
    q(\vartheta |\theta) = \varphi\left(\vartheta; \theta, \frac{l_\textopt^2}{d}{\Sigma}_p\right) = \varphi\left(\vartheta; \theta, \frac{2.2^2}{9}{\Sigma}_p\right),
    \label{intru2}
\end{align}
}where the scaling parameter \( l_{\text{opt}} = 2.2 \) follows the recommendation of \cite{schmon_large-sample_2021}, 
and \( \Sigma_p \) is set to the value used in \cite{schmon_large-sample_2021} (see Appendix~\ref{appendixE1} for the exact value of \( \Sigma_p \)).

The procedure outlined in Section~\ref{sec-syntheticex} will be followed for comparing both versions of the non-adaptive method and the APM algorithm. For this reason, some of the explanatory details will be omitted. In this example, simulations were conducted using the same computational environment as described for the synthetic data example in Appendix~\ref{appendixD1}. Similarly, key parameters were aligned across all implementations (see Table~\ref{tab:method_comparisonRE} in Appendix~\ref{appendixE1}).

The implementation steps for the non-adaptive methods (Section~\ref{sec-review}) on the real data example are now described. First and for each run, a preliminary execution of the PM algorithm was conducted. The corresponding results are in Table~\ref{re:tab1} in Appendix~\ref{appendixE2}.  

For implementing \ref{step2}, the value \( \sigma_{\text{opt}} = 1.44 \) was chosen from Table 1 in \cite{schmon_large-sample_2021} since the dimension $d=9$. For each run, using the posterior mean estimate from \ref{step1}, the standard deviations of the additive noise were estimated for various values of \( N \). In the naive non-adaptive framework, we tested all integers between $10$ and $100$. In the non-adaptive method with {dichotomic search}, the initial interval was set as \( [10, 100] \) for the number of particles and terminated when the length of the final interval reached \( a_1 = 1 \). Details of the estimations of each run along with the optimal values of $N$ are in Table~\ref{re:step2} in Appendix~\ref{appendixE3}.

Finally, \ref{step3} was performed using the optimal number of particles for each run. The Euclidean norm of the posterior mean and the Frobenius norm of the posterior covariance matrix were averaged over the 10 independent runs. The results are reported in the first column of Table~\ref{re:tab3} for the naive non-adaptive method and in the second column for the non-adaptive method using dichotomic search.

In the remainder of this section, we describe the implementation of the APM algorithm for the real data example. The additional APM-specific parameters were set to their default values given in Table~\ref{tab:APMparameters}. As in Section~\ref{sec-syntheticex}, a similar linear transformation was used for the proposed auxiliary variables. The auxiliary variables of the epoch $j$ (with \( \ell = Kj \)), \( V_{i,t,n}|\vartheta_i \sim \mathcal{N}(\hat {v}_{i,t}, \uptau_{i}) \), 
where \( i \in \{\ell-K+1, \ldots, \ell\} \), $t \in \{1, \ldots, T\} $, $n \in \{1, \ldots, N_{i-1}\}$, $\hat {v}_{i,t} = \argmax_{u}  g(y_t | u, \vartheta_i) f(u | \vartheta_i)$,  and $\uptau_i$ is the last component of $\vartheta_i$, were transformed as
{\small
\begin{align*}
    \widehat{V}_{i,t,n} = h(V_{i,t,n}, \vartheta_i, \hat\theta_{i-1}) 
    = \sqrt{\frac{\hat \tau_{i-1}}{\uptau_{i}}}(V_{i,t,n} - \hat {v}_{i,t}) + \widehat{\hat{u}}_{i-1,t},
\end{align*}
}where $\widehat{\hat{u}}_{i-1,t}=\argmax_{u}  g(y_t | u, \hat\theta_{i-1}) f(u | \hat\theta_{i-1})$. 

This enabled the evaluation of the log-likelihood estimator using Equation~\eqref{re:estim}, 
{\small
\begin{align*}
\log \widehat{p}_{T,N_{i-1}}(y|\hat\theta_{i-1},\widehat{V}_i)
&= 
\sum_{t=1}^T \log\left\{ \frac{1}{N_{i-1}} \sum_{n=1}^{N_{i-1}}
\frac{\varphi(\widehat{V}_{i,t,n}; 0, \hat\tau_{i-1})} {\varphi(\widehat{V}_{i,t,n}; \widehat{\hat{u}}_{i-1,t}, \hat\tau_{i-1})} \left[ \prod_{j=1}^{J_t} \frac{\exp\{y_{t,j}(c_{t,j}^\top \hat\beta_{i-1} + \widehat{V}_{i,t,n})\}}{1 + \exp\{c_{t,j}^\top \hat\beta_{i-1} + \widehat{V}_{i,t,n}\}} \right] \right \},
\end{align*}
}where $\hat\beta_{i-1}$ are the first eight components of $\hat\theta_{i-1}$ and $ \hat\tau_{i-1}$ is the last one. The standard deviation of this log-likelihood estimate was estimated using Equation~\eqref{hatsigma}.

As in Section~\ref{sec-syntheticex}, the methods were compared in terms of posterior estimates, acceptance rate, estimated inefficiency factor, and execution time. The inefficiency factor was estimated separately for each component of $\theta$ and then averaged across components. Trace and autocorrelation plots for each component of $\theta$ are provided in Appendix~\ref{appendixE4}.

\begin{table}[h]
\caption{Comparison of summary statistics and execution times, with same initial values for all approaches, between the non-adaptive methods and the APM algorithms using $10^6$ iterations (with 40\% burn-in) across 10 runs. Values are reported as mean $\pm$ standard deviation except for $N$ using median[min,max].}
  \centering
  \resizebox{\textwidth}{!}{
  \begin{tabular}{lccc}
    \hline
     &Naive Non-adaptive& Non-adaptive & APM \\
    \hline
    Optimal N & $22[21,23]$ & $22[21,23]$ & $22[20,24]$\\
    Norm of Posterior Mean  & $3.106 \pm 0.0039$ & $3.106 \pm 0.0041$ & $3.105 \pm 0.0041$ \\
    Norm of Posterior Variance & $0.388 \pm 0.0047$ & $0.388 \pm 0.0038$ & $0.387 \pm 0.0026$ \\
    Acceptance Rate $\widehat{P}$ (\%) & $14.141 \pm 0.2826$ & $14.344 \pm 0.4487$ & $14.245 \pm 0.1468$ \\
    Inefficiency factor $\widehat{IF}$ &$87.545 \pm 2.4921$& $88.639 \pm 9.9387$ & $85.270 \pm 2.7546$ \\
    Execution time & {18h19min42s} $\pm$ 07min34s& {11h10min51s} $\pm$ 30min37s& {11h07min05s} $\pm$ 33min45s \\
    \hline
  \end{tabular}
  }
  \label{re:tab3}
\end{table}

As shown in the third column of Table~\ref{re:tab3}, the Euclidean norms of the posterior means and the Frobenius norms of the covariance matrices match the estimates obtained in \ref{step3} for both non-adaptive methods.
The dichotomic search substantially reduces the computational cost of the non-adaptive method, with the average execution time decreasing from 18h19min42s to 11h10min51s, corresponding to a reduction of approximately $39\%$. The non-adaptive method with dichotomic search and the APM algorithm exhibit similar average execution times, both requiring approximately 11 hours.

For a fair comparison of sampling efficiency, we therefore focus on the non-adaptive method with dichotomic search rather than the naive non-adaptive method, since the latter incurs additional computational cost from the exhaustive search for the optimal number of particles. Using a $40\%$ burn-in (i.e., retaining 600,000 samples), the effective sample size per minute is computed as in Equation~\eqref{ESperMin} in Section~\ref{sec-syntheticex}. With this measure, the adaptive method attains approximately a $4.5\%$ higher sampling efficiency than the non-adaptive method with dichotomic search (10.55 vs. 10.09 effective samples per minute).

The median of the optimal values \( N_{\text{opt}} \) found by the APM, defined as the last number of particles reached in each run, ranged from 20 to 24. This range is  consistent with that obtained via the non-adaptive methods. Additional convergence diagnostics, including autocorrelation and trace plots of \( \theta \), are provided in Appendix~\ref{appendixE4}.

\section{Conclusion}
\label{sec-conclusion}
In this work, we have made three main contributions. 
First, we proposed an adaptive mechanism that overcomes the tuning difficulties commonly encountered in PM algorithms. Second, we established verifiable sufficient conditions ensuring the ergodicity of the resulting adaptive process. Third, we quantitatively illustrated the benefits of our approach through a series of numerical experiments.

Looking ahead, several avenues for future research remain open. In the non-adaptive method, the proposal variance is updated between \ref{step1} and \ref{step3}, whereas in the current APM scheme, the adaptation mechanism relies on a fixed proposal variance, which may limit its flexibility. In the non-adaptive implementation applied to the real example, the choice of the proposal variance is critical; for instance, setting the proposal covariance matrix to the identity leads to failure of the PM algorithm. Allowing both the number of particles \( N \) and the proposal variance \(\Sigma\) to adapt simultaneously would reduce the need for such careful manual tuning of \(\Sigma\). This highlights a promising direction for future research, namely the development of \textit{dual adaptation} schemes in which both \( N \) and the proposal distribution are adapted jointly.

Furthermore, another theoretical extension would be to investigate central limit theorems (CLTs) for the APM in order to provide a more complete theoretical characterization. It is worth noting, however, that most existing CLTs for adaptive chains require stronger conditions. In particular, \cite{atchade_2012} shows that a CLT holds in the polynomial case if the adaptation random variable converges almost surely, which is not guaranteed under the current adaptation mechanism. Therefore, future work could either focus on establishing a CLT for polynomially ergodic adaptive chains under alternative conditions, or on modifying the adaptation scheme to ensure that the standard convergence assumptions are satisfied.  

\section{Acknowledgments}
The authors would like to thank Radu Craiu (University of Toronto) for his contribution to the subject of this article.

\bibliography{bibliography.bib}

@article{beaumont_estimation_2003,
  title = {Estimation of population growth or decline in genetically monitored populations},
  author = {Beaumont, Mark A.},
  journal = {Genetics},
  volume = {164},
  number = {3},
  pages = {1139--1160},
  year = {2003},
  month = jul,
  doi = {10.1093/genetics/164.3.1139}
}

@article{andrieu_pseudo-marginal_2009,
  title = {The pseudo-marginal approach for efficient {Monte Carlo} computations},
  author = {Andrieu, Christophe and Roberts, Gareth O.},
  journal = {The Annals of Statistics},
  volume = {37},
  number = {2},
  pages = {697--725},
  year = {2009},
  month = apr,
  doi = {10.1214/07-AOS574}
}

@article{andrieu_establishing_2016,
  title = {Establishing some order amongst exact approximations of {MCMCs}},
  author = {Andrieu, Christophe and Vihola, Matti},
  journal = {The Annals of Applied Probability},
  volume = {26},
  number = {5},
  pages = {2661--2696},
  year = {2016},
  month = oct,
  doi = {10.1214/15-AAP1158}
}

@article{schmon_large-sample_2021,
  title = {Large-sample asymptotics of the pseudo-marginal method},
  author = {Schmon, S. M. and Deligiannidis, G. and Doucet, A. and Pitt, M. K.},
  journal = {Biometrika},
  volume = {108},
  number = {1},
  pages = {37--51},
  year = {2021},
  month = mar,
  doi = {10.1093/biomet/asaa044}
}

@article{roberts_coupling_2007,
  title = {Coupling and {E}rgodicity of {Adaptive} {Markov} {Chain} {Monte Carlo} {A}lgorithms},
  author = {Roberts, Gareth O. and Rosenthal, Jeffrey S.},
  journal = {Journal of Applied Probability},
  volume = {44},
  number = {2},
  pages = {458--475},
  year = {2007},
  month = jun,
  doi = {10.1239/jap/1183667414}
}

@article{andrieu_convergence_2015,
  title = {Convergence properties of pseudo-marginal {Markov} {chain} {Monte Carlo} algorithms},
  author = {Andrieu, Christophe and Vihola, Matti},
  journal = {The Annals of Applied Probability},
  volume = {25},
  number = {2},
  pages = {1030--1077},
  year = {2015},
  month = apr,
  doi = {10.1214/14-AAP1022}
}

@article{atchade_limit_2010,
  title = {Limit theorems for some adaptive {MCMC} algorithms with subgeometric kernels},
  author = {Atchadé, Yves and Fort, Gersende},
  journal = {Bernoulli},
  volume = {16},
  number = {1},
  pages = {116--154},
  year = {2010},
  month = feb,
  doi = {10.3150/09-BEJ199}
}

@article{zeger_generalized_1991,
  title = {Generalized {L}inear {M}odels with {R}andom {E}ffects; a {Gibbs} {S}ampling {A}pproach},
  author = {Zeger, Scott L. and Karim, M. Rezaul},
  journal = {Journal of the American Statistical Association},
  volume = {86},
  number = {413},
  pages = {79--86},
  year = {1991},
  month = mar,
  doi = {10.1080/01621459.1991.10475006}
}

@article{pitt_properties_2012,
  title = {On some properties of {Markov} chain {Monte Carlo} simulation methods based on the particle filter},
  author = {Pitt, Michael K. and Silva, Ralph dos Santos and Giordani, Paolo and Kohn, Robert},
  journal = {Journal of Econometrics},
  volume = {171},
  number = {2},
  pages = {134--151},
  year = {2012},
  month = dec,
  doi = {10.1016/j.jeconom.2012.06.004}
}

@article{doucet_efficient_2015,
  title = {Efficient implementation of {Markov} chain {Monte Carlo} when using an unbiased likelihood estimator},
  author = {Doucet, Arnaud and Pitt, Michael K. and Deligiannidis, George and Kohn, Robert},
  journal = {Biometrika},
  volume = {102},
  number = {2},
  pages = {295--313},
  year = {2015},
  month = jun,
  doi = {10.1093/biomet/asu075}
}

@article{sherlock_efficiency_2015,
  title = {On the efficiency of pseudo-marginal random walk {Metropolis} algorithms},
  author = {Sherlock, Chris and Thiery, Alexandre H. and Roberts, Gareth O. and Rosenthal, Jeffrey S.},
  journal = {The Annals of Statistics},
  volume = {43},
  number = {1},
  pages = {238--275},
  year = {2015},
  month = feb,
  doi = {10.1214/14-AOS1278}
}

@article{haggstrom_variance_2007,
  title = {On {V}ariance {C}onditions for {Markov} {C}hain {CLTs}},
  author = {Häggström, Olle and Rosenthal, Jeffrey S.},
  journal = {Electronic Communications in Probability},
  volume = {12},
  year = {2007},
  doi = {10.1214/ECP.v12-1336}
}

@article{berard_lognormal_2014,
  title = {A lognormal central limit theorem for particle approximations of normalizing constants},
  author = {Bérard, Jean and Del Moral, Pierre and Doucet, Arnaud},
  journal = {Electronic Journal of Probability},
  volume = {19},
  number = {94},
  pages = {1--28},
  year = {2014},
  doi = {10.1214/EJP.v19-3428}
}

@article{deligiannidis_correlated_2018,
  title = {The {C}orrelated {P}seudomarginal {M}ethod},
  author = {Deligiannidis, George and Doucet, Arnaud and Pitt, Michael K.},
  journal = {Journal of the Royal Statistical Society: Series B (Statistical Methodology)},
  volume = {80},
  number = {5},
  pages = {839--870},
  year = {2018},
  month = nov,
  doi = {10.1111/rssb.12280}
}

@article{roberts_examples_2009,
  title = {Examples of {A}daptive {MCMC}},
  author = {Roberts, Gareth O. and Rosenthal, Jeffrey S.},
  journal = {Journal of Computational and Graphical Statistics},
  volume = {18},
  number = {2},
  pages = {349--367},
  year = {2009},
  month = jan,
  doi = {10.1198/jcgs.2009.06134}
}

@article{gelman_weak_1997,
  title = {Weak convergence and optimal scaling of random walk {Metropolis} algorithms},
  author = {Gelman, Andrew and Gilks, Walter R. and Roberts, Gareth O.},
  journal = {The Annals of Applied Probability},
  volume = {7},
  number = {1},
  pages = {110--120},
  year = {1997},
  month = feb,
  doi = {10.1214/aoap/1034625254}
}

@book{robert_monte_2004,
  title = {{Monte Carlo} {S}tatistical {M}ethods},
  author = {Robert, Christian P. and Casella, George},
  publisher = {Springer},
  address = {New York},
  edition = {2},
  series = {Springer Texts in Statistics},
  year = {2004},
  isbn = {978-0-387-21239-5},
  doi = {10.1007/978-1-4757-4145-2}
}

@book{Vaart_1998,
  title = {Asymptotic {S}tatistics},
  author = {van der Vaart, A. W.},
  publisher = {Cambridge University Press},
  address = {Cambridge},
  year = {1998},
  doi = {10.1017/CBO9780511802256}
}

@article{flegal_batch_2010,
  title = {Batch means and spectral variance estimators in {Markov} chain {Monte Carlo}},
  author = {Flegal, James M. and Jones, Galin L.},
  journal = {The Annals of Statistics},
  volume = {38},
  number = {2},
  pages = {1034--1070},
  year = {2010},
  month = apr,
  doi = {10.1214/09-AOS735}
}

@article{atchade_2012,
  title = {Limit theorems for some adaptive {MCMC} algorithms with subgeometric kernels: Part {II}},
  author = {Atchadé, Yves F. and Fort, Gersende},
  journal = {Bernoulli},
  volume = {18},
  number = {3},
  pages = {975--1001},
  year = {2012},
  doi = {10.3150/11-BEJ360}
}

@article{jarner_geometric_2000,
	title = {Geometric ergodicity of {M}etropolis algorithms},
	volume = {85},
	copyright = {https://www.elsevier.com/tdm/userlicense/1.0/},
	issn = {03044149},
	doi = {10.1016/S0304-4149(99)00082-4},
	language = {en},
	number = {2},
	journal = {Stochastic Processes and their Applications},
	author = {Jarner, Søren Fiig and Hansen, Ernst},
	month = feb,
	year = {2000},
	pages = {341--361},
}

@book{shaked_stochastic_2007,
	address = {New York, NY},
	series = {Springer {Series} in {Statistics}},
	title = {Stochastic {Orders}},
	copyright = {http://www.springer.com/tdm},
	isbn = {978-0-387-32915-4},
    language = {en},
	publisher = {Springer New York},
	author = {Shaked, Moshe and Shanthikumar, J. George},
	year = {2007},
}

@article{roberts_geometric_1996,
	title = {Geometric convergence and central limit theorems for multidimensional {H}astings and {M}etropolis algorithms},
	volume = {83},
	issn = {0006-3444, 1464-3510},
	language = {en},
	number = {1},
	journal = {Biometrika},
	author = {Roberts, G.},
	month = mar,
	year = {1996},
	pages = {95--110},
}

@incollection{Geweke_1992,
    author = {Geweke, John},
    editor = {Bernardo, J M and Berger, J O and Dawid, P and Smith, A F M},
    title = {Evaluating the {A}ccuracy of {S}ampling-{B}ased {A}pproaches to the {C}alculation of {P}osterior {M}oments},
    booktitle = {Bayesian Statistics 4: Proceedings of the Fourth Valencia International Meeting},
    publisher = {Oxford University Press},
    year = {1992},
    month = aug,
    isbn = {9780198522669}
}

\newpage
\appendix

\renewcommand{\thetheorem}{S.\arabic{theorem}}
\renewcommand{\thelemma}{S.\arabic{lemma}}
\renewcommand{\theproposition}{S.\arabic{proposition}}
\renewcommand{\thecorollary}{S.\arabic{corollary}}
\renewcommand{\thedefinition}{S.\arabic{definition}}
\renewcommand{\theremark}{S.\arabic{remark}}
\renewcommand{\thehypothesis}{S.\arabic{hypothesis}}
\setcounter{equation}{0}
\renewcommand{\theequation}{E.\arabic{equation}}

\setcounter{theorem}{0}
\setcounter{lemma}{0}
\setcounter{proposition}{0}
\setcounter{corollary}{0}
\setcounter{definition}{0}
\setcounter{remark}{0}
\setcounter{example}{0}
\setcounter{hypothesis}{0}

\section{Theoretical Background}
\label{appendixA}
\begin{definition}[Simultaneous Strong Aperiodic Geometrical Ergodicity]  
A family $\{P_\gamma\}_{\gamma \in \mathcal{Y}}$ of Markov kernels defined on a space state $(\mathcal{X}, \mathcal{B}(\mathcal{X}))$ is \textit{simultaneously strongly aperiodically and geometrically ergodic} if there exist $C \in \mathcal{B}(\mathcal{X})$, $V: \mathcal{X} \to [1,\infty)$, $\delta > 0$, $\lambda < 1$, and $b < \infty$, such that $\sup_C V < \infty$, and

\begin{enumerate}
    \item[(i)] for each $\gamma \in \mathcal{Y}$, there exists a probability measure $\nu_\gamma(\cdot)$ on $C$ with $P_\gamma(x,\cdot) \geq \delta\nu_\gamma(\cdot)$ for all $x \in C$, and
    
    \item[(ii)] $P_\gamma V(x) \leq \lambda V(x) + b \ind_C(x)$ for all $x \in \mathcal{X}$.
\end{enumerate}
\label{def2}
\end{definition}

\begin{hypothesis}[Assumption~A1 of \cite{atchade_limit_2010}]
For a measurable set $C \subseteq \mathcal{X}$, denote by
\[
\tau_C=\inf\{n\geq1:X_n\in C\}
\]
the return time to $C$. Assumption~A1 consists of the following three conditions:
\begin{enumerate}
    \item[(i)] There exists a non-decreasing function
    $r:\mathbb{N}\to(0,+\infty)$ such that
    \[
    \sum_{n=1}^{\infty}\frac{1}{r(n)}<\infty,
    \text{ and }
    \sup_{n \geq 0}\sup_{(x,\gamma)\in C\times\mathcal{Y}}
    \E_{x,\gamma}^{(n)}[r(\tau_C)]
    <\infty.
    \]

    \item[(ii)] There exists a probability measure $\pi$ such that
    \[
    \lim_{n\to\infty}
    \sup_{x\in\mathcal{X}}
    V^{-1}(x)
    \sup_{\gamma\in\mathcal{Y}}
    \|P_\gamma^n(x,\cdot)-\pi\|_{\mathrm{TV}}
    =0.
    \]

    \item[(iii)] The inequality
    \(
    \sup_{\gamma}P_\gamma V\leq V
    \)
    holds on $C^c$, and
    \(
    \sup_{(x,\gamma)\in C\times\mathcal{Y}}
    \{P_\gamma V(x)+V(x)\}
    <\infty.
    \)
\end{enumerate}
\label{SM-def2}
\end{hypothesis}

\begin{theorem}[Theorem 2.1, \cite{atchade_limit_2010}]
    Assuming the \textit{diminishing adaptation} condition and that the conditions in Assumption~\ref{SM-def2} are satisfied. Then,
\[
\lim_{n\to+\infty} \sup_{\{f : \|f\|_\infty \leq 1\}} \left|\E[f(X_n) - \pi(f)]\right| = 0.
\]
\label{appendix:thm4}
\end{theorem}

\begin{remark}
Theorem~\ref{appendix:thm4} assumes that all kernels share a common invariant distribution. This assumption is not satisfied by PM kernels, since the invariant distribution of the augmented chain depends on the number of particles. In Section~\ref{sec-APMproof}, we have established a generalized version of Assumption~A1 and adapted the proof argument of Theorem~\ref{appendix:thm4}  to the PM setting.
\end{remark}

\begin{corollary}[Corollary A.2, \cite{atchade_limit_2010}]
Let \( P \) be a \(\phi\)-irreducible and aperiodic Markov kernel on \( (\mathcal{X}, \mathcal{B}(\mathcal{X})) \). Suppose there exist constants \( b, c > 0 \), a measurable set \( C \), an unbounded measurable function \( V \colon \mathcal{X} \to [1, \infty) \), and \( 0 < \alpha \leq 1 \) such that
\[
PV(x) \leq V(x) - c V^{1-\alpha}(x) + b\, \ind_C(x).
\]
If, in addition, all level sets of \( V \) are 1-small, then there exist a level set \( B \subset \mathcal{X} \), constants \( \varepsilon_B, c_B > 0 \), and a probability measure \( \nu_B \) such that
\begin{align*}
P(x,\cdot) &\geq \ind_B(x)\, \varepsilon_B \nu_B(\cdot), \text{ and } \quad 
PV(x) \leq V(x) - c_B V^{1-\alpha}(x) + b\, \ind_B(x),
\end{align*}
with $\sup_B V < \infty$, $\nu_B(B) > 0$, and $c_B \inf_{B^c} V^{1-\alpha} \geq b.$
\label{appendix:Corollary1}
\end{corollary}

\begin{proposition}[Proposition A.1, \cite{atchade_limit_2010}]
Let $P$ be a $\phi$-irreducible and aperiodic transition kernel on $(\mathcal{X}, \mathcal{B}(\mathcal{X}))$.
\begin{enumerate}
    \item Assume that there exist a probability measure $\nu$ on $\mathcal{X}$, positive constants $\varepsilon, b, c$, a measurable set $\mathcal{C}$, a measurable function $V:\mathcal{X} \to [1,+\infty)$ and $0 < \alpha \leq 1$ such that
    \begin{align*}
    P(x,\cdot) &\geq \varepsilon \ind_\mathcal{C}(x)\nu(\cdot), \text{ and }
    \qquad PV \leq V - cV^{1-\alpha} + b\ind_\mathcal{C}. 
    \end{align*}
    Then $P$ possesses an invariant probability measure $\pi$ and $\pi(V^{1-\alpha}) < +\infty$.
    \item Assume, in addition, that $c \inf_{\mathcal{C}^c} V^{1-\alpha} \geq b$, $\sup_\mathcal{C} V < +\infty$ and $\nu(\mathcal{C}) > 0$. Then there exists a constant $C$ depending on $\sup_\mathcal{C} V$, $\nu(\mathcal{C})$ and $\varepsilon, \alpha, b, c$, such that for any $0 \leq \beta \leq 1 - \alpha$ and $1 \leq \kappa \leq \alpha^{-1}(1 - \beta)$,
    \begin{equation}
    (n+1)^{\kappa-1}\|P^n(x,\cdot) - \pi(\cdot)\|_{V^\beta} \leq CV^{\beta+\alpha\kappa}(x). \label{eq:convergence_rate}
    \end{equation}
\end{enumerate}
\label{appendix:proposition1}
\end{proposition}

\begin{theorem}[Theorem 38, \cite{andrieu_convergence_2015}]
Let \( P_N \) denote a PM kernel with distributions \( Q_{N,\theta}(\mathrm{d}w) \) satisfying the moment condition
\begin{align}
    M_{W_N} := \underset{\theta \in \Theta}{\mathrm{ess\,sup}} \int (w^{-\alpha'} \vee w^{\beta'})\, Q_{N,\theta}(dw) < \infty,
\label{appendix:moments}
\end{align}
for some constants \( \alpha' > 0 \) and \( \beta' > 1 \). Assume that the marginal algorithm is a RWM with invariant density \( \pi \) and proposal density \( q \) satisfying Assumption~\ref{hyp5}.

\noindent
Define the function \( V: \Theta \times \mathcal{W} \to [1,\infty) \) as
\[
V(\theta,w) := c_\pi \pi^{-\eta}(\theta)\left(w^{-\alpha} \vee w^{\beta}\right), 
\]
with $ c_\pi := \sup_{\vartheta \in \Theta} \pi(\vartheta)$, for constants \( \eta \in (0, \alpha' \wedge 1 \wedge \beta'-1) \), \( \alpha \in (\eta, \alpha'] \), and \( \beta \in (1, \beta' - \eta) \).

Then, there exist constants \( \overline{w}, M, b \in [1, \infty) \), \( \underline{w} \in (0,1] \), and \( \delta_V > 0 \) such that
\[
P_N V(\theta,w) \leq 
\begin{cases}
V(\theta,w) - \delta_V V^{(\beta-1)/\beta}(\theta,w), \quad 
&\text{if } (\theta,w) \notin C, \\
b, \quad 
&\text{if } (\theta,w) \in C,
\end{cases}
\]
where the set \( C \subset \Theta \times \mathcal{W} \) is defined by
\[
C := \left\{ (\theta, w) \in \Theta \times \mathcal{W} : \|\theta\| \leq M, \, w \in [\underline{w}, \overline{w}] \right\}.
\]
\label{appendix:thm3}
\end{theorem}
\begin{theorem}[Theorem 2.1, \cite{jarner_geometric_2000}]
The RWM algorithm satisfying Assumptions~\ref{hyp1} and~\ref{hyp5}  is $\mu_{\text{Leb}}$-irreducible and aperiodic.
\label{appendix:thm1}
\end{theorem}

\begin{theorem}[Theorem 1, \cite{andrieu_pseudo-marginal_2009}]
Let $P$ be a $\phi$-irreducible and aperiodic MH chain with invariant distribution $\pi$. Then, for any $N \geq 1$ such that $\rho_N(\theta, w) > 0$ for all $(\theta, w)$ (as defined in Equation~\eqref{varrho_repar}), the PM kernel $P_N$ is also $\phi$-irreducible and aperiodic.
\label{appendix:thm2}
\end{theorem}

\begin{remark}
In Theorem~\ref{appendix:thm2}, the authors assume that the weights are well-defined using the concept of measure domination. This assumption is not required in our setting, as the weights are explicitly defined as $W_N(\theta) = \widehat{p}_{N,T}(y| \theta, U) / p_T(y| \theta)$.
\end{remark}

\begin{theorem}[Theorem 7, \cite{andrieu_establishing_2016}]
If there exists a probability space with random variables \( X' \) and \( Y' \) having the same distributions as \( X \) and \( Y\), respectively, and such that
\[\E[Y'|X'] = X' \quad \text{a.s.},\]
then \( X \preceq_{cx} Y \).
    \label{appendix:thm5}
\end{theorem}

\begin{corollary}[Corollary 3.A.22, \cite{shaked_stochastic_2007}]
Let \( X_1 \) and \( X_2 \) be a pair of independent random variables, and let \( Y_1 \) and \( Y_2 \) be another pair of independent random variables. If \( X_i {\preceq_{cx}} Y_i \), for \( i = 1, 2 \), then
\[X_1 X_2 \preceq_{cx} Y_1 Y_2.\]
\label{appendix:corollary3}
\end{corollary}

\section{Proofs of the main results}
\label{appendixB}

\subsection{Proof of the Diminishing Adaptation Condition}
\label{sub1appendixB}
\begin{proof}[Proof of Proposition~\ref{proposition1}]
Let $A \in \mathcal{B}(\Theta\times\mathcal{W})$, $(\theta, w )\in \Theta\times\mathcal{W}$. The difference between two consecutive kernels in~\eqref{equa5} can be expressed as:
{\small
\begin{align*}
P_{N_{\ell+1}}(\theta,w ; A)-P_{N_{\ell}}(\theta,w ; A) &=\delta_{(\theta,w)}(A) \int_{\Theta \times \mathcal{W}} q(\vartheta|\theta) \min\left\{1,  r(\theta, \vartheta) \frac{z}{w}\right\} \\
&\quad \times\left(\mathcal{Q}_{N_{\ell},\vartheta}(\mathrm{d}z) - \mathcal{Q}_{N_{\ell+1},\vartheta}(\mathrm{d}z) \right) \mathrm{d} \vartheta \\
&\quad + \int_A q(\vartheta|\theta) \min\left\{1,  r(\theta, \vartheta) \frac{z}{w}\right\} \left(\mathcal{Q}_{N_{\ell+1},\vartheta}(\mathrm{d}z)- \mathcal{Q}_{N_{\ell},\vartheta}(\mathrm{d}z) \right) \mathrm{d} \vartheta.
\end{align*}
}Taking the absolute value, we get:
{\small
\begin{align*}
    \left|P_{N_{\ell+1}}(\theta,w ; A)-P_{N_{\ell}}(\theta,w ; A)\right| &\leq 2 \int_{\Theta \times \mathcal{W}} q(\vartheta|\theta) \min\left\{1,  r(\theta, \vartheta) \frac{z}{w}\right\} 
    \left| \mathcal{Q}_{N_{\ell+1},\vartheta}(\mathrm{d}z) 
    - \mathcal{Q}_{N_{\ell},\vartheta}(\mathrm{d}z) \right| \mathrm{d} \vartheta.
\end{align*}
}
The total variation norm of the difference between successive kernels is then bounded by:
{\small
\begin{align*}
    \|P_{N_{\ell+1}}&(\theta,w ; .)-P_{N_{\ell}}(\theta,w ; .)\|_\TV \leq 4 \int_{\Theta \times \mathcal{W}} q(\vartheta|\theta)  |\mathcal{Q}_{N_{\ell+1},\vartheta}(\mathrm{d}z)- \mathcal{Q}_{N_{\ell},\vartheta}(\mathrm{d}z)| \mathrm{d} \vartheta.
\end{align*}
}Let $\epsilon>0$, taking the supremum over $\theta$, $D_\ell$ is bounded by
{
\begin{align*}
    D_\ell \leq 4 \, \underset{\theta}{\sup} \int_{\Theta \times \mathcal{W}} q(\vartheta|\theta)  |\mathcal{Q}_{N_{\ell+1},\vartheta}(\mathrm{d}z)- \mathcal{Q}_{N_{\ell},\vartheta}(\mathrm{d}z)| \mathrm{d} \vartheta.
\end{align*}}
If \(|N_{\ell+1}-N_\ell| <a\), where $a$ is the step size defined in Section~\ref{sec-APM} then $D_\ell < \epsilon$ ($D_\ell=0$). We can conclude that the event $\{D_\ell \geq \epsilon\} \subseteq \{|N_{\ell+1}-N_\ell| \geq a\}$ and hence, 
$$\Prob(D_\ell \geq \epsilon) \leq \Prob(|N_{\ell+1}-N_\ell| \geq a).$$
Furthermore, the probability of the number of particles changing by at least $a$ is
\begin{align*}
    \Prob(|N_{\ell+1}-N_{\ell}| \geq a) &= \E[\Prob(|N_{\ell+1}-N_{\ell}| \geq a|\mathcal{G}_\ell)]\\
    &= \E[\Prob(N_{\ell+1}=N_\ell+a|\mathcal{G}_\ell) + \Prob(N_{\ell+1}=N_\ell-a|\mathcal{G}_\ell)]\\
    &= \E[p_j \ind{(\hat{\sigma}_\ell \geq {\sigma}_{\text{opt}}
    +\sigma_\text{e})} + p_j \ind{(\hat{\sigma}_\ell < {\sigma}_{\text{opt}}-\sigma_\text{e}})]
    \leq 2p_j.
\end{align*}Since $p_j$, the probability of adapting the number of particles, converges to $0$ as $j \rightarrow \infty$, then so does $\ell=Kj$. We conclude that for all $\epsilon >0$, $\Prob(D_\ell \geq \epsilon)$ converges to $0$ as $\ell \rightarrow \infty$. 
\end{proof}

\subsection{Simultaneous Minorization and Polynomial Drift implies generalized Assumption~\ref{SM-def2}}
\label{prooflemma1}
\begin{proof}[Proof of Lemma~\ref{lemma1}]
First, we start by proving the condition (2) of Lemma~\ref{lemma1}. From Proposition~\ref{prop4}, for every $N\geq N_0$, the augmented PM kernel $\tilde{P}_N$ is $\mu_{\text{Leb}}$-irreducible and aperiodic. Then, by Proposition~\ref{prop2}, there exist $b>0$, $\delta_V>0$, a measurable set $K$, an unbounded measurable function $V : \Theta\times \mathcal{W} \to [1, \infty)$ and $\beta > 1 $ such that 
$$ \tilde{P}_N V \leq V - \delta_V V^{1-\frac{1}{\beta}}+ b \ind_K.$$
In addition, Lemma~\ref{lemma2} shows that every level set of $V$ is a compact of $\Theta\times \mathcal{W}$ and Proposition~\ref{prop3} shows that for any $N \geq N_0$ every compact of $\Theta\times \mathcal{W}$ is  1-small for $\tilde{P}_N$. For $N \geq N_0$, we deduce that all level sets of $V$ are 1-small for $\tilde{P}_N$.
Hence all assumptions of Corollary~\ref{appendix:Corollary1} are satisfied and there exist a level set $B$ of $V$, constants $\epsilon_B>0$, $\delta_B>0$ and a probability measure $\nu_B$ such that 
    $$ \tilde{P}_N(\theta,w ; \cdot) \geq \ind_B(\theta,w) \epsilon_B  \nu_B(\cdot),$$ 
    $$\tilde{P}_N V \leq V - \delta_B V^{1-\frac{1}{\beta}}+ b \ind_B,$$
     $\sup_B V <\infty$, $\nu_B(B)>0$ and $\delta_B \inf_{B^\complement} V^{1-\frac{1}{\beta}} \geq b$.
    Now, we apply Proposition~\ref{appendix:proposition1} which gives the existence of a constant $D$ depending on  $\sup_B V$, $\nu_B(B)$, $\epsilon_B$, $\beta$, $b$ and $\delta_B$ such that for any $N \geq N_0$, $0 \leq \iota \leq 1 -\frac{1}{\beta}$ and $ 1\leq \kappa \leq \beta (1 - \iota)$,
    $$(n+1)^{\kappa-1}\,\|\tilde{P}_N^n(\theta,w ; \cdot) - \tilde{\pi}_N(\cdot)\|_{V^\iota} \leq D\,V^{\iota+\frac{\kappa}{\beta}}(\theta,w).$$
    The constant $D$ does not depend on $N$, thus by choosing $\iota=0$ and $\kappa= \beta$, 
    $$\sup_{N\geq N_0}\|\tilde{P}_N^n(\theta,w ; \cdot) - \tilde{\pi}_N(\cdot)\|_\TV \leq D\,V(\theta,w)\frac{1}{(n+1)^{\beta-1}}.$$
    and thus
     $$\lim_{n \to \infty}\sup_{(\theta,w) \in \Theta \times \mathcal{W}} V^{-1}(\theta,w) \sup_{N\geq N_0}\|\tilde{P}_N^n(\theta,w ; \cdot) - \tilde{\pi}_N(\cdot)\|_\TV =0.$$
     The proofs of (1) and (3) of Lemma~\ref{lemma1} are the same as in the proof of Corollary 2.2 in \cite{atchade_limit_2010} (see Section 4.3, \cite{atchade_limit_2010}).
\end{proof}

\subsection{Uniform-in-\texorpdfstring{$N$}{ } Polynomial Drift}
\label{sub2appendixB}
\begin{proof}[Proof of Proposition~\ref{prop2}]
We first show that the perturbed weights satisfy the same convex-order relation as the original weights. Since
\[
W_{N,\theta}\preceq_{cx}W_{N_0,\theta},
\]
by Assumption~\ref{hyp2} and since $U\sim\Unif(1-\delta,1+\delta)$ is independent of $W_{N,\theta}$, Theorem~3.A.12(b) of \cite{shaked_stochastic_2007} yields
\begin{equation}
\widetilde W_{N,\theta}
=
W_{N,\theta}U
\preceq_{cx}
W_{N_0,\theta}U
=
\widetilde W_{N_0,\theta},
\qquad
N\ge N_0.
\label{eq:cx-tilde}
\end{equation}

We next verify that the perturbed weights possess uniformly bounded moments. Indeed,
\begin{align}
M_{\widetilde W_{N_0}}
&=\esssup  \, \E\left[\widetilde W_{N_0,\theta}^{\alpha_0}\vee\widetilde W_{N_0,\theta}^{\beta_0}\right]\nonumber\\
&\leq\esssup \,\E\left[(W_{N_0,\theta}U)^{-\alpha_0}\right]+\esssup \,\E\left[(W_{N_0,\theta}U)^{\beta_0}\right]\nonumber\\
&\leq(1-\delta)^{-\alpha_0} \esssup \, \E\left[W_{N_0,\theta}^{-\alpha_0}\right]+(1+\delta)^{\beta_0}\esssup\,\E\left[
W_{N_0,\theta}^{\beta_0}\right].
\label{eq:moment-bound}
\end{align}

Assumption~\ref{hyp3} implies that the right-hand side of
\eqref{eq:moment-bound} is finite, and therefore
\(M_{\widetilde W_{N_0}}<\infty.\)

Combining \eqref{eq:cx-tilde} and \eqref{eq:moment-bound} and using that
$x\mapsto x^{-\alpha_0}\vee x^{\beta_0}$ is convex gives
\begin{equation}
M_{\widetilde W_N}=\esssup \,\E\left[\widetilde W_{N,\theta}^{-\alpha_0}\vee\widetilde W_{N,\theta}^{\beta_0}\right]
\leq M_{\widetilde{W}_{N_0}}
<\infty,
\qquad
\forall N\ge N_0.
\label{eq:uniform-moments}
\end{equation}

By Assumptions~\ref{hyp1} and~\ref{hyp5}, together with
\eqref{eq:uniform-moments}, Theorem~38 of
\cite{andrieu_convergence_2015} applies to the perturbed PM kernel $\tilde{P}_N$ for all $N \geq N_0$.
Hence, for all $N\geq N_0$, there exists the drift function
\[V(\theta,w)=c_\pi^\eta\pi^{-\eta}(\theta)\left(w^{-\alpha}\vee w^\beta\right),
\]
where
\[
\eta \in (0, 1\wedge \alpha_{0} \wedge \{\beta_{0}-1\}),
\qquad
\alpha \in (\eta,\alpha_0],
\qquad
\beta \in (1,\beta_0-\eta],
\]
which does not depend on $N$ but only on $N_0$ through $\alpha_0$ and $\beta_0$.

We thus obtain constants $\barw,M,b'\in[1,\infty)$,$\basw\in(0,1]$,
and $\delta_V>0$, independent of $N$, such that
\[\tilde{P}_N V(\theta,w)
\leq
\begin{cases}
V(\theta,w)-\delta_VV^{(\beta-1)/\beta}(\theta,w),
&
(\theta,w)\notin K,
\\[1ex]
b',
&
(\theta,w)\in K,
\end{cases}
\]
where
\[
K=
\left\{
(\theta,w):
\|\theta\|\leq M,\;
w\in[\basw,\barw]
\right\}.
\]
We now rewrite the drift inequality in its usual form. If $(\theta,w)\notin K$, then
\begin{align*}
\tilde{P}_N V(\theta,w)
&\leq V(\theta,w)-\delta_VV^{(\beta-1)/\beta}(\theta,w)\\
&=V(\theta,w)-\delta_VV^{(\beta-1)/\beta}(\theta,w)
+b'\ind_K(\theta,w),
\end{align*}
since $\ind_K(\theta,w)=0$. Now suppose that $(\theta,w)\in K$. Then
\begin{align*}
\tilde{P}_N V(\theta,w)
&\leq b'\\
&=V(\theta,w)-\delta_VV^{(\beta-1)/\beta}(\theta,w)
-\left(V(\theta,w)-\delta_VV^{(\beta-1)/\beta}(\theta,w)\right)
+b'\\
&\leq
V(\theta,w)-\delta_VV^{(\beta-1)/\beta}(\theta,w)
+\delta_VV^{(\beta-1)/\beta}(\theta,w)
+b',
\end{align*}
where the last inequality follows from the fact that $V(\theta,w)\geq 1$. Since $V$ is continuous, it is bounded on the compact set $K$. Therefore,
\[
V(\theta,w)\leq c_K,
\qquad
\forall (\theta,w)\in K,
\]
for some constant $c_K>0$. Consequently,
\begin{align*}
\tilde{P}_N V(\theta,w)
&\leq
V(\theta,w)-\delta_VV^{(\beta-1)/\beta}(\theta,w)
+\delta_VK^{(\beta-1)/\beta}+b'\\
&=
V(\theta,w)-\delta_VV^{(\beta-1)/\beta}(\theta,w)
+b\,\ind_K(\theta,w),
\end{align*}
where \(b=\delta_V c_K^{(\beta-1)/\beta}+b'.\)
We therefore obtain, for every $(\theta,w)\in\Theta\times\mathcal W$,
\[
\tilde{P}_N V(\theta,w)
\leq
V(\theta,w)-\delta_VV^{(\beta-1)/\beta}(\theta,w)
+b\,\ind_K(\theta,w),
\]
which is the desired drift condition.
\end{proof}

\subsection{Compactness of Level Sets of \texorpdfstring{$V$}{ }}
\label{sub3appendixB}
\begin{proof}[Proof of Lemma~\ref{lemma2}]
Let \( b > 1 \). We claim that the set \( B \) is bounded. Suppose, for contradiction, that for every \( M' > 0 \), there exists \( (\theta, w) \) such that \( \|\theta\| > M' \) and \( V(\theta, w) \leq b \). Then, since $w^{-\alpha} \vee w^{\beta} = 1$,
\begin{align*}
    V(\theta, w) \leq b
    &\Longleftrightarrow
    \pi^{-\eta}(\theta) (w^{-\alpha} \vee w^{\beta}) \leq b \, c_\pi^{-\eta}, \\
    &\Longleftrightarrow
    \pi^{-\eta}(\theta) \leq b \, c_\pi^{-\eta} (w^{-\alpha} \vee w^{\beta})^{-1}, \\
    &\Longleftrightarrow
    \pi(\theta) \geq b^{-\frac{1}{\eta}} c_\pi \left( w^{-\alpha} \vee w^{\beta} \right)^{\frac{1}{\eta}}, \\
    &\Longrightarrow
    \pi(\theta) \geq b^{-\frac{1}{\eta}} c_\pi.
\end{align*}
This inequality implies that \( \pi(\theta) \) is bounded away from zero as \( \|\theta\| \to \infty \), which contradicts Assumption~\ref{hyp1}, under which \( \pi(\theta) \to 0 \) as \( \|\theta\| \to \infty \). We conclude that there exists a constant \( M_1 > 0 \) such that for all \( (\theta, w) \in B \), it holds that \( \|\theta\| \leq M_1 \).

Similarly, we suppose, again by contradiction, that for every \( M'' > 0 \), there exists \( (\theta, w) \in B \) such that \( w > M'' \). Then,
\begin{align*}
    V(\theta, w) \leq b
    &\Longrightarrow
    w^{\beta} \leq b c_\pi^{-\eta} \pi^{\eta}(\theta), \\
    &\Longrightarrow
    w^{\beta} \leq b c_\pi^{-\eta} c_\pi^{\eta}, \quad \text{since } \pi(\theta) \leq c_\pi, \\
    &\Longrightarrow
    w \leq b^{1/\beta}.
\end{align*}
This contradicts the assumption that \( w \) can be made arbitrarily large while satisfying \( V(\theta, w) \leq b \). Therefore, there exists \( M_2 > 0 \) such that \( w \leq M_2 \) for all \( (\theta, w) \in B \).

In conclusion, both \( \theta \) and \( w \) are bounded on the level set \( B \) of $V$, and hence there exists \( M > 0 \) such that \( \|(\theta, w)\| \leq M \).

The set $B$ is also closed. Let $(\theta_\ell, w_\ell)$ be a sequence in $B^\nats$ that converges to $(\theta, w)$ as $\ell \to \infty$. The function $\tilde{V}(\theta,w) = \pi^{-\eta}(\theta) (w^{-\alpha} \vee w^{\beta})$ is continuous as the product of two continuous functions. Therefore,
\[
\tilde{V}(\theta_\ell, w_\ell) \to \tilde{V}(\theta, w) \quad \text{as } \ell \to \infty.
\]
Since $\tilde{V}(\theta_\ell, w_\ell) \leq b c_\pi^{-\eta}$ for all $\ell$, it follows that $\tilde{V}(\theta, w) \leq b c_\pi^{-\eta}$, implying $(\theta, w) \in B$. Hence, $B$ is compact.
\end{proof}

\subsection{Proof That Compacts are Uniformly-in-\texorpdfstring{$N$}{} 1-Small for \texorpdfstring{$\tilde{P}_N$}{ }}
\label{sub4appendixB}
\begin{proof}[Proof of Proposition~\ref{prop3}]
Let $C$ be a compact subset of $\mathbb{R}^d$. First, we show that under Assumptions~\ref{hyp3}, \ref{hyp2}, \ref{hyp7}, and \ref{hyp8}, we have that for all $\epsilon > 0$,
\[
\sup_{\theta \in C} \mathbb{P}(|W_{N,\theta}-1|>\epsilon) \xrightarrow[N \to \infty]{} 0.
\]

Let $\theta \in C$. By Assumptions~\ref{hyp3} and~\ref{hyp2}, there exists $\beta>1$ such that
\[
\sup_{N\ge N_0}\E[W_{N,\theta}^{\beta}]
=
\E[W_{N_0,\theta}^{\beta}]
<\infty.
\]
Since $w^\beta/w\to\infty$ as $w\to\infty$, the family
$\{W_{N,\theta}:N\ge N_0\}$ is uniformly integrable by the
de la Vallée--Poussin criterion.

Assumption~\ref{hyp7} implies that
$W_{N,\theta}\xrightarrow{\mathbb{P}}1$, therefore 
\[
\E|W_{N,\theta}-1|
\xrightarrow[N\to\infty]{}0,
\]
since convergence in probability together with uniform integrability implies
$L^1$ convergence.

Define
\[
f_N(\theta)=\E|W_{N,\theta}-1|.
\]
By Assumption~\ref{hyp8}, the mapping
$\theta\mapsto f_N(\theta)$ is continuous on $C$. Moreover, since
$x\mapsto |x-1|$ is convex, Assumption~\ref{hyp2} implies that
$\{f_N(\theta)\}_{N\ge N_0}$ is decreasing in $N$ for every
$\theta\in C$. Hence, Dini's theorem yields
\[
\sup_{\theta\in C}\E|W_{N,\theta}-1|
\xrightarrow[N\to\infty]{}0.
\]

Finally, by Markov's inequality,
\[
\sup_{\theta\in C}
\mathbb{P}(|W_{N,\theta}-1|>\epsilon)
\le
\frac{1}{\epsilon}
\sup_{\theta\in C}
\E|W_{N,\theta}-1|
\xrightarrow[N\to\infty]{}0.
\]

Let $\delta \in (0,1)$. We aim to show that there exists $N_1=N_1(\delta)$ such that, for all $N \geq N_1$,
\[
\inf_{z \in [1-\delta/2,\,1+\delta/2]}
\inf_{\theta \in C}
\widetilde{\mathcal{Q}}_{N,\theta}(z)
>
\frac{1-\delta}{2\delta}>0.
\]
To establish this, it suffices to prove that there exists $N_1=N_1(\delta)$ such that, for all $N\geq N_1$, all $z\in[1-\delta/2,1+\delta/2]$ and all $\theta\in C$,
\[
\left|
\E\left[
\frac{1}{W_{N,\theta}}
\ind_{z\in[W_{N,\theta}(1-\delta),\,W_{N,\theta}(1+\delta)]}
\right]
-1
\right|
\leq
\delta,
\]
since
\begin{align*}
\widetilde{\mathcal{Q}}_{N,\theta}(z)
&=
\int
\frac{1}{2\delta w}
\ind_{z\in[w(1-\delta),\,w(1+\delta)]}
\mathcal{Q}_{N,\theta}(\mathrm{d}w)\\
&=
\frac{1}{2\delta}
\E\left[
\frac{1}{W_{N,\theta}}
\ind_{z\in[W_{N,\theta}(1-\delta),\,W_{N,\theta}(1+\delta)]}
\right].
\end{align*}
Let $\epsilon'>0$, let $z\in[1-\delta/2,1+\delta/2]$, and define
\[
f_z(x)
=
\frac{1}{x}
\ind_{x\in\left[\frac{z}{1+\delta},\,\frac{z}{1-\delta}\right]}.
\]
Then, for any $\theta\in C$,
\begin{align*}
\left|
\E\left[f_z(W_{N,\theta})-1\right]
\right|
&=
\left|
\E\left[
(f_z(W_{N,\theta})-1)\ind_{|W_{N,\theta}-1|\leq\epsilon'}
+
(f_z(W_{N,\theta})-1)\ind_{|W_{N,\theta}-1|>\epsilon'}
\right]
\right|\\
&\leq
\E\left[
|f_z(W_{N,\theta})-1|
\ind_{|W_{N,\theta}-1|\leq\epsilon'}
\right]
+
\E\left[
|f_z(W_{N,\theta})-1|
\ind_{|W_{N,\theta}-1|>\epsilon'}
\right].
\end{align*}
We can choose $\epsilon'$ such that
\[
[1-\epsilon',\,1+\epsilon']
\subset
\left[
\frac{z}{1+\delta},
\frac{z}{1-\delta}
\right]
\]
for every $z\in[1-\delta/2,1+\delta/2]$. For instance, this holds for any
\[
\epsilon'
\in
\left(0,\frac{\delta/2}{1+\delta}\right).
\]
In the remainder of the proof, we fix such an $\epsilon'$ depending on $\delta$. Then,
\begin{align*}
\E\left[
|f_z(W_{N,\theta})-1|
\ind_{|W_{N,\theta}-1|\leq\epsilon'}
\right]
&=
\E\left[
\left|
\frac{1}{W_{N,\theta}}-1
\right|
\ind_{|W_{N,\theta}-1|\leq\epsilon'}
\right]\\
&=
\E\left[
\frac{|W_{N,\theta}-1|}{W_{N,\theta}}
\ind_{|W_{N,\theta}-1|\leq\epsilon'}
\right]\\
&\leq
\frac{\epsilon'}{1-\epsilon'}.
\end{align*}
Since
\[
\frac{\epsilon'}{1-\epsilon'}
<
\frac{\delta}{2+\delta}
<
\frac{\delta}{2},
\]
it follows that
\[
\E\left[
|f_z(W_{N,\theta})-1|
\ind_{|W_{N,\theta}-1|\leq\epsilon'}
\right]
<
\frac{\delta}{2}.
\]
For the second expectation,
\begin{align*}
\E\left[
|f_z(W_{N,\theta})-1|
\ind_{|W_{N,\theta}-1|>\epsilon'}
\right]
&\leq
\E\left[
\left(
\frac{1}{W_{N,\theta}}
\ind_{W_{N,\theta}\in
[\frac{z}{1+\delta},\,\frac{z}{1-\delta}]}
+1
\right)
\ind_{|W_{N,\theta}-1|>\epsilon'}
\right]\\
&\leq
\left(
1+\frac{1+\delta}{z}
\right)
\mathbb{P}\!\left(
|W_{N,\theta}-1|>\epsilon'
\right)\\
&\leq
\left(
1+\frac{1+\delta}{1-\delta/2}
\right)
\sup_{\theta\in C}
\mathbb{P}\!\left(
|W_{N,\theta}-1|>\epsilon'
\right).
\end{align*}
Since
\[
\sup_{\theta\in C}
\mathbb{P}\!\left(
|W_{N,\theta}-1|>\epsilon'
\right)
\xrightarrow[N\to\infty]{}
0,
\]
there exists $N_1=N_1(\delta)$ such that, for all $N\geq N_1$,
\[
\sup_{\theta\in C}
\mathbb{P}\!\left(
|W_{N,\theta}-1|>\epsilon'
\right)
\leq
\left(
1+\frac{1+\delta}{1-\delta/2}
\right)^{-1}
\frac{\delta}{2}.
\]
Indeed, $\epsilon'$ has been fixed as a function of $\delta$, so the resulting threshold satisfies
\[
N_1 := N_1(\delta, \epsilon'(\delta))=N_1(\delta).
\]
Consequently, for all $N\geq N_1$,
\begin{align*}
\left|
\E\left[
f_z(W_{N,\theta})-1
\right]
\right|
&<
\frac{\delta}{2}+\left(1+\frac{1+\delta}{1-\delta/2}\right)\left(1+\frac{1+\delta}{1-\delta/2}\right)^{-1}\frac{\delta}{2}\\
&=
\delta.
\end{align*}
Therefore, for all $N\geq N_1$, all $z\in[1-\delta/2,1+\delta/2]$, and all $\theta\in C$,
\[
\left|
2\delta\widetilde{\mathcal{Q}}_{N,\theta}(z)-1
\right|
<
\delta,
\]
which implies
\[
2\delta\widetilde{\mathcal{Q}}_{N,\theta}(z)
>
1-\delta.
\]
Hence,
\[
\widetilde{\mathcal{Q}}_{N,\theta}(z)
>
\frac{1-\delta}{2\delta}
>0.
\]

Next, we now show that any compact set $K \subset \mathbb{R}^d \times \mathbb{R}_+^*$ is 1-small. We start by showing the property for $N > N_1$ and then for  $N \leq N_1$.

Let $(\theta,w) \in K$, $A \times B \in \mathcal{B}(\mathbb{R}^d)\times \mathcal{B}(\mathbb{R}_+^*)$, and $N > N_1$. Define $C := \{\theta \in \mathbb{R}^d : (\theta,w) \in K\}$ ($C$ is compact). Then
\begin{align*}
    \tilde{P}_N(\theta,w; A\times B) 
    &\geq \int_{A \times B} q(\vartheta|\theta) \left\{1 \wedge \frac{\pi(\vartheta)}{\pi(\theta)} \frac{z}{w}\right\} 
    \widetilde{\mathcal{Q}}_{N,\vartheta}(z) \,\mathrm{d}z\,\mathrm{d}\vartheta \\
     &\geq \int_{A} q(\vartheta|\theta) \left\{1 \wedge\frac{\pi(\vartheta)}{\pi(\theta)}\right\}  
    \int_{B} \left\{1 \wedge \frac{z}{w}\right\}  
    \widetilde{\mathcal{Q}}_{N,\vartheta}(z)\,\mathrm{d}z\,\mathrm{d}\vartheta
\end{align*}
Since $C$ is a compact of $\reals^d$, Assumptions~\ref{hyp1} and \ref{hyp5} imply that there exists $\epsilon_C>0$ and a probability measure $\nu_C$ such that for all $\theta \in C$ and $A \in \mathcal{B}(\reals^d)$, (see for example the proof of Theorem 2.2 of \cite{roberts_geometric_1996}),
\[
\int_{A} q(\vartheta|\theta) \left\{1 \wedge\frac{\pi(\vartheta)}{\pi(\theta)}\right\} \mathrm{d}\vartheta \geq \epsilon_C \, \nu_C(A).
\]
Furthermore, the function $\vartheta \mapsto \int_{B} \left\{1\wedge \frac{z}{w}\right\}  \widetilde{\mathcal{Q}}_{N,\vartheta}(z)\,\mathrm{d}z $ is positive and measurable. Hence, 
\begin{align*}
    \tilde{P}_N(\theta,w; A\times B)
    &\geq \int_{A}   
    \left(\int_{B} \left\{1\wedge  \frac{z}{w}\right\}  \widetilde{\mathcal{Q}}_{N,\vartheta}(z)\,\mathrm{d}z \right) \left\{1\wedge  \frac{\pi(\vartheta)}{\pi(\theta)}\right\} q(\vartheta|\theta) \mathrm{d}\vartheta\\
    &\geq \epsilon_{C}\int_{A} \left(\int_{B} \left\{1 \wedge \frac{z}{w}\right\}  
    \widetilde{\mathcal{Q}}_{N,\vartheta}(z)\,\mathrm{d}z \right) \nu_{C}(\mathrm{d}\vartheta) \\
    &\geq \epsilon_{C}\int_{A} \int_{B} \left\{1 \wedge \frac{z}{\bar{w}_K}\right\}  
    \widetilde{\mathcal{Q}}_{N,\vartheta}(z)\,\mathrm{d}z\,\nu_{C}(\mathrm{d}\vartheta), \quad \text{ since } w \in K \text{ and $K$ is compact}  \\
    &\geq \epsilon_{C}\int_{A} \int_{B \cap [1-\delta/2,1+\delta/2]} 
    \left\{1 \wedge \frac{z}{\bar{w}_K}\right\}  
    \widetilde{\mathcal{Q}}_{N,\vartheta}(z)\,\mathrm{d}z\,\nu_{C}(\mathrm{d}\vartheta) \\
    &\geq \epsilon_{C} \left\{1 \wedge \frac{1-\delta/2}{\bar{w}_K}\right\} 
    \int_{A} \int_{B \cap [1-\delta/2,1+\delta/2]}  
    \widetilde{\mathcal{Q}}_{N,\vartheta}(z)\,\mathrm{d}z\, \nu_{C}(\mathrm{d}\vartheta) \\
    &\geq \frac{1-\delta}{2\delta}\, \epsilon_{C} \left\{1 \wedge \frac{1-\delta/2}{\bar{w}_K}\right\} 
    \nu_C(A)\, \mu_{\mathrm{Leb}}(B \cap [1-\delta/2,1+\delta/2]).
\end{align*}
We now show that $K$ is uniformly 1-small for $\tilde{P}_N$ for all $N_0 \leq N \leq N_1$. In the proof of Proposition~\ref{proposition1}, we showed that for all $N\geq N_0$, $M_{\widetilde{W}_{N,\theta}} \leq M_{\widetilde{W}_{N_0,\theta}} < \infty$. This allows us to apply Lemma 29 of Andrieu to obtain $\tilde{w}_0 > 1$ only depending on $M_{\widetilde{W}_{N_0}}$ such that for all $\theta \in \mathbb{R}^d$,
\[
\int \{\tilde{w}_0 \wedge w\} \widetilde{\mathcal{Q}}_{N,\theta}(w)\,\mathrm{d}w \geq \int \{\tilde{w}_0 \wedge w\} \widetilde{\mathcal{Q}}_{N_0,\theta}(w)\,\mathrm{d}w > 0.
\]
Since, under Assumption~\ref{hyp2} and Lemma~5 in \cite{andrieu_establishing_2016}, we have for all $t \in \reals_+^*$, 
$$\E[t\wedge \widetilde{W}_{N,\theta}] \geq \E[t\wedge \widetilde{W}_{N_0,\theta}].$$ 

\noindent
Thus, for $N_0 \leq N \leq N_1$, $(\theta,w) \in K$ and $A\times B \in \mathcal{B}(\reals^d) \times \mathcal{B}(\reals_+^*)$, we have
\begin{align*}
        \tilde{P}_N(\theta,w; A\times B) 
        &\geq \epsilon_{C}\int_{A} \int_{B} \left\{1 \wedge \frac{z}{\bar{w}_K}\right\}  
        \widetilde{\mathcal{Q}}_{N,\vartheta}(z)\,\mathrm{d}z \, \nu_{C}(\mathrm{d}\vartheta) \\
        &\geq \epsilon_{C}\int_{A} \int_{B} \left\{1 \wedge \frac{z}{\bar{w}_K}\right\}  
        \left\{\min_{N \in \{N_0, \ldots, N_1\}} \widetilde{\mathcal{Q}}_{N,\vartheta}(z)\right\}\mathrm{d}z \, \nu_{C}(\mathrm{d}\vartheta) \\
        &\geq \epsilon_{C}\int_{A} \int_{B} \left\{1 \wedge \frac{z}{\bar{w}_K \vee \tilde{w}_0}\right\}  
        \left\{\min_{N \in \{N_0, \ldots, N_1\}} \widetilde{\mathcal{Q}}_{N,\vartheta}(z)\right\}\mathrm{d}z \, \nu_{C}(\mathrm{d}\vartheta) \\
         &\geq \epsilon_{C} \frac{1}{\bar{w}_K \vee \tilde{w}_0}
        \int_{A} \int_{B} \left\{(\bar{w}_K \vee \tilde{w}_0) \wedge z\right\}  
        \left\{\min_{N \in \{N_0, \ldots, N_1\}} \widetilde{\mathcal{Q}}_{N,\vartheta}(z)\right\}\mathrm{d}z \, \nu_{C}(\mathrm{d}\vartheta)\\
        &\geq \epsilon_{C} \frac{1}{\bar{w}_K \vee \tilde{w}_0}
        \int_{A} \int_{B} \left\{\tilde{w}_0 \wedge z\right\}  
        \left\{\min_{N \in \{N_0, \ldots, N_1\}} \widetilde{\mathcal{Q}}_{N,\vartheta}(z)\right\}\mathrm{d}z \, \nu_{C}(\mathrm{d}\vartheta).
\end{align*}
Define
\[
\nu (A \times B) = \int_{A} \tilde{K}'(\vartheta, B)\, \nu_{C}(\mathrm{d}\vartheta),
\]
where
\[
\tilde{K}' (\vartheta, B)
= \int_{B} \left\{\tilde{w}_0 \wedge z\right\}  
\left\{\min_{N \in \{N_0, \ldots, N_1\}} \widetilde{\mathcal{Q}}_{N,\vartheta}(z)\right\}\mathrm{d}z.
\]
Then $\nu$ is non-trivial, since for all $\theta \in \reals^d$
 \begin{align*}
        \tilde{K}'(\theta, \reals_+^*) =  \int_{\reals_+^*} \{\tilde{w}_0 \wedge z\} \left\{\min_{N \in \{N_0, \ldots, N_1\}} \widetilde{\mathcal{Q}}_{N,\vartheta}(z)\right\}\mathrm{d}z >0.
    \end{align*}
\end{proof}

\subsection{Proof of \texorpdfstring{$\phi$}{}-irreducibiliby and aperiodicity of \texorpdfstring{$\tilde{P}_N$}{ }}
\label{sub5appendixB}
\begin{proof}[Proof of Proposition~\ref{prop4}]
 For every $N \geq N_0$, Assumptions~\ref{hyp1} and~\ref{hyp5} imply, by Theorem~2.1 of \cite{jarner_geometric_2000}, that the marginal random-walk Metropolis algorithm associated with the PM algorithm is $\mu_{\mathrm{Leb}}$-irreducible and aperiodic.

It therefore remains to show that Assumption~\ref{hyp4} implies that the rejection probability of the augmented PM kernel $\tilde{P}_N$ is strictly positive. The conclusion then follows from Theorem~1 of \cite{andrieu_pseudo-marginal_2009}, which states that if the marginal algorithm is $\mu_{\mathrm{Leb}}$-irreducible and aperiodic and the rejection probability is strictly positive, then the corresponding PM algorithm is also $\mu_{\mathrm{Leb}}$-irreducible and aperiodic.

Let $\xi \sim \Unif(1-\delta,1+\delta)$. Since $1 \preceq_{cx} \xi$, it follows that, for every $N \geq N_0$ and every $\theta \in \mathbb{R}^d$,
\[
W_{N,\theta} \preceq_{cx} W_{N,\theta}\xi
    = \widetilde{W}_{N,\theta}.
\]
Moreover, the function
\(
z \longmapsto 1 \wedge \frac{\pi(\vartheta)}{\pi(\theta)}\frac{z}{w}
\)
is concave. Therefore, by the definition of the convex order,
\[
\E\left[1 \wedge\frac{\pi(\vartheta)}{\pi(\theta)}\frac{W_{N,\vartheta}}{w}\,\middle|\,\vartheta\right]
\geq
\E\left[1 \wedge\frac{\pi(\vartheta)}{\pi(\theta)}\frac{\widetilde{W}_{N,\vartheta}}{w}\,\middle|\,\vartheta\right].
\]
Taking expectations with respect to $\vartheta$ yields
\[\E\left[\E\left[1 \wedge\frac{\pi(\vartheta)}{\pi(\theta)}\frac{W_{N,\vartheta}}{w}\,\middle|\,\vartheta\right]\right]
\geq
\E\left[\E\left[1 \wedge\frac{\pi(\vartheta)}{\pi(\theta)}\frac{\widetilde{W}_{N,\vartheta}}{w}\,\middle|\,\vartheta\right]\right].
\]
Finally, Assumption~\ref{hyp4} implies that
\[
\E\left[\E\left[1 \wedge\frac{\pi(\vartheta)}{\pi(\theta)}\frac{W_{N,\vartheta}}{w}\,\middle|\,\vartheta\right]\right]
<1,\]
and therefore
\[
\E\left[\E\left[1 \wedge\frac{\pi(\vartheta)}{\pi(\theta)}\frac{\widetilde{W}_{N,\vartheta}}{w}\,\middle|\,\vartheta\right]\right]
<1.
\]
Hence, the rejection probability of the augmented PM kernel $\tilde{P}_N$ is strictly positive, completing the proof.
\end{proof}

\subsection{Proof of the ergodicity of the APM}
\label{sub6appendixB}

\begin{proof}[Proof of Theorem~\ref{thm1}]
We prove the convergence of the marginal distribution of the parameter
component of the APM chain in total variation. By the bounded-function
representation of the total variation norm,
\[
\left\|\mathcal{L}(\theta_\ell)-\pi\right\|_{\mathrm{TV}}
=
\sup_{\substack{f:\Theta\to\mathbb{R}\\\|f\|_\infty\leq 1}}\left|\E[f(\theta_\ell)]-\pi(f)\right|.
\]
Thus, it suffices to show that
\[
\sup_{\substack{f:\Theta\to\mathbb{R}\\\|f\|_\infty\leq 1}}\left|\E[f(\theta_\ell)]-\pi(f)\right|\underset{\ell\to\infty}{\longrightarrow}0.
\]

The proof follows the argument of \cite[Section~4.2.3]{atchade_limit_2010}, with a modification to account for the dependence of the invariant distribution $\tilde\pi_N$ on the adaptive parameter $N$. We first show that, when $f$ depends only on the parameter component, the corresponding centered function is independent of $N$. Since the intermediate bounds involving return times and the coupling construction remain unchanged, we recall the corresponding results from \cite[Section~4.2.3]{atchade_limit_2010} and briefly explain how they are obtained, using the notation specific to the APM algorithm. We then provide the details of the inequalities that must be modified to handle the dependence of $\tilde\pi_N$ on $N$.

Under Assumptions~\ref{hyp1}, \ref{hyp5}, \ref{hyp3},
\ref{hyp2}, \ref{hyp7}, \ref{hyp8}, and~\ref{hyp4},
Propositions~\ref{prop2}, \ref{prop3}, and~\ref{prop4}
establish the conditions of Lemma~\ref{lemma1}
(see Subsection~\ref{prooflemma1}). Hence, the conclusions
of Lemma~\ref{lemma1} hold over the family of PM
kernels $\{\tilde P_N\}_{N\geq N_0}$. Thus, by letting $\varepsilon>0$, we have by point (1) of Lemma~\ref{lemma1} that there exists an integer $Q$ such that
\begin{equation}
\sup_{n\in\mathbb N}
\sup_{x\in K}
\sup_{N\geq N_0}
\E^{(n)}_{x,N}[r(\tau_K)]
\sum_{j\geq Q}\frac{1}{r(j)}
\leq\varepsilon.
\label{eq:apm-return-bound}
\end{equation}
By point (2) of Lemma~\ref{lemma1}, we have that there exists an integer $m$ such that
\begin{equation}
\sup_{N\geq N_0}
\sup_{x\in\Theta\times\mathcal W}
V^{-1}(x)
\left\|
\tilde P_N^m(x,\cdot)-\tilde\pi_N
\right\|_{\mathrm{TV}}
\leq\frac{\varepsilon}{Q},
\label{eq:apm-uniform-bound}
\end{equation}

and, by Proposition~\ref{proposition1} that there exists an integer $\ell^\star$ such that, for any $j\geq\ell^\star$,
\begin{equation}
\Prob\left(
D_j\geq
\frac{\varepsilon}{2(m+Q-1)^2Q}
\right)
\leq
\frac{\varepsilon}{4(m+Q-1)^2Q}.
\label{DAresult}
\end{equation}

Let $f:\Theta\to\mathbb R$ be a bounded measurable function such that $\|f\|_\infty\leq1$. Define $g:\Theta\times\mathcal W \to\reals$ by
\[
g(x)=f(\theta),
\qquad \text{ where } x=(\theta,w).
\]

For every $N\geq N_0$, the invariant distribution
$\tilde\pi_N$ of $\tilde P_N$ has the $\pi$ as the marginal in $\theta$. That is,
\[
\int_{\mathcal W}
\tilde\pi_N(\mathrm d\theta,\mathrm dw)
=
\pi(\mathrm d\theta).
\]
Consequently,
\begin{align*}
\tilde\pi_N(g)
&=
\int_{\Theta\times\mathcal W}
f(\theta)\,
\tilde\pi_N(\mathrm d\theta,\mathrm dw)\\
&=
\int_\Theta f(\theta)\,\pi(\mathrm d\theta)
=\pi(f).
\end{align*}

Writing $\bar f=f-\pi(f)$, the centered function is therefore
\[
\bar g(x,N)
=g(x)-\tilde\pi_N(g)
=f(\theta)-\pi(f)
=\bar f(\theta).
\]

In particular, $\bar g$ and its supremum norm are independent of $N$. We identify $\bar f$ with its extension to $\Theta\times\mathcal W$, defined by $\bar f(\theta,w)=\bar f(\theta)$.

It follows that
\[
\E[\bar g(X_\ell,N_\ell)]
=
\E[\bar f(\theta_\ell)].
\]
Hence, the desired convergence is equivalent to
\begin{equation*}
\sup_{\substack{f:\Theta\to\mathbb R\\
\|f\|_\infty\leq1}}\left|\E[\bar g(X_\ell,N_\ell)]\right| 
\underset{\ell\to\infty}{\longrightarrow}0.
\end{equation*}

\medskip
\noindent\textbf{Last-visit decomposition}

Let $\ell\geq m+Q+\ell^\star$. We use the notation
$X_\ell=(\theta_\ell,\widetilde W_\ell)$ and
$Z_\ell=(X_\ell,N_\ell)$.

For integers $a\leq b$, let $\{X_{a:b}\notin K\}=\bigcap_{j=a}^{b}\{X_j\notin K\}$, with the convention that
$\{X_{b+1:b}\notin K\}=\Omega$.

Following \cite{atchade_limit_2010}, we partition the trajectories according to the last visit to $K$ before time $\ell-m$. We obtain
\begin{align}
\E[\bar g(X_\ell,N_\ell)]
&=
\E\left[\bar f(\theta_\ell)\ind_{\{X_{0:\ell-m}\notin K\}}\right]+\sum_{k=0}^{\ell-m}\E\left[\bar f(\theta_\ell)\ind_{\{X_k\in K\}}\ind_{\{X_{k+1:\ell-m}\notin K\}}\right].
\label{eq:apm-last-visit}
\end{align}

We bound the terms on the right-hand side separately.

First, since $\bar f$ is bounded,
\[
\left|
\E\left[
\bar f(\theta_\ell)
\ind_{\{X_{0:\ell-m}\notin K\}}
\right]
\right|
\leq
\|\bar f\|_\infty
\Prob(X_{0:\ell-m}\notin K)
\leq
\|\bar f\|_\infty
\Prob(\tau_K>\ell-m).
\]

As in \cite[Section~4.2.3]{atchade_limit_2010}, the definition of the hitting time $\tau_K=\inf\{j\geq0:X_j\in K\}$ and since $\ind_{\tau_K> \ell-m} \leq 1\wedge (\tau_k/\ell-m+1$),
\[
\Prob(\tau_K>\ell-m)
\leq
\E\left[\frac{\tau_K}{\ell-m+1}\wedge1\right].
\]
Since the integrand is bounded by $1$, the dominated convergence theorem implies that the right-hand side tends to zero as $\ell\to\infty$. Thus, for $\ell$ sufficiently large,
\begin{equation}
\left|
\E\left[
\bar f(\theta_\ell)
\ind_{\{X_{0:\ell-m}\notin K\}}
\right]
\right|
\leq
\|\bar f\|_\infty\varepsilon.
\label{eq:apm-initial-bound}
\end{equation}

Next, we follow the authors in \cite{atchade_limit_2010} by splitting the sum in~\eqref{eq:apm-last-visit} into two parts, corresponding to early and recent visits to $K$.

For the first part, the Markov inequality and Equation \eqref{eq:apm-return-bound} give, as in \cite[Equation~(17)]{atchade_limit_2010},
\begin{align}
\left|\sum_{k=0}^{\ell-m-Q}\E\left[\bar f(\theta_\ell)\ind_{\{X_k\in K\}}\ind_{\{X_{k+1:\ell-m}\notin K\}}\right]\right|
\leq
\|\bar f\|_\infty\sup_{n\in\mathbb N}\sup_{x\in K}\sup_{N\geq N_0}\E^{(n)}_{x,N}[r(\tau_K)]\sum_{j\geq Q}\frac{1}{r(j)}
\leq
\|\bar f\|_\infty\varepsilon.
\label{eq:apm-early-bound}
\end{align}

It remains to control the last $Q$ terms. We apply the optimal coupling of \cite[Proposition~4.3]{atchade_limit_2010} over $m+Q$ iterations. For $n\geq0$ and $z=(x,N)$, let $\bar{\Prob}^{(n)}_{z,z}$ denote the coupling probability between the non-homogeneous adaptive process and the homogeneous Markov chain with transition kernel $\tilde P_N$ and initial distribution $\delta_x$. We denote the corresponding expectation by $\bar{\E}^{(n)}_{z,z}$.

Let $(X_j,\tilde X_j)$ denote the coupled processes, where $\tilde X_j=(\tilde\theta_j,\tilde{\widetilde{W}}_j)$. Define the coupling event $\mathcal E_{m+Q}=\{X_j=\tilde X_j,\ 0\leq j\leq m+Q\}.$

Using Proposition~4.3 of \cite{atchade_limit_2010}, the diminishing adaptation bound~\eqref{DAresult}, and the fact that $D_j\leq2$, we obtain, for every $r\geq\ell^\star$,
\begin{equation}
\E\left[\ind_{\{X_r\in K\}}\bar{\Prob}^{(r)}_{Z_r,Z_r}(\mathcal E_{m+Q}^{\complement})\right]\leq\frac{\varepsilon}{Q}.
\label{CouplingIneq}
\end{equation}

Let $k\in\{\ell-(m+Q)+1,\ldots,\ell-m\}$ and set $t=\ell-m-k$. Then $0\leq t\leq Q-1$.

On the coupling event, the adaptive and frozen
trajectories coincide up to time $m+Q$. Therefore,
\begin{align*}
\left|\E\left[\bar f(\theta_\ell)\ind_{\{X_k\in K\}}\ind_{\{X_{k+1:\ell-m}\notin K\}}\right]\right|
&\leq\E\left[\ind_{\{X_k\in K\}}\left|\bar{\E}^{(k)}_{Z_k,Z_k}\left[\bar f(\tilde\theta_{m+t})\ind_{\{\tilde X_{1:t}\notin K\}}\right]\right|\right]\\
&\quad+\|\bar f\|_\infty\E\left[\ind_{\{X_k\in K\}}\bar{\Prob}^{(k)}_{Z_k,Z_k}(\mathcal E_{m+Q}^{\complement})\right]
\\
&\leq\E\left[\ind_{\{X_k\in K\}}\left|\bar{\E}^{(k)}_{Z_k,Z_k}\left[\bar f(\tilde\theta_{m+t})\ind_{\{\tilde X_{1:t}\notin K\}}\right]\right|\right]+\|\bar f\|_\infty\frac{\varepsilon}{Q},
\end{align*}
where the first inequality follows from the proof in \cite[Subsection 4.2.3]{atchade_limit_2010} and the last inequality follows from~\eqref{CouplingIneq}.

Since the frozen process is a homogeneous Markov chain with transition kernel $\tilde P_{N_k}$, the Markov property gives
\begin{align*}
\left|\bar{\E}^{(k)}_{Z_k,Z_k}\left[\bar f(\tilde\theta_{m+t})\ind_{\{\tilde X_{1:t}\notin K\}}\right]\right|
&=
\left|\bar{\E}^{(k)}_{Z_k,Z_k}\left[\ind_{\{\tilde X_{1:t}\notin K\}}\tilde P_{N_k}^{m}\bar f(\tilde X_t)\right]\right|\\
&\leq 
\bar{\E}^{(k)}_{Z_k,Z_k}\left[\ind_{\{\tilde X_{1:t}\notin K\}}\left|\tilde P_{N_k}^{m}\bar f(\tilde X_t)\right|\right].
\end{align*}

Consequently,
\begin{align}
\left|\E\left[\bar f(\theta_\ell)\ind_{\{X_k\in K\}}\ind_{\{X_{k+1:\ell-m}\notin K\}}\right]\right|
\leq
\E\left[\ind_{\{X_k\in K\}}\bar{\E}^{(k)}_{Z_k,Z_k}\left[\ind_{\{\tilde X_{1:t}\notin K\}}\left|\tilde P_{N_k}^{m}\bar f(\tilde X_t)\right|\right]\right]+\|\bar f\|_\infty\frac{\varepsilon}{Q}.
\label{eq:apm-coupled-bound}
\end{align}

We now use the fact that $\bar f$ is independent of $N$. For every $N\geq N_0$, the invariant distribution $\tilde\pi_N$ satisfies \[\tilde\pi_N(\bar f)=\pi(f)-\pi(f)=0.\]

Hence, for any $x\in\Theta\times\mathcal W$,
\begin{align*}
\tilde P_N^m\bar f(x)
=
\tilde P_N^m\bar f(x)-\tilde\pi_N(\bar f)
=
\int_{\Theta\times\mathcal W}\bar f(\vartheta)\left[\tilde P_N^m(x;\mathrm d\vartheta,\mathrm dz)-\tilde\pi_N(\mathrm d\vartheta,\mathrm dz)\right].
\end{align*}

By the definition of the total variation norm,
\[
\left|\tilde P_N^m\bar f(x)\right|
\leq
\|\bar f\|_\infty\left\|\tilde P_N^m(x,\cdot)-\tilde\pi_N\right\|_{\mathrm{TV}}.
\]

Applying the uniform convergence
bound~\eqref{eq:apm-uniform-bound} yields
\begin{equation}
\left|\tilde P_N^m\bar f(x)\right|
\leq
\|\bar f\|_\infty
\frac{\varepsilon}{Q}V(x),
\qquad N\geq N_0.
\label{eq:apm-centered-bound}
\end{equation}

In particular, this bound applies with $N=N_k$
and $x=\tilde X_t$.

Substituting~\eqref{eq:apm-centered-bound}
into~\eqref{eq:apm-coupled-bound}, we obtain
\begin{align}
\left|\E\left[\bar f(\theta_\ell)\ind_{\{X_k\in K\}}\ind_{\{X_{k+1:\ell-m}\notin K\}}\right] \right|
\leq
\|\bar f\|_\infty\frac{\varepsilon}{Q}\E\left[\ind_{\{X_k\in K\}}\bar{\E}^{(k)}_{Z_k,Z_k}\left[\ind_{\{\tilde X_{1:t}\notin K\}}V(\tilde X_t)\right]\right]
+\|\bar f\|_\infty\frac{\varepsilon}{Q}.
\label{eq:apm-before-drift}
\end{align}

We now bound the remaining expectation using point (3) of Lemma~\ref{lemma1}. Recall that the frozen process starts at $\tilde X_0=X_k$ and evolves according to the fixed kernel $\tilde P_{N_k}$. Suppose first that $t\geq1$. For any $j\geq1$, the Markov property gives
\begin{align*}
\bar{\E}^{(k)}_{Z_k,Z_k}\left[V(\tilde X_{j+1})\ind_{\{\tilde X_{1:j+1}\notin K\}}\right]
&=
\bar{\E}^{(k)}_{Z_k,Z_k}\left[\ind_{\{\tilde X_{1:j}\notin K\}}\int_{K^c}V(y)\tilde P_{N_k}(\tilde X_j,\mathrm dy)\right]\\
&\leq
\bar{\E}^{(k)}_{Z_k,Z_k}\left[\ind_{\{\tilde X_{1:j}\notin K\}}\tilde P_{N_k}V(\tilde X_j\right]\\
&\leq
\bar{\E}^{(k)}_{Z_k,Z_k}\left[\ind_{\{\tilde X_{1:j}\notin K\}}V(\tilde X_j)\right].
\end{align*}
The last inequality follows from point (3) of Lemma~\ref{lemma1}, since $\tilde X_j\notin K$.

Applying the previous property successively yields
\begin{align*}
\bar{\E}^{(k)}_{Z_k,Z_k}\left[V(\tilde X_t)\ind_{\{\tilde X_{1:t}\notin K\}}\right] 
\leq \bar{\E}^{(k)}_{Z_k,Z_k}\left[V(\tilde X_1)\ind_{\{\tilde X_1\notin K\}}\right]
\leq \tilde P_{N_k}V(X_k).
\end{align*}
When $t=0$, using the convention that $\{X_{1:0}\notin K\}=\Omega$,
\[
\bar{\E}^{(k)}_{Z_k,Z_k}\left[V(\tilde X_0)\ind_{\{\tilde X_{1:0}\notin K\}}\right]
=
\bar{\E}^{(k)}_{Z_k,Z_k}\left[V(\tilde X_0)]\right]
=
V(X_k).
\]

Define the finite constants
\[
A_K
:=
\sup_{N\geq N_0}\sup_{x\in K}\tilde P_N V(x),
\qquad
B_K
:=
\sup_{x\in K}V(x).
\]

The two cases above imply, using $\ind_{\{X_k\in K\}}\leq1$, 
\begin{equation}
\E\left[\ind_{\{X_k\in K\}}\bar{\E}^{(k)}_{Z_k,Z_k}\left[\ind_{\{\tilde X_{1:t}\notin K\}}V(\tilde X_t)\right]\right]
\leq A_K+B_K.
\label{eq:apm-drift-bound}
\end{equation}

Combining~\eqref{eq:apm-before-drift} and~\eqref{eq:apm-drift-bound} gives
\begin{align*}
\left|\E\left[\bar f(\theta_\ell)\ind_{\{X_k\in K\}}\ind_{\{X_{k+1:\ell-m}\notin K\}}\right]\right|
&\leq
\|\bar f\|_\infty\frac{\varepsilon}{Q}\left(1+A_K+B_K\right).
\end{align*}

Consequently,
\begin{align}
\left|\sum_{k=\ell-(m+Q)+1}^{\ell-m}\E\left[\bar f(\theta_\ell)\ind_{\{X_k\in K\}}\ind_{\{X_{k+1:\ell-m}\notin K\}}\right]\right|
&\leq
\sum_{k=\ell-(m+Q)+1}^{\ell-m}\left|\E\left[\bar f(\theta_\ell)\ind_{\{X_k\in K\}}\ind_{\{X_{k+1:\ell-m}\notin K\}}\right]\right| \nonumber\\
&\leq
\|\bar f\|_\infty\varepsilon\left(1+A_K+B_K\right).
\label{eq:apm-recent-bound}
\end{align}

\textbf{Conclusion}

Combining~\eqref{eq:apm-last-visit}, \eqref{eq:apm-initial-bound}, \eqref{eq:apm-early-bound}, and \eqref{eq:apm-recent-bound}, we obtain, for $\ell$ sufficiently large,
\[
\left|\E[\bar f(\theta_\ell)]\right|
\leq
\|\bar f\|_\infty\varepsilon \left(3+A_K+B_K\right).
\]

Since $\|f\|_\infty\leq1$, we have
$\|\bar f\|_\infty\leq2$. Therefore,
\[
\sup_{\substack{f:\Theta\to\mathbb R\\ \|f\|_\infty\leq1}}
\left|
\E[f(\theta_\ell)]-\pi(f)\right|
\leq
2\varepsilon\left(3+A_K+B_K\right).
\]
The constants $A_K$ and $B_K$ are finite and
independent of $\ell$ and $f$. Since $\varepsilon>0$ is arbitrary, this completes the proof.
\end{proof}

\section{Synthetic Data Example: Theoretical Guarantees for Ergodicity}
\label{appendixD0}
\begin{proof}[Proof of Proposition~\ref{se:verification}]
We start by verifying Assumption~\ref{hyp1}. Under the instrumental proposal distribution specified in~\eqref{intru}, the MH algorithm corresponding to the PM algorithm is a RWM algorithm targeting the posterior distribution \( \pi \), whose density is given in~\eqref{SEpi}, up to a normalizing constant. This density is continuously differentiable and supported on \( \mathbb{R} \). 

\noindent
To analyze the behavior of the log-density, we observe that
\begin{align*}
\nabla \log &\pi(\theta)
= \theta \Bigg\{ \frac{T}{\theta^2 + 1} - \frac{T}{\theta^2 + 2} 
 - \frac{1}{(\theta^2 + 2)^2} \sum_{t=1}^T (\theta - y_t)^2 - \frac{\theta^2 + 1}{\theta^2 + 2} \left( T - \frac{1}{\theta} \sum_{t=1}^T y_t \right) 
- \frac{1}{\sigma_0^2} \Bigg\},
\end{align*}
from which it follows that
\[
\frac{\theta}{\|\theta\|} \nabla \log \pi(\theta) 
\underset{\|\theta\| \to \infty}{\sim} 
-\frac{\theta^2}{\|\theta\|} \left(T + \frac{1}{\sigma_0^2}\right) 
\underset{\|\theta\| \to \infty}{\longrightarrow}-\infty.
\]
\noindent
Furthermore, for some constant \( C > 0 \), the gradient of the  posterior density satisfies
{
\begin{align*}
\nabla \pi(\theta)
&= C \cdot \exp\left\{-\frac{1}{2}\left( 
\frac{\theta^2+1}{\theta^2+2} \sum_{t=1}^{T}(\theta - y_t)^2 
+ \frac{\theta^2}{\sigma_0^2} \right)\right\} \times \Bigg\{ \frac{T\theta}{(\theta^2 + 2)^2} 
\left(\frac{\theta^2 + 1}{\theta^2 + 2}\right)^{\frac{T}{2} - 1}  
- \left(\frac{\theta^2 + 1}{\theta^2 + 2}\right)^{\frac{T}{2}} \\
&\quad \times \left[ \frac{\theta}{(\theta^2 + 2)^2} \sum_{t=1}^{T}(\theta - y_t)^2 + \left( T\theta - \sum_{t=1}^T y_t \right) \frac{\theta^2 + 1}{\theta^2 + 2} 
+ \frac{\theta}{\sigma_0^2} \right] \Bigg\}.
\end{align*}
}Consequently,
\[\frac{\theta}{\|\theta\|} \frac{\nabla \pi(\theta)}{|\nabla \pi(\theta)|} 
\underset{\|\theta\| \to \infty}{\sim} 
\frac{-(T + 1/\sigma_0^2)\theta^2}{(T + 1/\sigma_0^2)\theta^2} = -1 < 0.\]

The proposal density \( q \) in~\eqref{intru} is Gaussian, and hence symmetric and bounded away from zero on compact sets. It follows that all conditions of Assumption~\ref{hyp2} are satisfied in the synthetic example setting.

Assumption~\ref{hyp3} is verified next with \( N_0 = 1 \). Define
\begin{align*}
W_1 := W_{1,\theta} = \frac{\widehat{p}_{T,1}(y|\theta, U)}{p_T(y|\theta)}= \frac{\prod_{t=1}^T  \varphi\left(y_t; U_{t,1}, 1 \right)}{\prod_{t=1}^T \varphi\left(y_t; \theta, \tfrac{\theta^2 + 2}{\theta^2 + 1} \right)},
\end{align*}
where \( U_{t,1} \sim \mathcal{N}(\theta, {1}/\{\theta^2 + 1\}) \). We aim to show that for some \( \alpha_1 > 0 \), \( \beta_1 > 1 \),
\[
\underset{\theta \in \Theta}{\mathrm{ess\,sup}} \, \E[W_1^{-\alpha_1} \vee W_1^{\beta_1}] < \infty.
\]
The evaluation of \( \E[W_1^{\beta_1}] \) can be carried out by completing the square in the exponent:
\begin{align*}
\E[W_1^{\beta_1}] 
= \prod_{t=1}^{T} \sqrt{ \frac{(\theta^2 + 2)^{\beta_1}}{(\theta^2 + 1)^{\beta_1 - 1} (\beta_1 + \theta^2 + 1)} }  \exp\Bigg\{ \frac{\beta_1}{2} (\theta^2 + 1)(\theta - y_t)^2 \left( \frac{1}{\theta^2 + 2} - \frac{1}{\beta_1 + \theta^2 + 1} \right) \Bigg\}.
\end{align*}
Choosing \( \beta_1 = 2 \), we obtain:
\begin{align*}
\E[W_1^2] &= \frac{(\theta^2 + 2)^T}{(\theta^2 + 1)^{T/2} (\theta^2 + 3)^{T/2}} 
\exp\left\{ \frac{(\theta^2 + 1)}{(\theta^2 + 2)(\theta^2 + 3)} \sum_{t=1}^T (\theta - y_t)^2 \right\},
\end{align*}
which is finite for all \( \theta \in \mathbb{R} \) and converges to \( \exp\{T\} \) as \( \|\theta\| \to \infty \). Therefore,
\[
\underset{\theta \in \Theta}{\mathrm{ess\,sup}} \, \E[W_1^2] < \infty.
\]
Similarly, taking \( \alpha_1 = 1/2 \) yields
\[
\underset{\theta \in \Theta}{\mathrm{ess\,sup}} \, \E[W_1^{-1/2}] < \infty.
\]
Thus, for \( N_0 = 1 \), \( \alpha_1 = 1/2 \), and \( \beta_1 = 2 \), Assumption~\ref{hyp3} holds via the inequality
\[
\E[W_1^{-\alpha_1} \vee W_1^{\beta_1}] 
\leq \E[W_1^{-\alpha_1}] + \E[W_1^{\beta_1}].
\]

The verification of Assumption~\ref{hyp2} is now provided for $N_0=1$. For any \( N \geq 1 \), define
\begin{align*}
    W_{N} := W_{N,\theta} &= \frac{\widehat p_{T,N}(y| \theta, U)}{p_{T}(y| \theta)}= \prod_{t=1}^{T} \frac{1}{N} \sum_{n=1}^{N} 
    \frac{\varphi(y_t ; U_{t,n}, 1)}{\varphi\left(y_t; \theta, \tfrac{\theta^2 + 2}{\theta^2 + 1} \right)}.
\end{align*}
The terms  {\({\varphi(y_t ; U_{t,n}, 1)}/\varphi(y_t; \theta, \{\theta^2 + 2\}/\{\theta^2 + 1\} )\)} are iid for \( n \in \{1, \dots, N\} \), and satisfy the conditional expectation identity:
\begin{align*}
    \E&\left[ \frac{\varphi(y_t ; U_{t,1}, 1)}{\varphi\left(y_t; \theta, \tfrac{\theta^2 + 2}{\theta^2 + 1} \right)} 
    \; \Bigg| \; \frac{1}{N} \sum_{n=1}^{N} \frac{\varphi(y_t ; U_{t,n}, 1)}{\varphi\left(y_t; \theta, \tfrac{\theta^2 + 2}{\theta^2 + 1} \right)} \right]
    = \frac{1}{N} \sum_{n=1}^{N} \frac{\varphi(y_t ; U_{t,n}, 1)}{\varphi\left(y_t; \theta, \tfrac{\theta^2 + 2}{\theta^2 + 1} \right)}.
\end{align*}
Therefore, for each \( t \in \{1, \ldots, T\} \),
\[
\frac{1}{N} \sum_{n=1}^{N} \frac{\varphi(y_t ; U_{t,n}, 1)}{\varphi\left(y_t; \theta, \tfrac{\theta^2 + 2}{\theta^2 + 1} \right)} 
\preceq_{cx} \frac{\varphi(y_t ; U_{t,1}, 1)}{\varphi\left(y_t; \theta, \tfrac{\theta^2 + 2}{\theta^2 + 1} \right)},
\]
by direct application of Theorem~\ref{appendix:thm5} with
\begin{align*}
    X = X' = \frac{1}{N} \sum_{n=1}^{N} \frac{\varphi(y_t ; U_{t,n}, 1)}{\varphi\left(y_t; \theta, \tfrac{\theta^2 + 2}{\theta^2 + 1} \right)} \quad \text{ and } \quad
    Y = Y' = \frac{\varphi(y_t ; U_{t,1}, 1)}{\varphi\left(y_t; \theta, \tfrac{\theta^2 + 2}{\theta^2 + 1} \right)}.
\end{align*}
Hence, by Corollary~\ref{appendix:corollary3}, it follows that \( W_N \preceq_{cx} W_1 \), establishing Assumption~\ref{hyp2}.

Next, we verify Assumption~\ref{hyp7}. We have that for all $N\geq 1$
 \begin{align*}
    W_{N,\theta} &= \frac{\widehat p_{T,N}(y| \theta, U)}{p_{T}(y| \theta)}
    = \prod_{t=1}^{T} \frac{1}{N} \sum_{n=1}^{N} 
    \frac{\varphi(y_t ; U_{t,n}, 1)}{\varphi\left(y_t; \theta, \tfrac{\theta^2 + 2}{\theta^2 + 1} \right)}, \quad \text{ where }  U_{t,n} \sim_{\text{iid}} \mathcal{N}\left(\theta, \frac{1}{\theta^2 + 1}\right)
\end{align*}
Let's denote 
$$W_{t,N,\theta} :=  \frac{1}{N} \sum_{n=1}^{N} 
    \frac{\varphi(y_t ; U_{t,n}, 1)}{\varphi\left(y_t; \theta, \tfrac{\theta^2 + 2}{\theta^2 + 1} \right)}$$
We have that ${\varphi(y_t ; U_{t,n}, 1)}/{\varphi\left(y_t; \theta, \tfrac{\theta^2 + 2}{\theta^2 + 1} \right)}$ are iid and for each $N \geq 1$
$$\E\left[W_{t,N,\theta}\right] = 1 < \infty.$$ 
So by the weak law of large numbers we have that for all $t \in \{1, \cdots, T\}$ 
$$ W_{t,N,\theta} \underset{N \to \infty}{\longrightarrow} 1, \quad \text{ in Probability }$$
By the continuity of the product, 
$$ W_{N,\theta} \underset{N \to \infty}{\longrightarrow} 1, \quad \text{ in Probability }$$

\noindent
Let's check Assumption~\ref{hyp8}. Define
\(
g_N(\theta)=\E\left[|W_{N,\theta}-1|\right].
\)
Let $K\subset\mathbb{R}$ be a compact set. We want to show that, for every sequence $(\theta_m)_{m\geq 1}$ in $K$ such that $\theta_m\to\theta$ as $m\to\infty$,
\[
\E\left[|W_{N,\theta_m}-1|\right]
\longrightarrow
\E\left[|W_{N,\theta}-1|\right].
\]
This will prove the continuity of the function $g_N$ on $K$.

For $\theta\in K$, the expectation can be written as
\begin{align*}
g_N(\theta)
&=
\int_{\reals^{TN}}
\left|W_{N,\theta}\left(u\right)-1\right|
\prod_{t=1}^{T}\prod_{n=1}^{N}
\varphi\left(
u_{t,n};
\theta,\frac{1}{\theta^2+1}
\right)
\,\mathrm{d}u ,
\end{align*}
where $u=(u_{t,n})_{t,n}$.

For each $t,n$, with the change of variables
\[
u_{t,n}
=
\theta+\frac{z_{t,n}}{\sqrt{\theta^2+1}},
\qquad
z_{t,n}\in\mathbb{R},
\]
the corresponding Jacobian is
\[
\mathrm{d}u_{t,n}
=
\frac{1}{\sqrt{\theta^2+1}}\,\mathrm{d}z_{t,n}.
\]
Moreover,
\begin{align*}
\varphi\left(
u_{t,n};
\theta,\frac{1}{\theta^2+1}
\right)\mathrm{d}u_{t,n}
&=
\sqrt{\frac{\theta^2+1}{2\pi}}
\exp\left\{
-\frac{\theta^2+1}{2}(u_{t,n}-\theta)^2
\right\}
\frac{1}{\sqrt{\theta^2+1}}\,\mathrm{d}z_{t,n}\\
&=
\frac{1}{\sqrt{2\pi}}
\exp\left\{-\frac{z_{t,n}^2}{2}\right\}\mathrm{d}z_{t,n}\\
&=
\varphi(z_{t,n};0,1)\,\mathrm{d}z_{t,n}.
\end{align*}
Hence,
\begin{align*}
g_N(\theta)
&=
\int
\left|
W_{N,\theta}\left(
\theta+\frac{z}{\sqrt{\theta^2+1}}
\right)-1
\right|
\prod_{t=1}^{T}\prod_{n=1}^{N}
\varphi(z_{t,n};0,1)\,\mathrm{d}z,
\end{align*}
where
\(z=\{z_{t,n}\}_{\{1\leq t\leq T,\ 1\leq n\leq N\}}.\)

For every fixed $z$, the mapping
\[
\theta\mapsto
\theta+\frac{z}{\sqrt{\theta^2+1}}
\]
is continuous. Since
\[
\theta\mapsto
\varphi\left(
y_t;\theta,\frac{\theta^2+2}{\theta^2+1}
\right)
\]
is continuous and strictly positive, it follows that
\[
\theta\mapsto
W_{N,\theta}\left(
\theta+\frac{z}{\sqrt{\theta^2+1}}
\right)
\]
is continuous. Therefore,
\[
\left|
W_{N,\theta_m}\left(
\theta_m+\frac{z}{\sqrt{\theta_m^2+1}}
\right)-1
\right|
\longrightarrow
\left|
W_{N,\theta}\left(
\theta+\frac{z}{\sqrt{\theta^2+1}}
\right)-1
\right|
\]
for every fixed $z$.

It remains to establish an integrable dominating function. For every
$\theta\in K$ and every $z_{t,n}\in\mathbb{R}$,
\[
\varphi\left(
y_t;
\theta+\frac{z_{t,n}}{\sqrt{\theta^2+1}},
1
\right)
\leq \frac{1}{\sqrt{2\pi}}.
\]
Furthermore, since
\[
\theta\mapsto
\varphi\left(
y_t;\theta,\frac{\theta^2+2}{\theta^2+1}
\right)
\]
is continuous and strictly positive, we have
\[
c_{t,K}
:=
\inf_{\theta\in K}
\varphi\left(
y_t;\theta,\frac{\theta^2+2}{\theta^2+1}
\right)>0.
\]
Consequently,
\begin{align*}
W_{N,\theta}
\left(
\theta+\frac{z}{\sqrt{\theta^2+1}}
\right)
&\leq
\prod_{t=1}^{T}
\frac{1}{\sqrt{2\pi}c_{t,K}}
=:C_K<\infty.
\end{align*}
Thus,
\[
\left|
W_{N,\theta}
\left(
\theta+\frac{z}{\sqrt{\theta^2+1}}
\right)-1
\right|
\leq C_K+1.
\]
Since
\(
(C_K+1)
\prod_{t=1}^{T}\prod_{n=1}^{N}
\varphi(z_{t,n};0,1)
\)
is integrable, the dominated convergence theorem gives
\[
g_N(\theta_m)\longrightarrow g_N(\theta),
\]
which proves the continuity of $g_N$ in $\theta$.

To verify Assumption~\ref{hyp4} in the context of the synthetic example, it suffices to establish that, for all \( (\theta, u) \in \Theta \times \mathcal{U} \) and for all \( N \geq N_0 \),
\begin{align*}
    \E[\alpha_N(\theta, u; \vartheta, V) |\theta, u] 
   &= \int_{\Theta \times \mathcal{U}} q(\vartheta |\theta) m_{N,\vartheta}(v)
    \alpha_N(\theta, u; \vartheta, v) \, d\vartheta \, dv < 1,
\end{align*}
where \( \alpha_N(\theta, u; \vartheta, V) = \min\{1, \Xi_N(\theta, u; \vartheta, V)\} \) denotes the acceptance probability, and \( \Xi_N(\theta, u; \vartheta, V) \) is the associated acceptance ratio for the PM algorithm:
\begin{align*}
\Xi_N(\theta, u; &\vartheta, V) =\frac{ \prod_{t=1}^T \frac{1}{N} \sum_{n=1}^N \varphi(y_t; V_{t,n}, 1) \cdot \varphi(\vartheta; 0, \sigma_0^2)}{ \prod_{t=1}^T \frac{1}{N} \sum_{n=1}^N \varphi(y_t; U_{t,n}, 1) \cdot \varphi(\theta; 0, \sigma_0^2) },
\end{align*}
where \( \vartheta \sim \mathcal{N}(\theta, 8/T) \), and the variables \( V_{t,n} \sim \mathcal{N}(\vartheta, 1/\{\vartheta^2 + 1\}) \) are independently drawn for \( t \in \{1, \ldots, T\} \), \( n \in \{1, \ldots, N\} \).
Define the set
\[
\Omega_N^{\theta, u} := \left\{ (\vartheta, V) \in \Theta \times \mathcal{U} : \Xi_N(\theta, u; \vartheta, V) < 1 \right\}.
\]
The conditional expectation may then be decomposed as
{
\begin{align*}
\E[\alpha_N(\theta, u; \vartheta, V) |\theta, u] = \E[\Xi_N(\theta, u; \vartheta, V) \, \ind_{\Omega_N^{\theta, u}} |\theta, u] + \Prob(\bar{\Omega}_N^{\theta, u} |\theta, u).
\end{align*}
}To show that this quantity is strictly less than one, consider the event $A := \{(\vartheta, V) : \|\vartheta\| > \|\theta\| 
 \text{ and } |V_{t,n} - y_t| > |U_{t,n} - y_t|  \text{ for all } t, n \}.$
On this event, each component of the numerator in \( \Xi_N \) is strictly smaller than the corresponding term in the denominator:
\begin{align*}
\varphi(\vartheta; 0, \sigma_0^2) < \varphi(\theta; 0, \sigma_0^2) \quad \text{ and } \quad \varphi(y_t; V_{t,n}, 1) < \varphi(y_t; U_{t,n}, 1), \quad \text{for all } t, n,
\end{align*}
and hence \( \Xi_N(\theta, u; \vartheta, V) < 1 \), implying that \( A \subseteq \Omega_N^{\theta, u} \). Since the proposal distribution, in this setting, is a product of continuous Gaussian densities, the probability of event \( A \) is strictly positive and
\[
\Prob(\Omega_N^{\theta, u} |\theta, u) \geq \Prob(A |\theta, u) > 0.
\]
It follows that, since \( \Prob(\Omega_N^{\theta, u} |\theta, u) > 0 \),
{
\begin{align*}
\E[\alpha_N(\theta, u; \vartheta, V) |\theta, u]&= \E[\Xi_N(\theta, u; \vartheta, V) \, \ind_{\Omega_N^{\theta, u}} |\theta, u]+ \Prob(\bar{\Omega}_N^{\theta, u} |\theta, u) \\
&< \Prob(\Omega_N^{\theta, u} |\theta, u) + \Prob(\bar{\Omega}_N^{\theta, u} |\theta, u) = 1.
\end{align*}
}Therefore, Assumption~\ref{hyp4} is satisfied.

In summary, the verification of all seven assumptions required by Theorem~\ref{thm1} has been completed for the synthetic example described above. As a result, the APM algorithm is theoretically guaranteed to be ergodic in this setting.
\end{proof}

\section{Synthetic Data Example: Simulations}
\label{appendixD}

\subsection{Common Parameters and Fair Comparison Between Non Adaptive and APM methods}
\label{appendixD1}
To ensure a fair and reproducible comparison, the APM and the non adaptive methods were implemented using consistent coding practices, with shared components reused when applicable. Simulations were conducted in the same computational environment (Linux kernel 5.14, AMD Ryzen 9 5950X, 62~GB RAM) using \texttt{R} version 4.2.1, random seeds were fixed to ensure reproducibility, and execution times were recorded via \texttt{Sys.time}.

Key parameters were aligned across implementations to support a meaningful comparison. The APM algorithm (Algorithm~\ref{algo:apm}) was run for \(10^6\) iterations, matching iterations used in~\ref{step3} of the non adaptive method. The same burn-in of $2 \cdot 10^5$ (20\% of total samples) was considered for the APM algorithm and for the final run (~\ref{step3}) of the non adaptive method. The burn-in was determined through Geweke diagnostics (see \cite{geweke_1992})  by comparing the first 20\% versus last 50\% of chains. Table~\ref{tab:method_comparison} in Appendix~\ref{appendixD2} summarizes the corresponding settings.

\begin{table}[H]
\caption{Analogous parameter settings used in the MH, non adaptive and APM methods}
    \centering
    \renewcommand{\arraystretch}{1}
    {\small
    \begin{tabular}[\linewidth]{l|c|c|c}
        \textbf{Parameter} & \textbf{MH} &\textbf{Non adaptive methods} & \textbf{APM} \\
        \hline
        \multicolumn{4}{l}{\textbf{Algorithm configuration}} \\
        \hline
        Number of iterations & $10^6$& $10^6$ (\ref{step3}) & $10^6$ \\
        Burn-in & $2\cdot10^5$ & $2\cdot10^5$ (\ref{step3}) & $2\cdot10^5$ \\
        Initial parameter value & $\theta_0 = 0$ & $\theta_0 = 0$ (~\ref{step1}) & $\theta_0 = 0$ \\
        Initial number of particles &-& $N_1 = 100$ (\ref{step1}) & $N_0 = 100$ \\
        Step size / Precision & - &$a_1 = 1$ (~\ref{step2}) & $a = 1$ \\
        \hline
        \multicolumn{4}{l}{\textbf{Implementation details}} \\
        \hline
        Instrumental distribution & Same as in \eqref{intru}&Same as in \eqref{intru} (\ref{step1}) & Same as in \eqref{intru}\\
    \end{tabular}}
    \label{tab:method_comparison}
\end{table}

\subsection{Preliminary Run Results for the Non Adaptive Methods}
\label{appendixD2}
For each run, a preliminary execution of the PM algorithm was conducted using \( N_1 = 100 \) particles and the proposal variance specified in Equation~\eqref{intru}. Summary results of all runs are reported in Table~\ref{SM:setab1}.
    \begin{table}[h]
    \centering
    \caption{Preliminary run of the PM algorithm on the synthetic example with \(10^5\) iterations and \(N_1 = 100\). Reported values are means and standard deviations over 10 independent runs.}
    \begin{tabular}{l@{\hspace{1.6cm}}cc}
    \hline
    Statistic &  Naive Non Adaptive & Non Adaptive   \\
    \hline
    Posterior Mean $\hat\theta_{100}$ & -0.024 ± 0.0022 &$-0.026 \pm 0.0015$ \\
    Posterior Variance $\widehat{\sigma}_{100}^2$ & 0.009 ± 0.0003 &$0.009 \pm 0.0001$ \\
    Acceptance Rate $\widehat P$ (\%) & 15.819 ± 0.2609 &$15.893 \pm 0.2786$ \\
    \hline
    \end{tabular}
    \label{SM:setab1}
    \end{table}

\subsection{Search for the Optimal Number of Particles}
\label{appendixD3}
For both non adaptive methods and for each run, using the estimate \( \hat\theta_{100} \), the standard deviations \( \sigma_N(\hat \theta_{100}) \) of the additive noise were estimated for various values of \( N \). These estimates were obtained via Monte Carlo simulation using the identity in Equation~\eqref{estimationsd} and $10^4$ iterations for each $N$, with the goal of identifying an optimal number of particles \( N_{\text{opt}} \) such that \( \hat \sigma_{N_{\text{opt}}}(\hat \theta_{100}) \approx 1.16 \).

In~\ref{step2} of the naive non adaptive method, we evaluated the standard deviation of the log-likelihood estimator for each $N \in \{100,\ldots,1000\}$, with a step size of $a_1=1$. Figure~\ref{fig:naive_search} shows the estimated standard deviation of the log-likelihood estimator as a function of $N$ for the 10 independent runs. The dashed horizontal line indicates the target value, while the different curves correspond to the independent runs. The optimal values of $N$ are concentrated around $N=202$, with some variability across runs.
\begin{figure}
    \centering
    \includegraphics[width=1\linewidth]{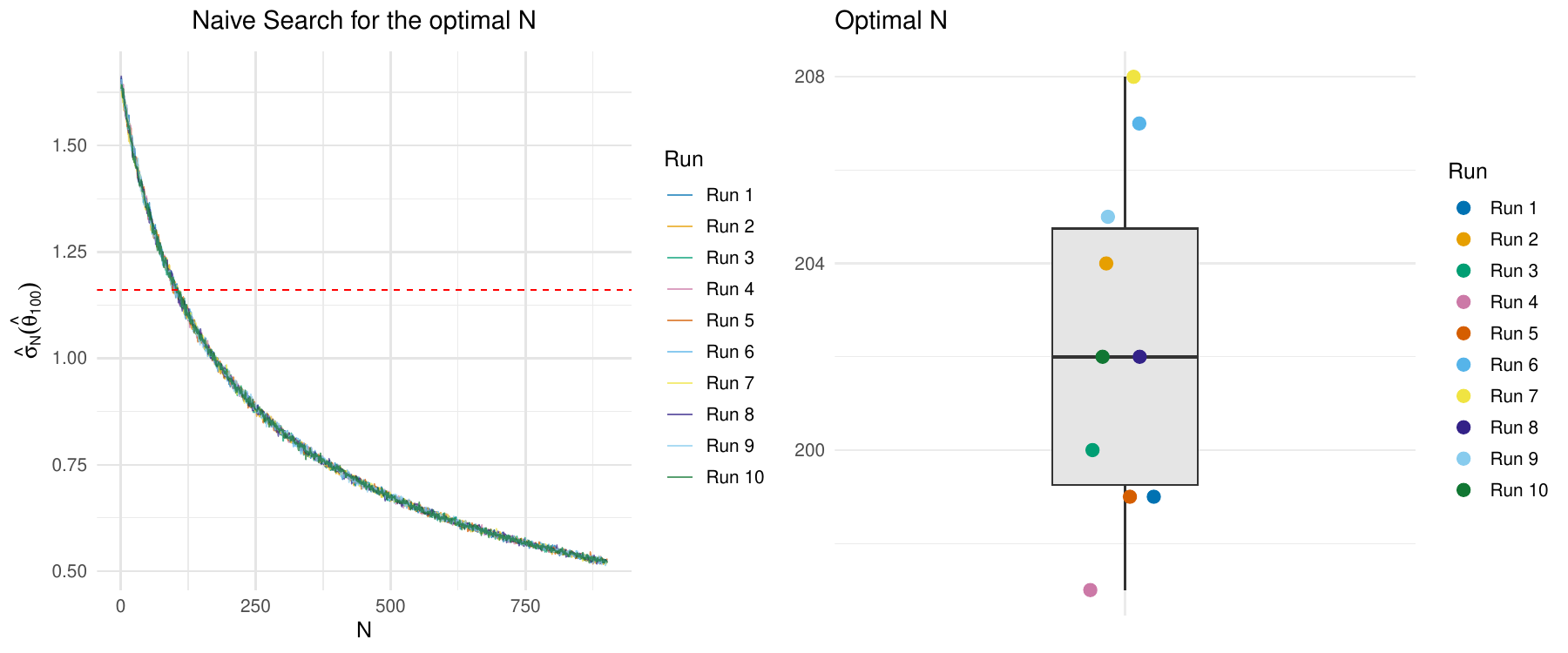}
    \caption{Naive search for the optimal number of particles $N$. Left: estimated standard deviation of the log-likelihood estimator as a function of the particles $N$ for the ten independent runs, with the dashed line indicating the target $\sigma_{\mathrm{opt}}=1.16$. Right: corresponding optimal values of $N$ for each run.}
    \label{fig:naive_search}
\end{figure}

We now detail the implementation of~\ref{step2} in the non adaptive method using dichotomic search. \label{dichotomic}
The \textit{dichotomic search} interval was initialized as \( [100, 1000] \) for the number of particles. At each iteration, the estimate \( \hat\sigma_N(\hat\theta_{100}) \) was computed at the current midpoint of the interval. Initially, evaluations were performed at \( N = 100 \) and \( N = 1000 \). The midpoint value \( N = 550 \) was then tested. If the estimated standard deviation at the midpoint exceeded the target value \( \sigma_{\mathrm{opt}} = 1.16 \), the lower bound of the interval was updated to the midpoint; otherwise, the upper bound was updated accordingly.  

\begin{table}[H]
\centering
\caption{Dichotomic Search Progression for Run 1}
\label{tab:dichotomic_run1}
\begin{tabular}{lllll}
\toprule
\textbf{Iteration} & 
\textbf{Tested $N$} & 
$\mathbf{\hat{\sigma}_N}$ & 
\textbf{Start} & 
\textbf{End} \\
\hline
0 & 100  & 1.661 & 100 & 1000 \\
0 & 1000 & 0.517 & 100 & 1000 \\
1 & 550  & 0.701 & 100 & 550 \\
2 & 325  & 0.913 & 100 & 325 \\
3 & 213  & 1.131 & 100 & 213 \\
4 & 157  & 1.320 & 157 & 213 \\
5 & 185  & 1.206 & 185 & 213 \\
6 & 199  & 1.169 & 199 & 213 \\
7 & 206  & 1.154 & 199 & 206 \\
8 & 203  & 1.161 & 203 & 206 \\
9 & 205  & 1.149 & 203 & 205 \\
10 & 204 & 1.164 & 203 & 204 \\
\bottomrule
\end{tabular}
\end{table}

The search terminated when the length of the final interval, $[203, 204]$, was $a_1 = 1$. The optimal value was found at $N_{\text{opt}} = 203$, where $\hat{\sigma}_{203}(\hat{\theta}_{100}) = 1.161$, which is the closest value to the target of 1.16.

 Detailed values of \( N \) selected by the {dichotomic search} algorithm in~\ref{step2} of the non adaptive method with dichotomic search, along with the corresponding estimates of \( \sigma_N(\hat\theta_{100}) \), are provided in Table~\ref{comp:step2}.

\begin{table}[H]
\centering
\caption{Estimations of $\sigma_N(\hat \theta_{100})$ for multiple independent runs at each $N$, using $10^{4}$ Monte Carlo iterations.}
\label{comp:step2}
{\small
\begin{tabular}{@{}l*{10}{c}@{}}
\toprule
\textbf{N} & \multicolumn{10}{c}{\textbf{Run}} \\
\cmidrule(lr){2-11}
 & 1 & 2 & 3 & 4 & 5 & 6 & 7 & 8 & 9 & 10 \\
\hline
100 & 1.661 & 1.628 & 1.667 & 1.657 & 1.657 & 1.627 & 1.629 & 1.662 & 1.664 & 1.644 \\
157 & 1.320 & 1.328 & 1.311 & 1.323 & 1.316 & 1.310 & 1.316 & 1.324 & 1.297 & 1.304 \\
185 & 1.206 & 1.206 & 1.212 & 1.212 & 1.219 & 1.222 & 1.216 & 1.221 & 1.216 & 1.224 \\
192 & & & & & & & & & & 1.189 \\
196 & & & & & & & & & & 1.181 \\
198 & & & & & & & & & & 1.169 \\
199 & 1.169 & \textbf{1.161} & 1.172 & 1.174 & 1.169 & 1.173 & 1.182 & 1.174 & 1.173 & \textbf{1.158} \\
200 & & {1.163} & {1.165} & & & & & {1.166} & & \\
201 & & 1.156 & {1.157} & & & 1.176 & & 1.149 & & \\
202 & & & & & & 1.147 & & & & \\
203 & \textbf{1.161} & 1.152 & \textbf{1.160} & & 1.166 & \textbf{1.150} & {1.164} & \textbf{1.156} & & \\
204 & {1.164} & & & & \textbf{1.159} & & \textbf{1.163} & & & \\
205 & 1.149 & & & & 1.144 & & & & & \\
206 & 1.154 & 1.155 & 1.129 & \textbf{1.161} & 1.157 & 1.138 & 1.144 & 1.126 & \textbf{1.162} & \\
207 & & & & 1.140 & & & & & 1.144 & \\
208 & & & & 1.133 & & & & & 1.131 & \\
210 & & & & 1.133 & & & & & 1.136 & \\
213 & 1.131 & 1.137 & 1.138 & 1.128 & 1.132 & 1.132 & 1.131 & 1.129 & 1.124 & 1.136 \\
325 & 0.913 & 0.902 & 0.914 & 0.926 & 0.920 & 0.916 & 0.921 & 0.922 & 0.915 & 0.906 \\
550 & 0.701 & 0.707 & 0.697 & 0.707 & 0.716 & 0.707 & 0.714 & 0.707 & 0.710 & 0.701 \\
1000 & 0.517 & 0.521 & 0.521 & 0.522 & 0.525 & 0.523 & 0.526 & 0.522 & 0.516 & 0.521 \\
\hline
\textbf{Optimal N} & 203 & 199 & 203 & 206 & 204 & 203 & 204 & 203 & 206 & 199 \\
\bottomrule
\end{tabular}}
\begin{minipage}{\linewidth}
\raggedright
\textit{Note:} Empty cells indicate $N$ values not tested in that run. Bottom row shows the optimal $N$ (closest to 1.16) for each run.
\end{minipage}
\end{table}

\subsection{Additional Convergence Diagnostics}
\label{appendixD4}

Figure~\ref{se:conv} presents convergence diagnostics for the APM and the PM (in~\ref{step3} of the non adaptive method with dichotomic search) algorithms. The trace plots (left panel) show the evolution of the parameter \( \theta \) over the last 100 iterations for a representative run of each method, indicating similar mixing behavior. The autocorrelation functions (right panel) also reveal comparable levels of dependence across iterations. 

\begin{figure}
  \begin{minipage}{0.5\textwidth}
    \centering
    \includegraphics[width=\linewidth]{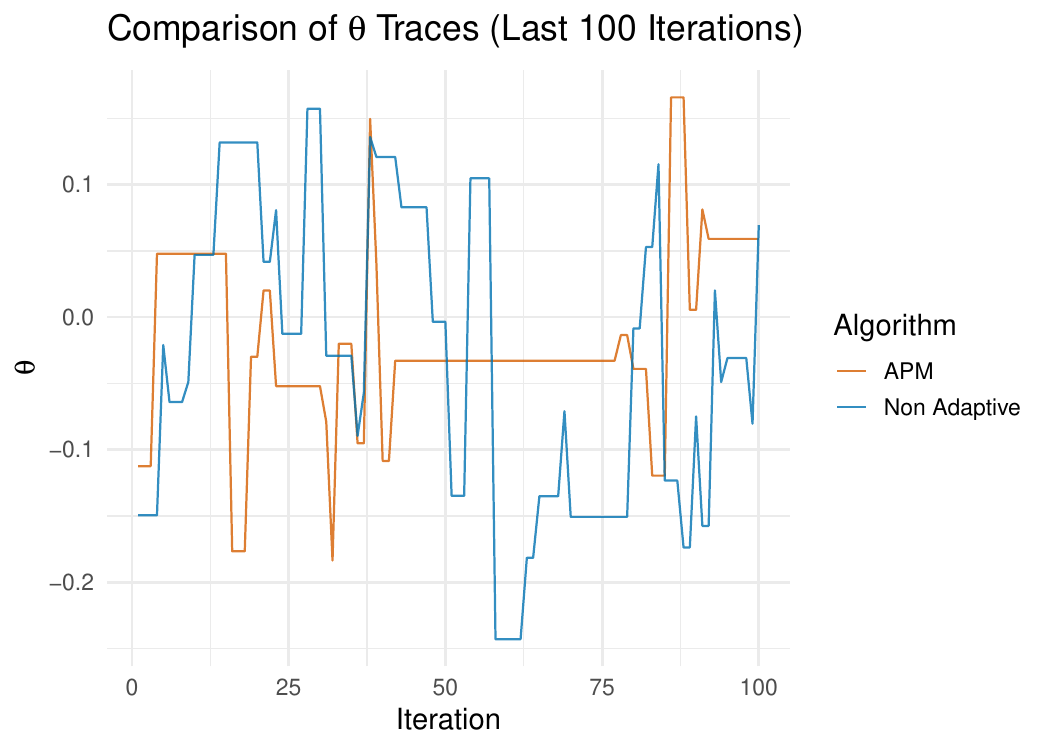}
  \end{minipage}
  \begin{minipage}{0.48\textwidth}
    \centering
    \vspace{0.1cm}
    \includegraphics[width=\linewidth]{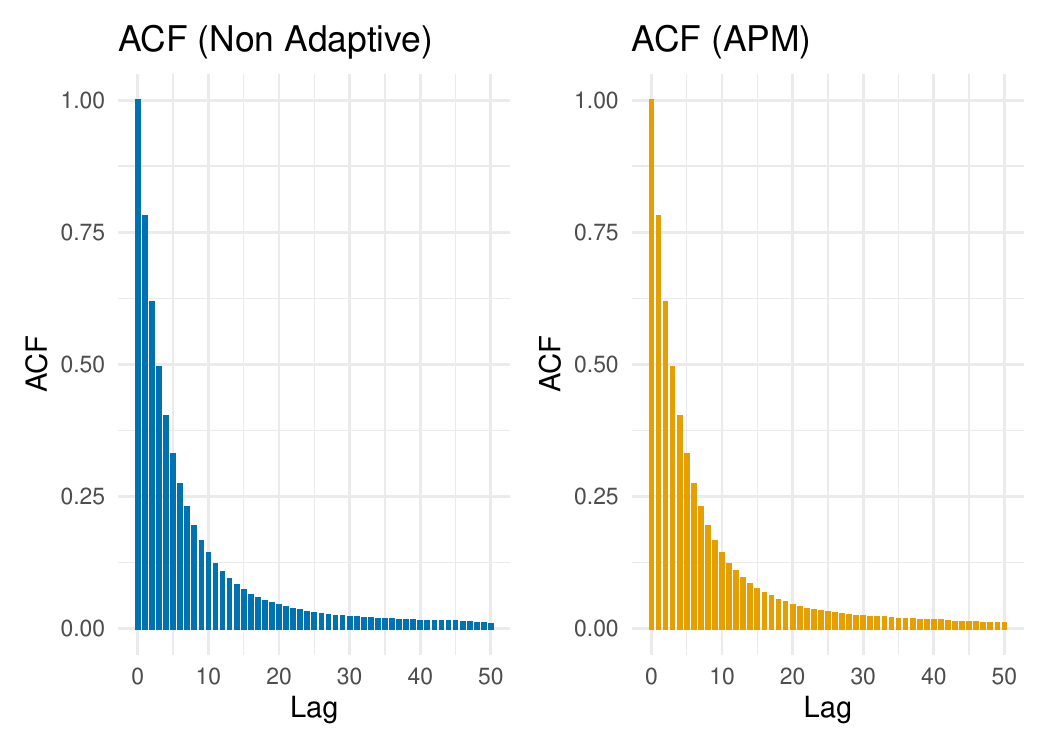}
  \end{minipage}
  \caption{Top: Trace plots of \( \theta \) over the last 100 iterations for a single run of the APM and non adaptive method with dichotomic search. Bottom: Autocorrelation figure for the same run using a Burn-in of $2\cdot10^5$.}
  \label{se:conv}
\end{figure}

\section{Real Data Study: Simulations}
\label{appendixE}

\subsection{Analogous parameters used for implementing the non adaptive and APM methods}
\label{appendixE1}
Table~\ref{tab:method_comparisonRE} presents the key parameters used in both the non adaptive and APM methods and highlights their correspondence to ensure a fair comparison.
\begin{table}
\caption{Analogous parameter settings used in the non adaptive and APM methods for the real data example}
    \centering
    \renewcommand{\arraystretch}{1.3}
    {\small
    \begin{tabular}{l|c|c}
        \textbf{Parameter} & \textbf{Non adaptive methods} & \textbf{APM} \\
        \hline
        \multicolumn{3}{l}{\textbf{Algorithm configuration}} \\
        \hline
        Number of iterations & $10^6$ (\ref{step3}) & $10^6$ (total chain length $L$) \\
        Burn-in & $4 \cdot 10^5$ (\ref{step3}) & $4 \cdot 10^5$ (out of total chain length $L$) \\
        Initial parameter value & $\theta_0 $ (see~\eqref{re-theta0}) & $\theta_0$ \\
        Initial number of particles & $N_1 = 10$ (\ref{step1}) & $N_0 = 10$ (initial) \\
        Step size / Precision & $a_1 = 1$ (\ref{step2}) & $a = 1$ \\
        \hline
        \multicolumn{3}{l}{\textbf{Implementation details}} \\
        \hline
        Instrumental distribution & Same as in \eqref{intru2} (\ref{step1}) & Same as in \eqref{intru2}\\
    \end{tabular}}
    \label{tab:method_comparisonRE}
\end{table}
\noindent
Here, 
\begin{align}
    \theta_0= (-2.788, -0.035, 0.560, -0.614, -0.173, -0.461, -0.052, 0.192, 0.944),
    \label{re-theta0}
\end{align} 
and 
{\small
\begin{equation*}
\Sigma_p = 
\begin{bmatrix}
  0.0530  &  0.0003  & -0.0211  &  0.0149  &  0.0103  & -0.0251  & -0.0009  & -0.0343  & -0.0384 \\ 
  0.0003  &  0.0001  & -0.0004  &  0.0000  &  0.0001  &  0.0001  &  0.0001  & -0.0003  & -0.0003 \\
 -0.0211  & -0.0004  & {0.2570}  & -0.0103  & -0.0065  &  0.0112  &  0.0000  & -0.0094  & -0.0119 \\
  0.0149  &  0.0000  & -0.0103  &  0.0318  &  0.0070  &  0.0001  &  0.0003  &  0.0050  & -0.0033 \\
  0.0103  &  0.0001  & -0.0065  &  0.0070  &  0.0321  & -0.0006  &  0.0004  & -0.0005  &  0.0000 \\
 -0.0251  &  0.0001  &  0.0112  &  0.0001  & -0.0006  &  0.0761  &  0.0002  & -0.0015  & -0.0075 \\
 -0.0009  &  0.0001  &  0.0000  &  0.0003  &  0.0004  &  0.0002  &  0.0008  &  0.0071  & -0.0011 \\
 -0.0343  & -0.0003  & -0.0094  &  0.0050  & -0.0005  & -0.0015  &  0.0071  & {0.2169}  &  0.0064 \\
 -0.0384  & -0.0003  & -0.0119  & -0.0033  &  0.0000  & -0.0075  & -0.0011  &  0.0064  & {0.1348}
\end{bmatrix}
\end{equation*}}

\subsection{Preliminary Run Results for both non adaptive methods}
\label{appendixE2}
  \begin{table}[h]
    \centering
    \caption{Preliminary run of the PM algorithm on the real data example with \(10^5\) iterations and \(N_1 = 10\). Reported values are the mean and standard deviation of euclidean norms over 10 independent runs.}
    \begin{tabular}{l@{\hspace{0.5cm}}cc}
    \hline
    Statistic & Naive Non Adaptive& Non adaptive \\
    \hline
    Norm of Posterior Mean $\hat\theta_{10}$ & 3.105 $\pm$ 0.0161 & $3.091 \pm 0.0163$ \\
    Norm of Posterior Covariance $\widehat{\Sigma}_{10}$ & 0.398 $\pm$ 0.0155 &$0.394 \pm 0.0102$ \\
    Acceptance Rate $\widehat P$ (\%) & 7.335 $\pm$ 0.3874 &$7.169 \pm 0.4194$ \\
    \hline
    \end{tabular}
    \label{re:tab1}
    \end{table}

\subsection{Estimations of \texorpdfstring{$\sigma_N$}{} for multiple \texorpdfstring{$N$}{} and runs}
\label{appendixE3}

\begin{figure}
    \centering
    \includegraphics[width=1\linewidth]{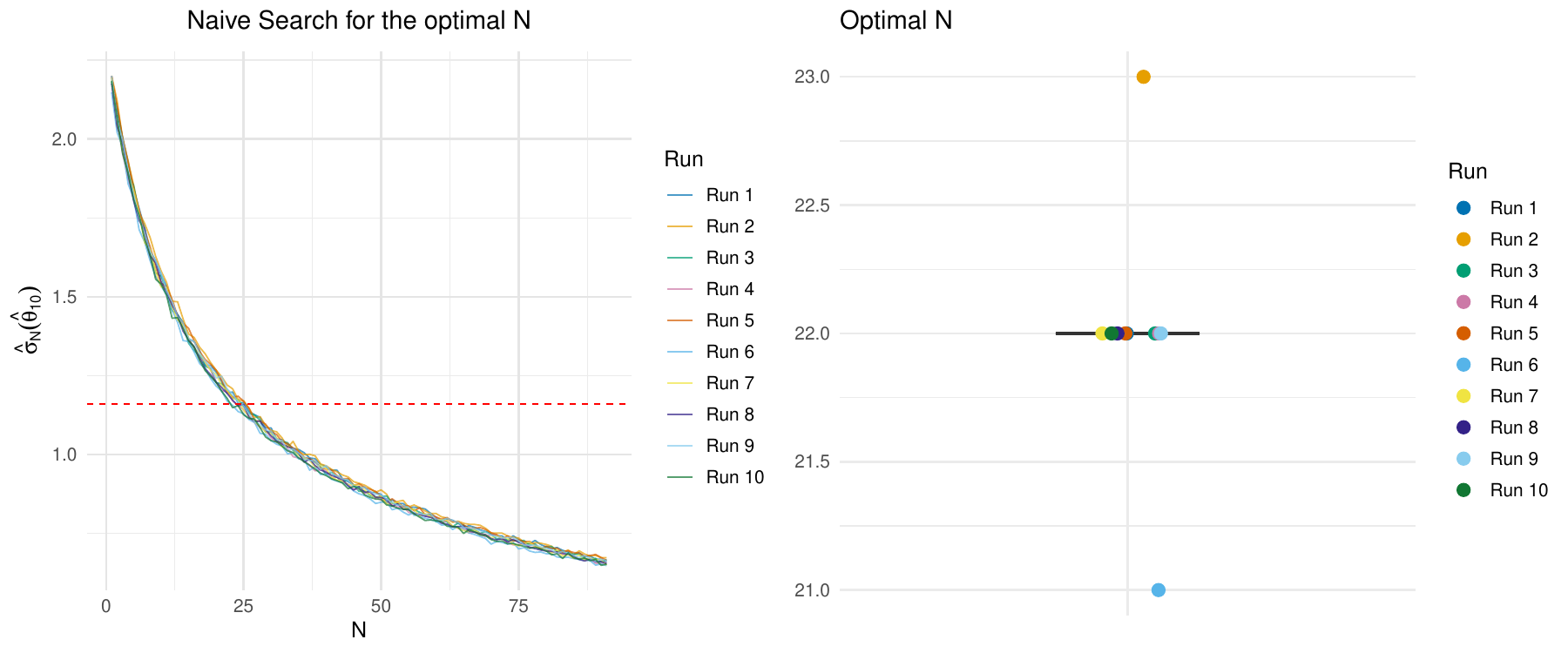}
    \caption{Naive search for the optimal number of particles $N$ in the real data setting. Left: estimated standard deviation of the log-likelihood estimator as a function of the particles $N$ for the ten independent runs, with the dashed line indicating the target $\sigma_{\mathrm{opt}}=1.44$. Right: corresponding optimal values of $N$ for each run.}
    \label{fig:re_naive_search}
\end{figure}

Detailed values of \( N \) selected by the {dichotomic search} algorithm in~\ref{step2} of the non adaptive method applied to the real data example, along with the corresponding estimates of \( \sigma_N(\hat\theta_{10}) \), are provided in Table~\ref{re:step2}.

\begin{table}
\centering
\caption{Estimations of $\sigma_N(\hat \theta_{10})$ for multiple independent runs at each $N$, using $10^{4}$ Monte Carlo iterations.}
\label{re:step2}
{\small
\begin{tabular}{@{}l*{10}{c}@{}}
\toprule
\textbf{N} & \multicolumn{10}{c}{\textbf{Run}} \\
\cmidrule(lr){2-11}
 & 1 & 2 & 3 & 4 & 5 & 6 & 7 & 8 & 9 & 10 \\
\hline
10 & 2.159 & 2.165 & 2.147 & 2.178 & 2.149 & 2.221 & 2.217 & 2.198 & 2.197 & 2.222 \\
21 & 1.470 & \textbf{1.441} & \textbf{1.464} & 1.456 & \textbf{1.448} & 1.498 & 1.493 & 1.473 & \textbf{1.450} & 1.503 \\
22 & \textbf{1.423} & 1.400 & 1.413 & \textbf{1.426} & 1.430 & \textbf{1.455} & \textbf{1.457} & \textbf{1.433} & 1.429 & 1.479 \\
23 & & & & & & 1.418 & 1.404 & & & \textbf{1.419} \\
24 & 1.365 & 1.357 & 1.353 & 1.361 & 1.355 & 1.373 & 1.388 & 1.404 & 1.384 & 1.381 \\
26 & 1.288 & 1.289 & 1.310 & 1.313 & 1.301 & 1.331 & 1.322 & 1.316 & 1.314 & 1.344 \\
32 & 1.169 & 1.165 & 1.161 & 1.179 & 1.165 & 1.190 & 1.173 & 1.177 & 1.177 & 1.204 \\
55 & 0.891 & 0.873 & 0.881 & 0.889 & 0.873 & 0.905 & 0.907 & 0.897 & 0.890 & 0.904 \\
100 & 0.654 & 0.647 & 0.654 & 0.660 & 0.653 & 0.663 & 0.667 & 0.666 & 0.659 & 0.670 \\
\hline
\textbf{Optimal N} & 22 & 21 & 21 & 22 & 21 & 22 & 22 & 22 & 21 & 23 \\
\bottomrule
\end{tabular}}
\vspace{0cm}
\begin{minipage}{\textwidth}
\raggedright
\textit{Note:} Empty cells indicate N values not tested in that run. Bottom row shows the optimal N (with $\hat\sigma_{N}(\hat\theta_{10})$ closest to 1.44) for each run.
\end{minipage}
\end{table}

\subsection{Convergence figures}
\label{appendixE4}
Figure~\ref{re:fig2} shows the evolution of the number of particles \( N \) during each APM run. In all cases, \( N \) gradually approaches the corresponding optimal value, indicating effective adaptation. 

\begin{figure}
    \centering
    \includegraphics[width=0.8\linewidth]{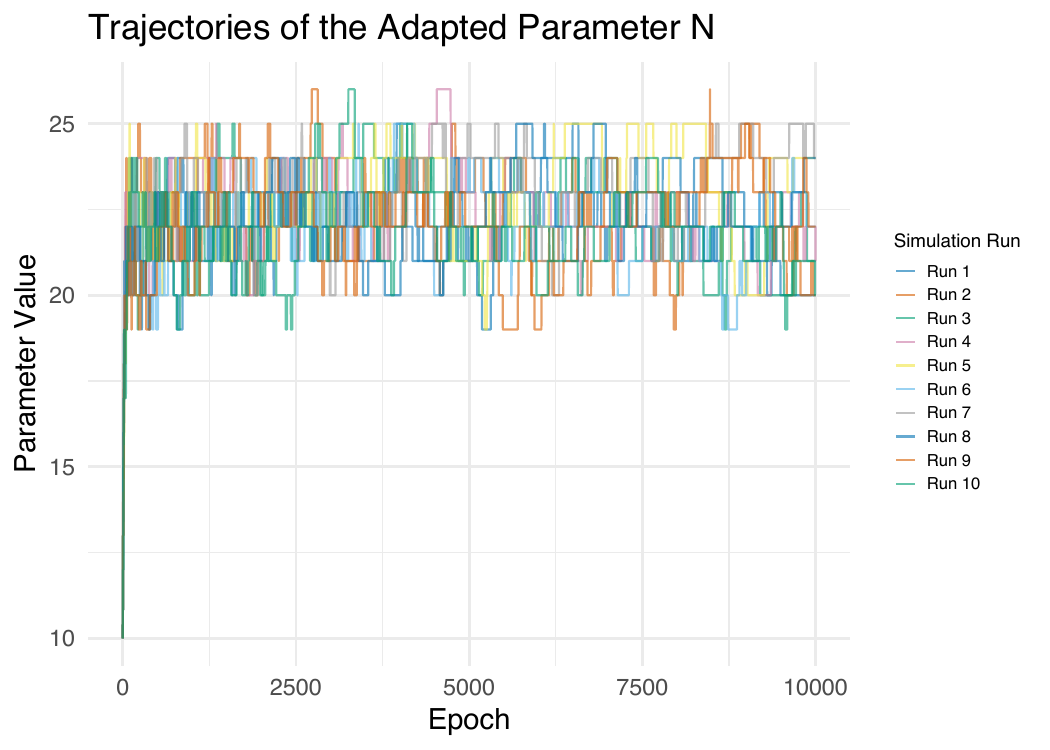}
    \caption{Trace of \( N \) using \( 10^4 \) \textit{epochs} of $100$ iterations across 10 APM runs applied on the real data example.}
    \label{re:fig2}
\end{figure}

Figure~\ref{re:conv9} displays convergence diagnostics for the APM and PM algorithms (the latter in~\ref{step3} of the non adaptive method) in the real data example. Nine panels are shown, one for each component of the parameter. For each component, the trace plots (left) show the evolution of the parameter's component over the last 100 iterations in a representative run, indicating comparable mixing behavior. The autocorrelation functions (right) likewise exhibit similar dependence across iterations.

\begin{figure}[H]
  \begin{minipage}{0.48\textwidth}
    \centering
    \includegraphics[width=\linewidth]{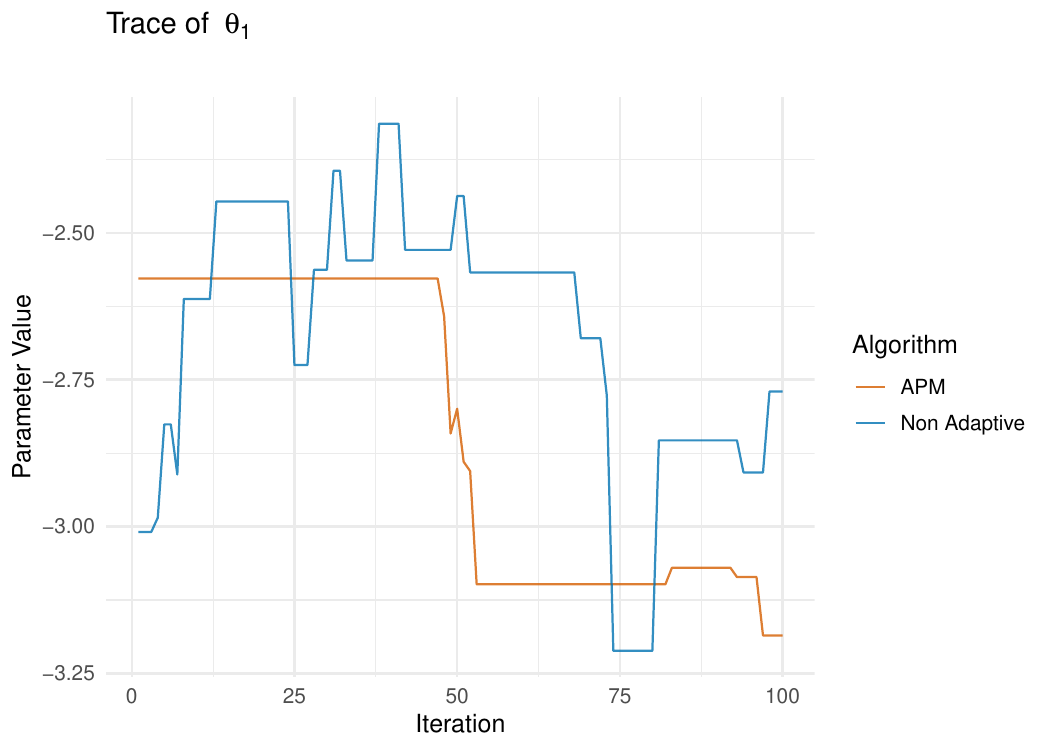}
  \end{minipage}
  \hfill
  \begin{minipage}{0.48\textwidth}
    \centering
    \vspace{0.3cm}
    \includegraphics[width=\linewidth]{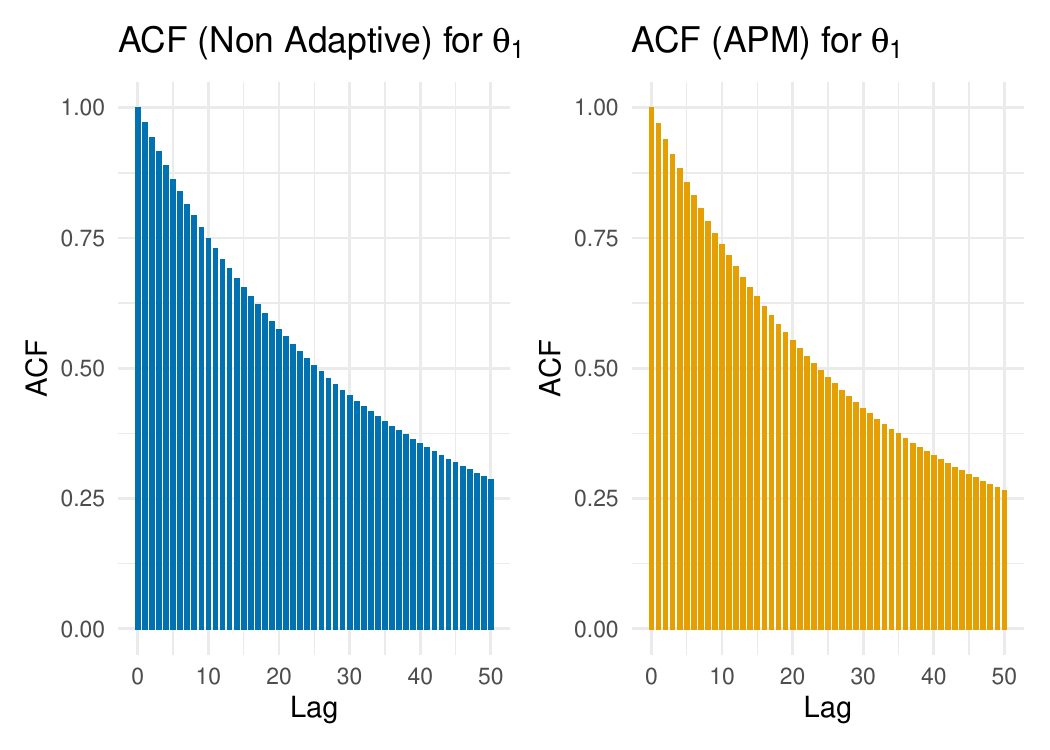}
  \end{minipage}

  \begin{minipage}{0.48\textwidth}
    \centering
    \includegraphics[width=\linewidth]{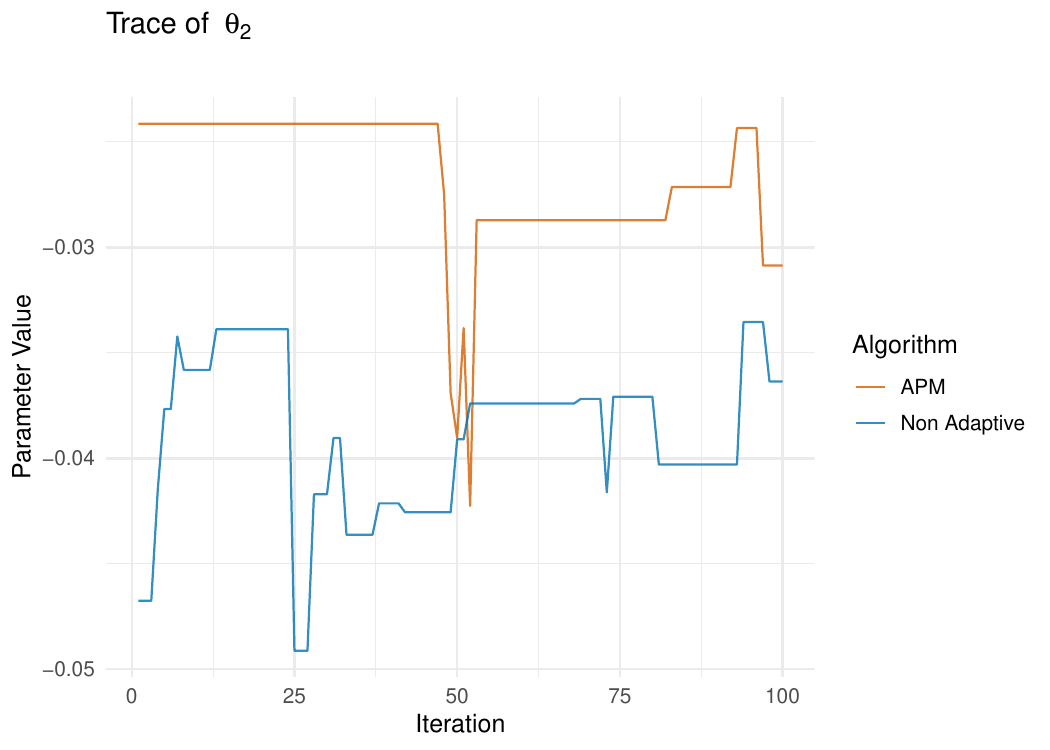}
  \end{minipage}
  \hfill
  \begin{minipage}{0.48\textwidth}
    \centering
    \vspace{0.3cm}
    \includegraphics[width=\linewidth]{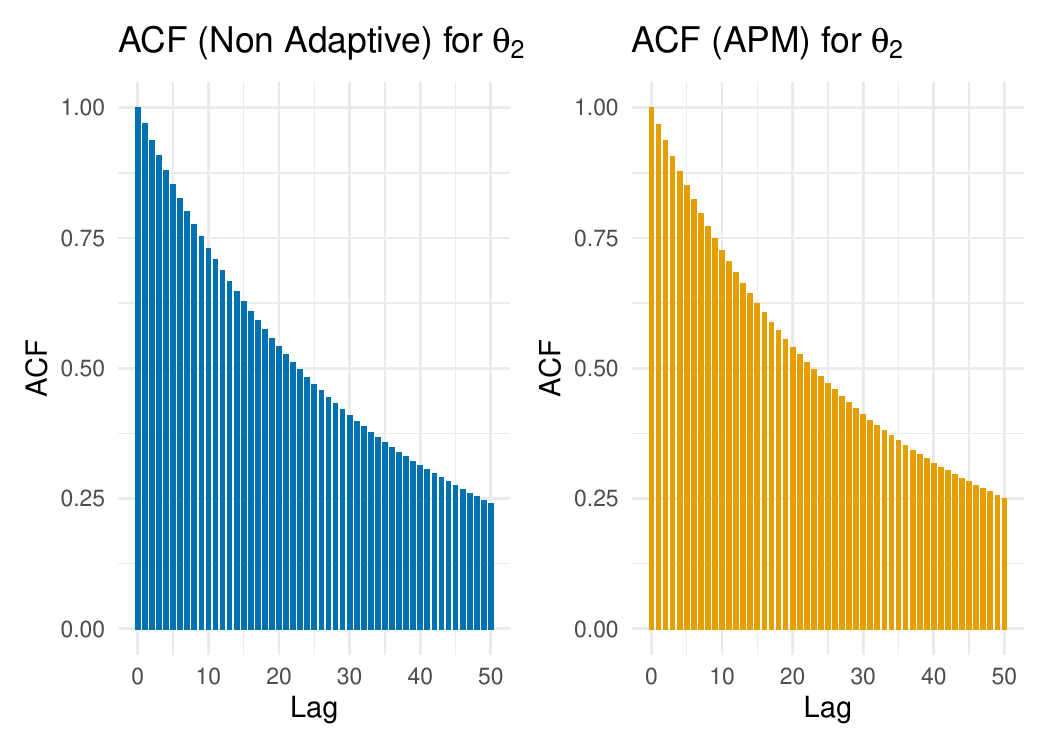}
  \end{minipage}

\end{figure}

\begin{figure}[H]
     \begin{minipage}{0.48\textwidth}
    \centering
    \includegraphics[width=\linewidth]{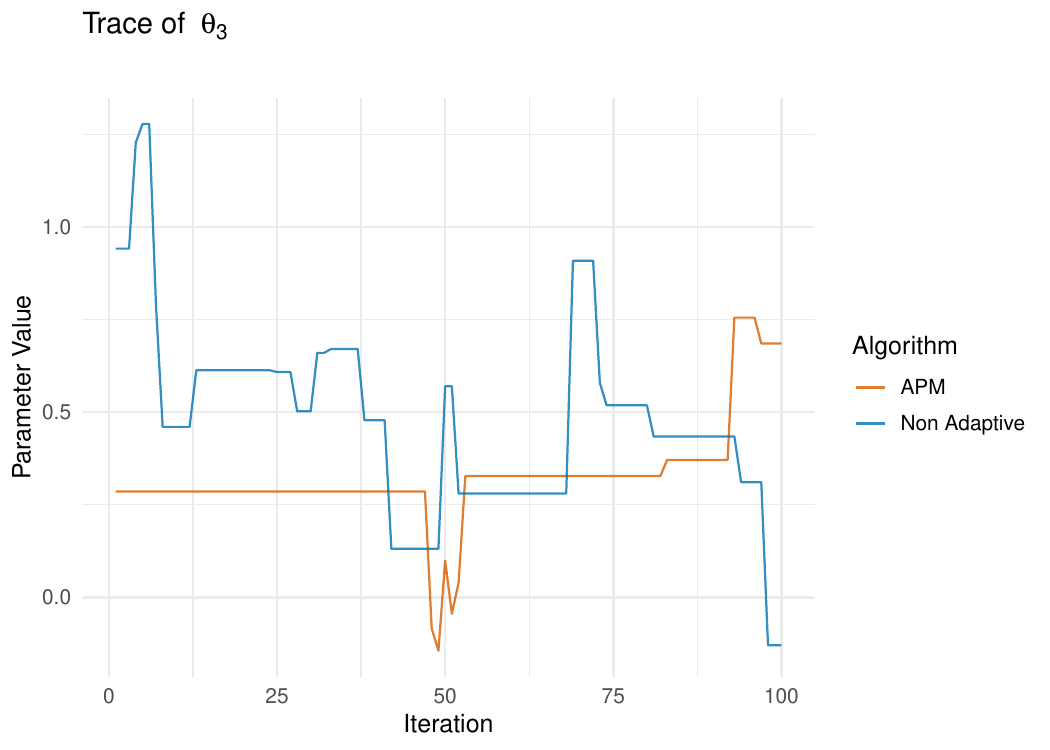}
  \end{minipage}
  \hfill
  \begin{minipage}{0.48\textwidth}
    \centering
    \vspace{0.3cm}
    \includegraphics[width=\linewidth]{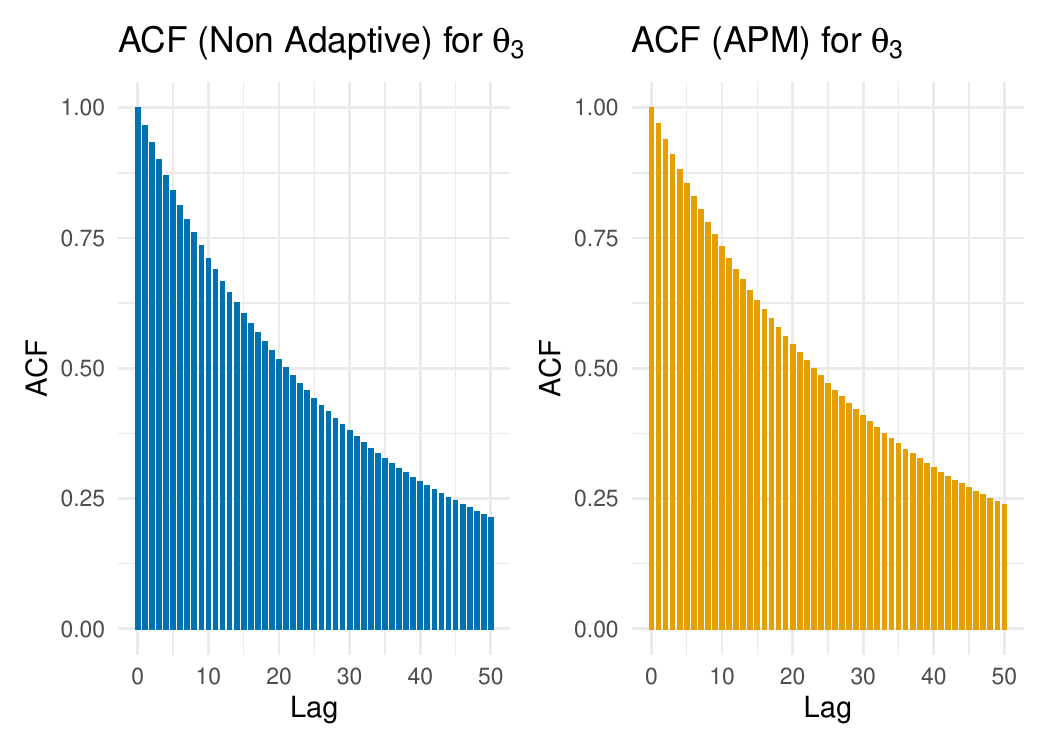}
  \end{minipage}
   \begin{minipage}{0.48\textwidth}
    \centering
    \includegraphics[width=\linewidth]{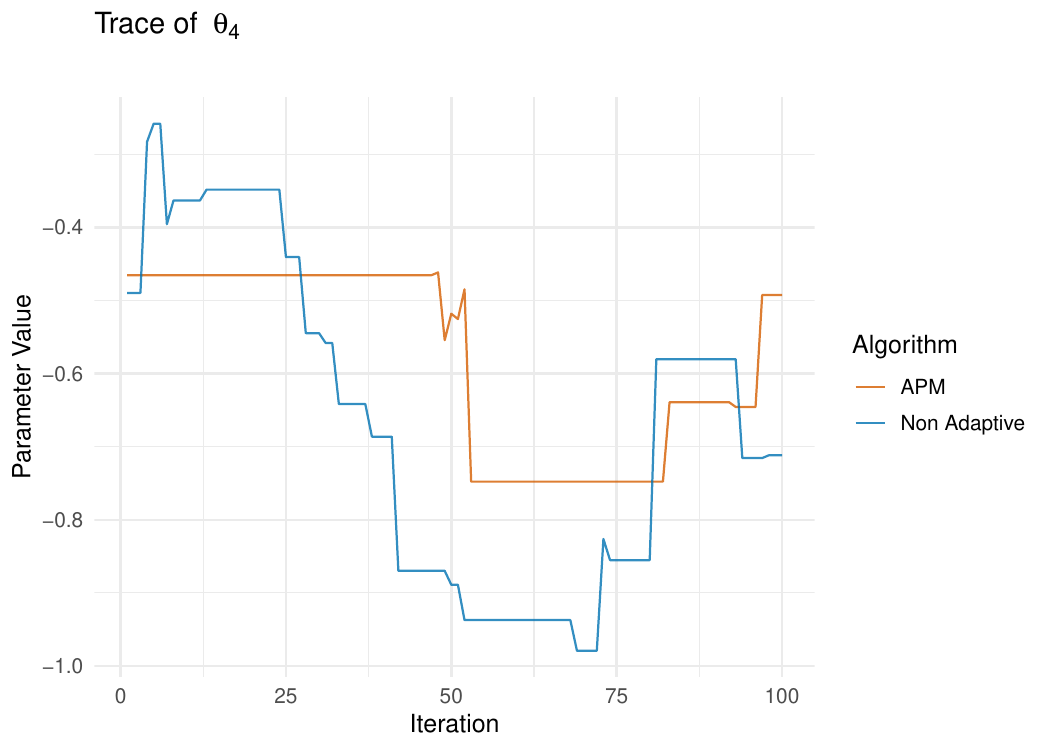}
  \end{minipage}
  \hfill
  \begin{minipage}{0.48\textwidth}
    \centering
    \vspace{0.3cm}
    \includegraphics[width=\linewidth]{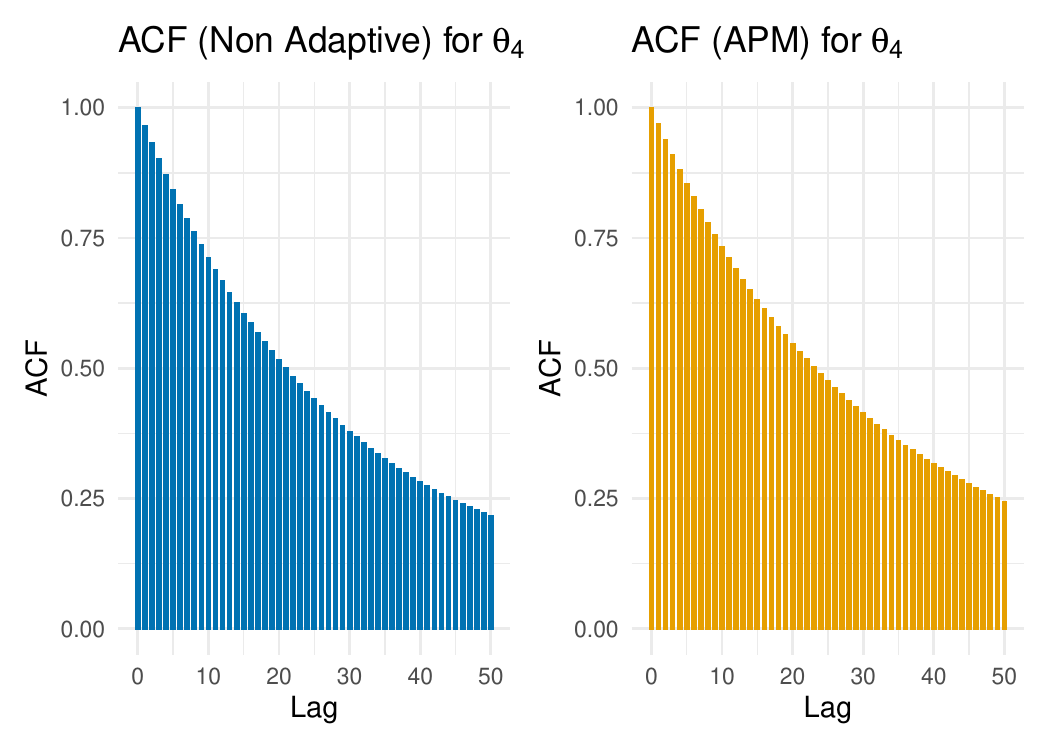}
  \end{minipage}
  
  \begin{minipage}{0.48\textwidth}
    \centering
    \includegraphics[width=\linewidth]{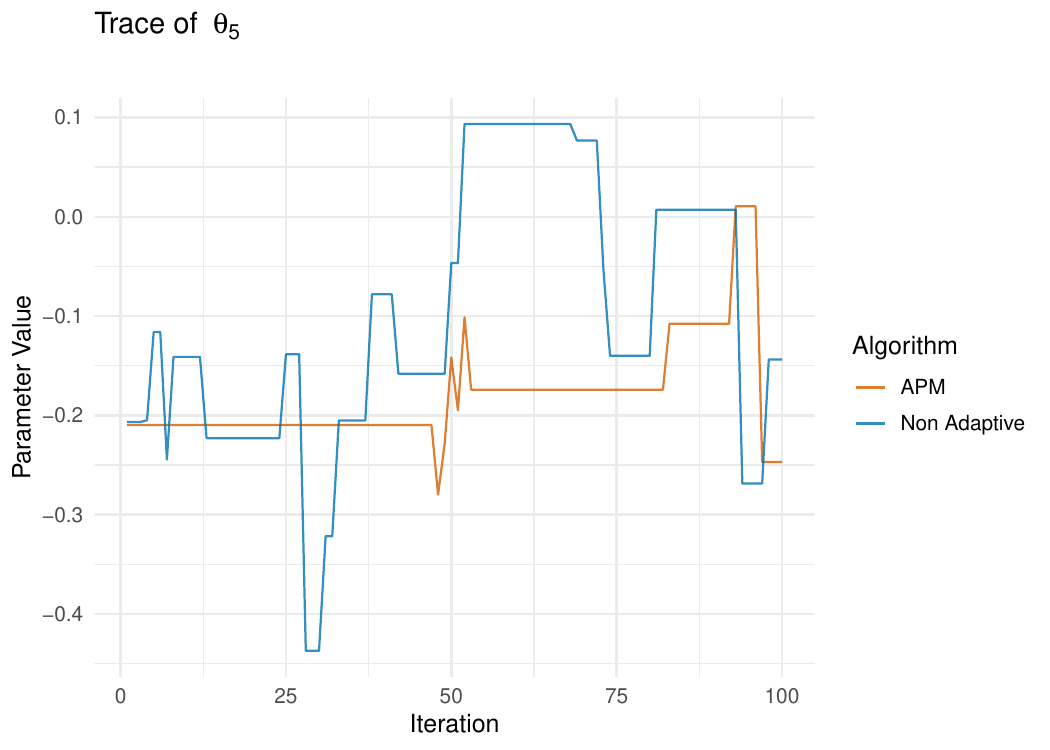}
  \end{minipage}
  \hfill
  \begin{minipage}{0.48\textwidth}
    \centering
    \vspace{0.3cm}
    \includegraphics[width=\linewidth]{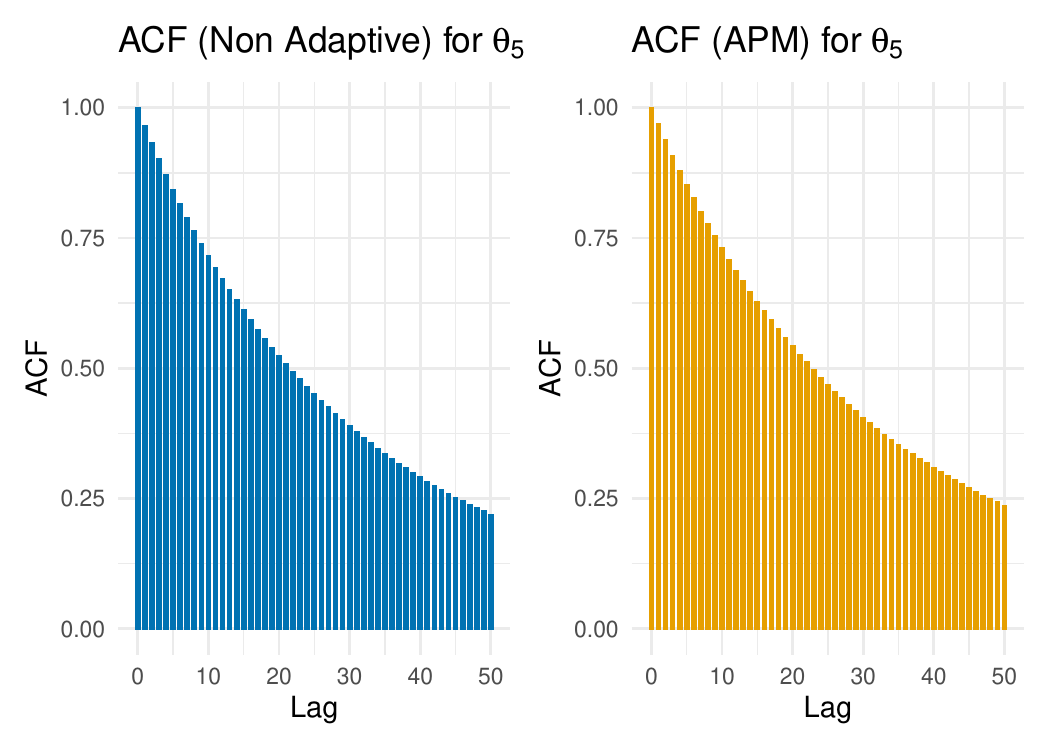}
  \end{minipage}

  \begin{minipage}{0.48\textwidth}
    \centering
    \includegraphics[width=\linewidth]{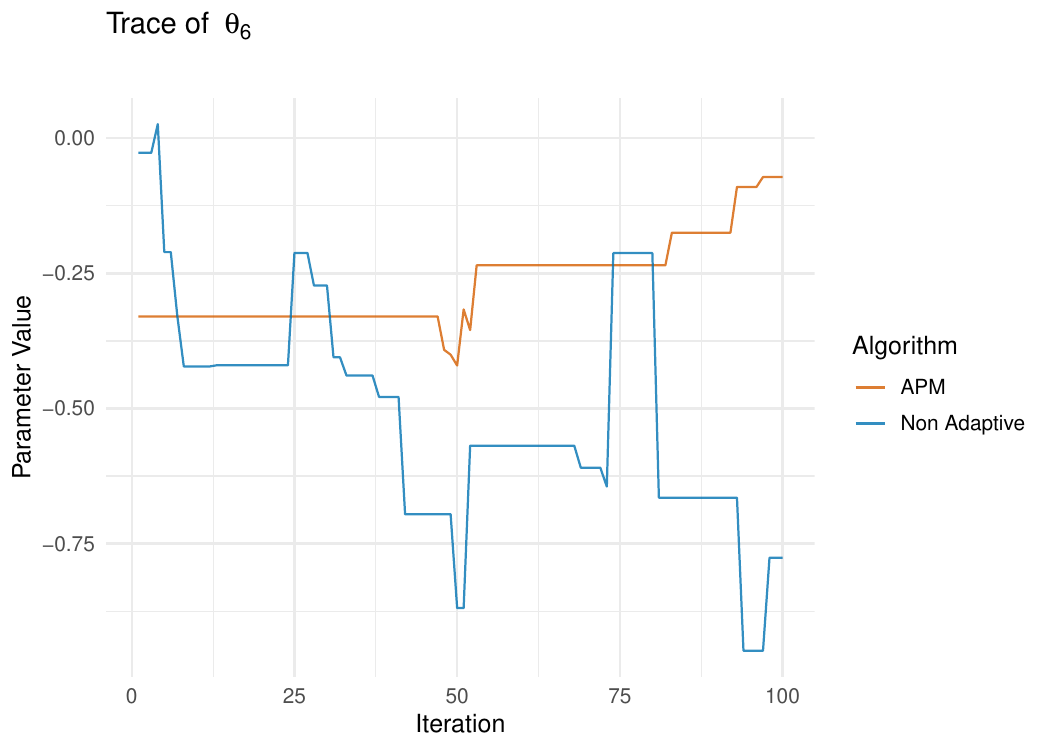}
  \end{minipage}
  \hfill
  \begin{minipage}{0.48\textwidth}
    \centering
    \vspace{0.3cm}
    \includegraphics[width=\linewidth]{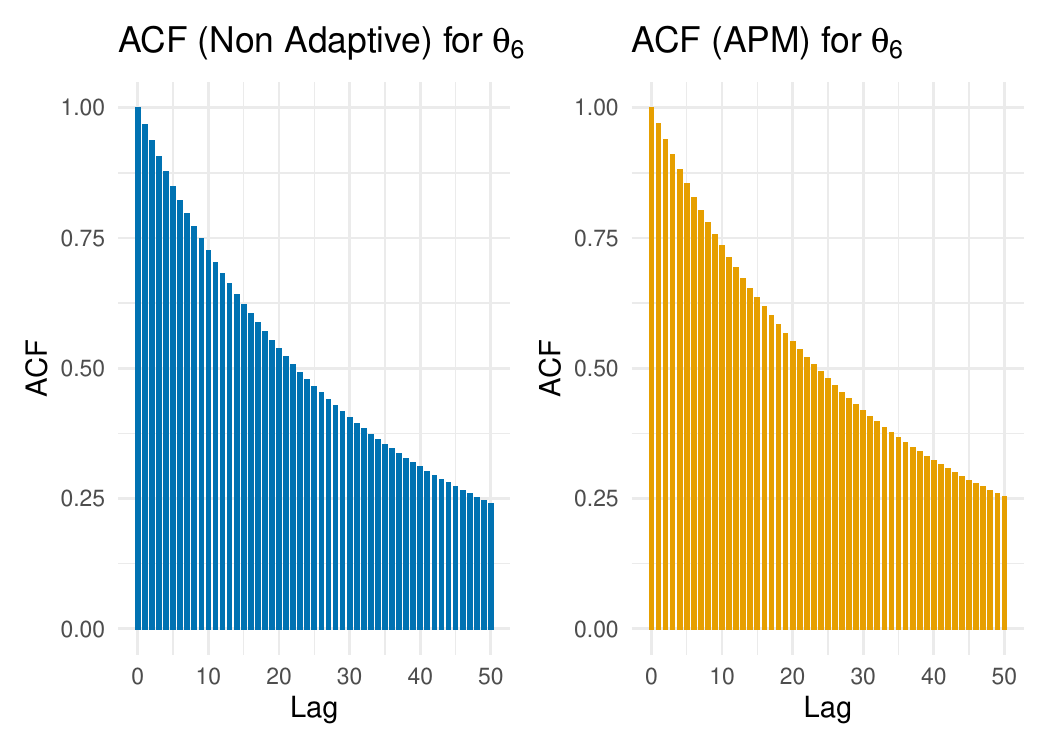}
  \end{minipage}

\end{figure}

\begin{figure}[H]
  \begin{minipage}{0.48\textwidth}
    \centering
    \includegraphics[width=\linewidth]{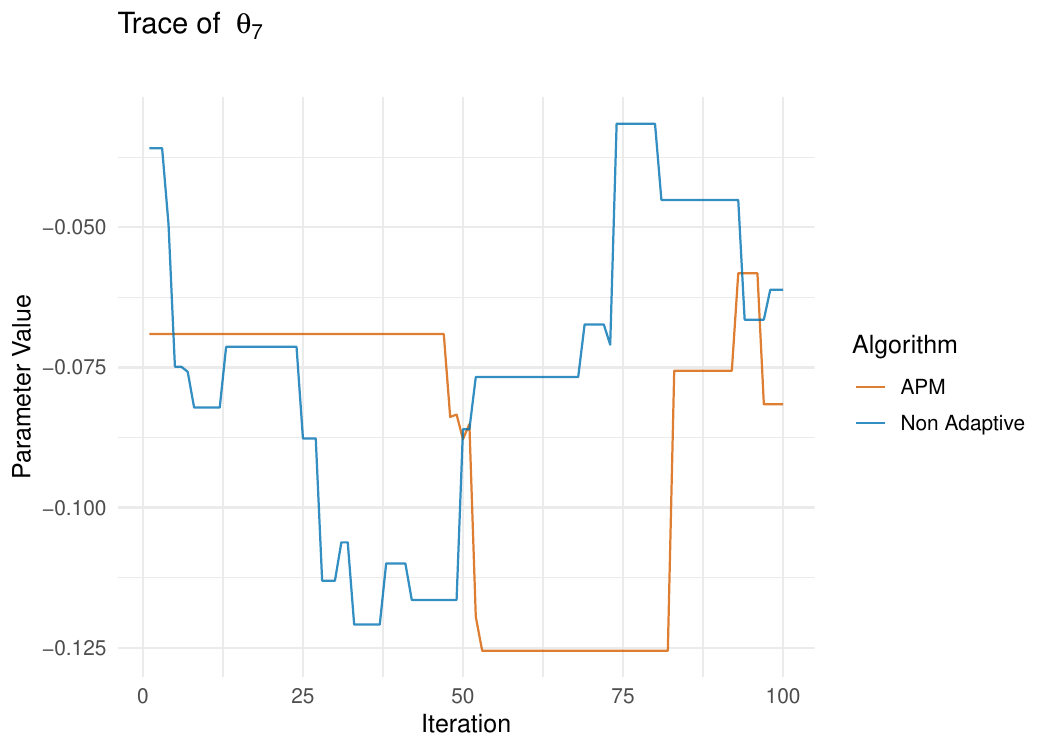}
  \end{minipage}
  \hfill
  \begin{minipage}{0.48\textwidth}
    \centering
    \vspace{0.3cm}
    \includegraphics[width=\linewidth]{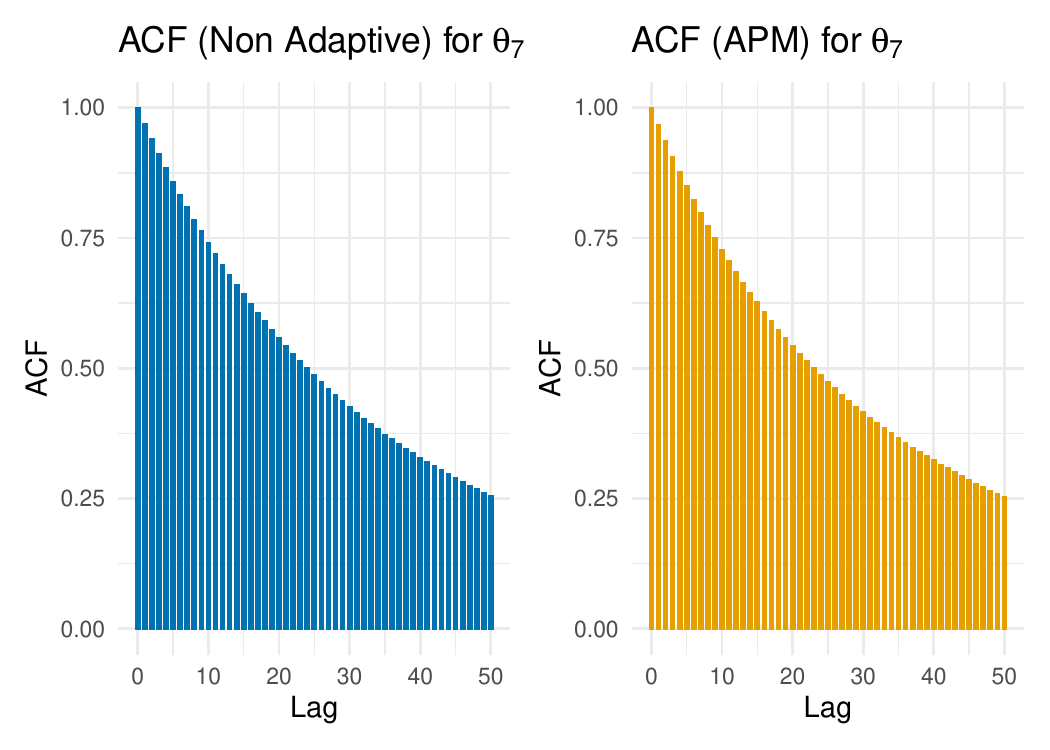}
  \end{minipage}

    \begin{minipage}{0.48\textwidth}
    \centering
    \includegraphics[width=\linewidth]{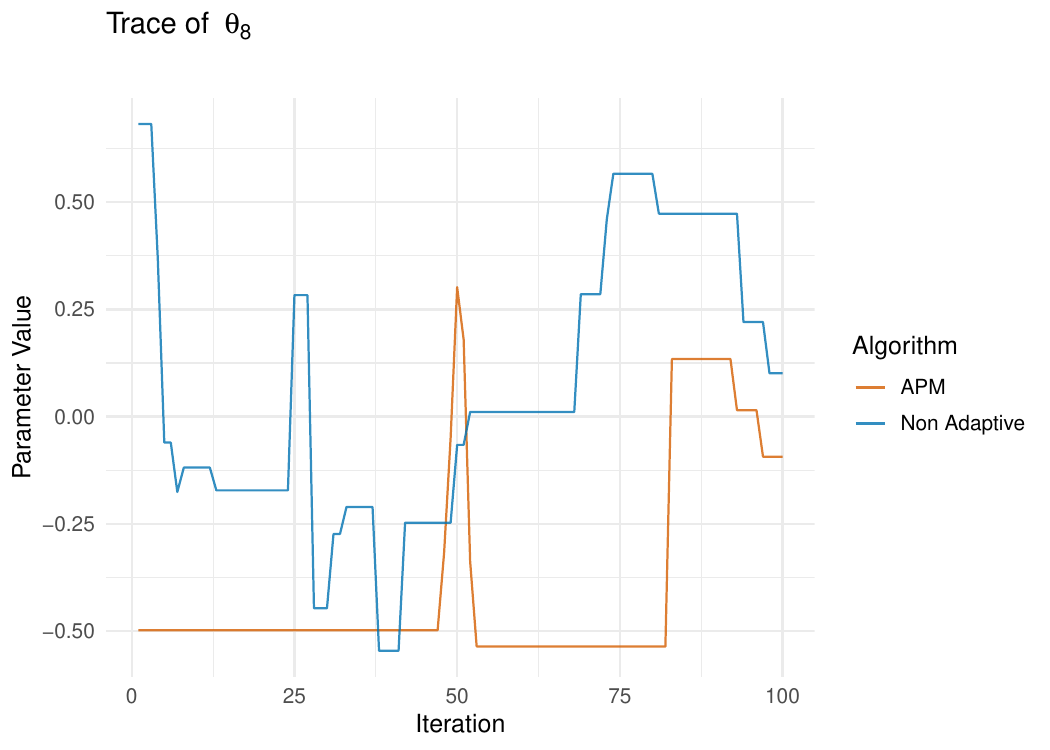}
  \end{minipage}
  \hfill
  \begin{minipage}{0.48\textwidth}
    \centering
    \vspace{0.3cm}
    \includegraphics[width=\linewidth]{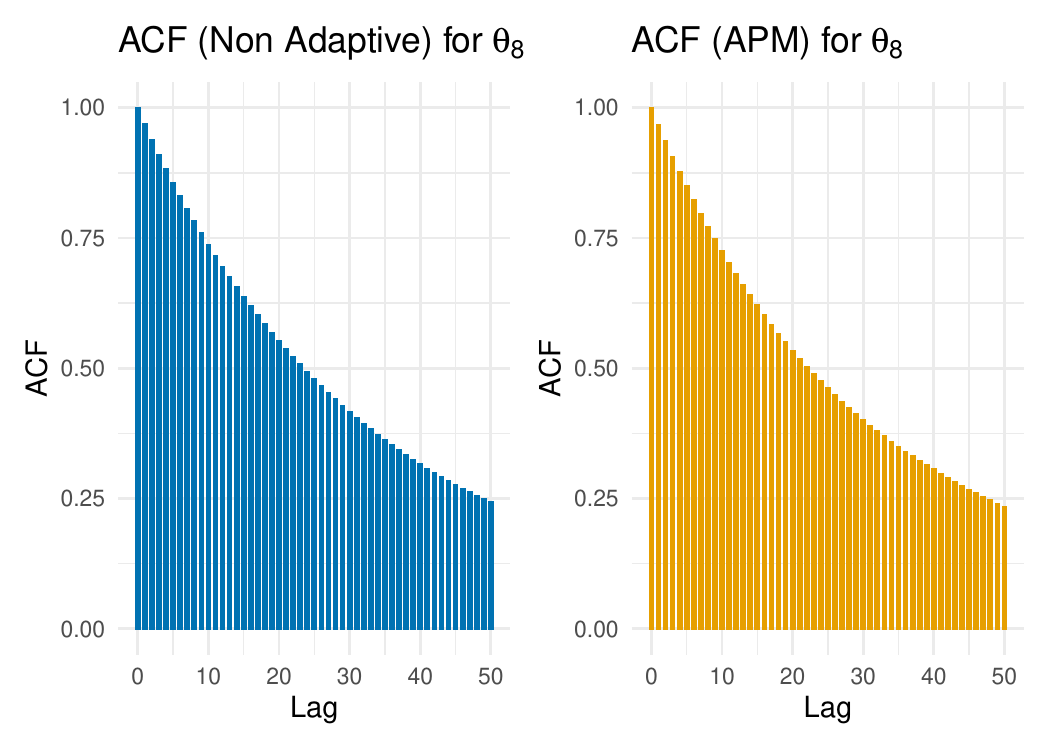}
  \end{minipage}
  
  \begin{minipage}{0.48\textwidth}
    \centering
    \includegraphics[width=\linewidth]{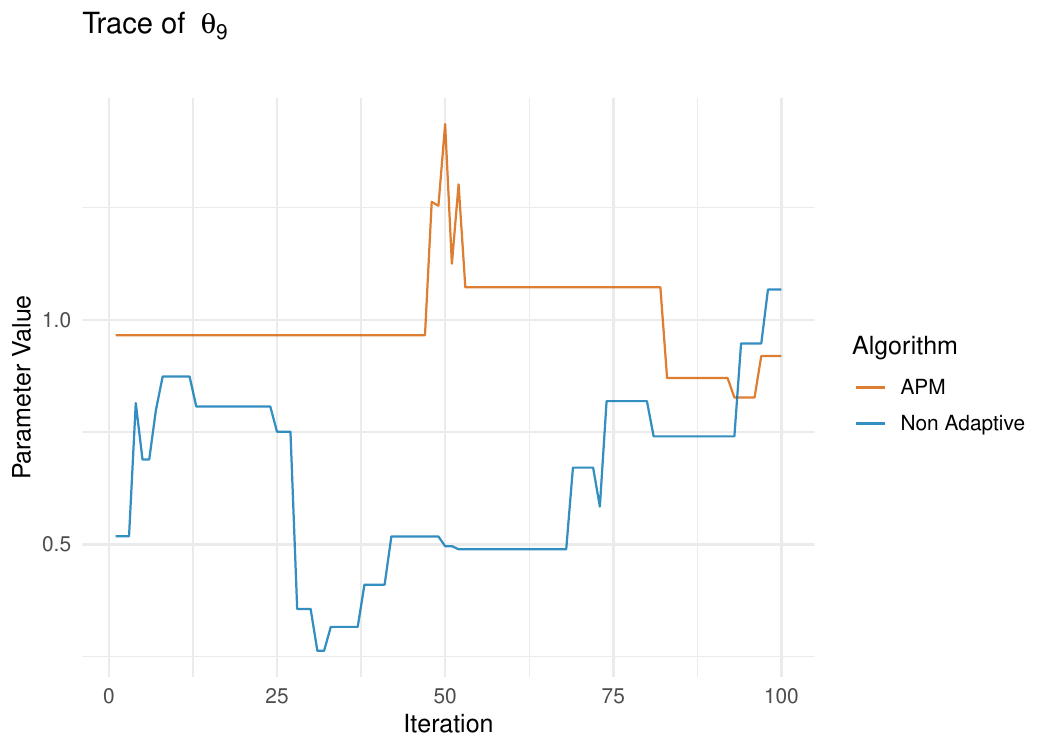}
  \end{minipage}
  \hfill
  \begin{minipage}{0.48\textwidth}
    \centering
    \vspace{0.3cm}
    \includegraphics[width=\linewidth]{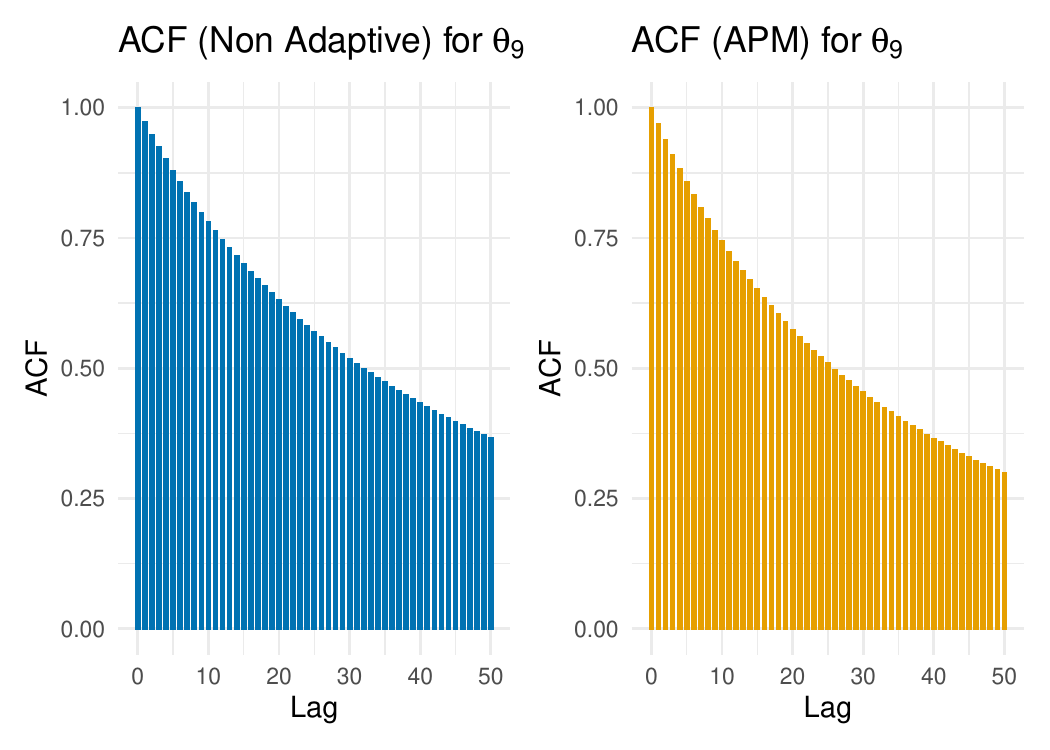}
  \end{minipage}
  \caption{Left: Trace plots of \( \theta \) over the last 100 iterations for a single run of the APM and non adaptive method with dichotomic search. Right: Autocorrelation functions for the same run.}
  \label{re:conv9}
\end{figure}

\end{document}